\documentclass[12pt]{article}

\RequirePackage{amssymb}
\RequirePackage{amsmath}
\RequirePackage{latexsym}
\RequirePackage{mathrsfs}

\DeclareMathOperator*{\cov}{cov}
\DeclareMathOperator*{\var}{var}

\newcommand{\R}{\ensuremath{\mathbb{R}}}
\newcommand{\Exp}{\ensuremath{\mathbb{E}}}

\usepackage{xr-hyper}
\pdfoutput=1
\usepackage{hyperref}
\hypersetup{
    colorlinks=true,
    linkcolor=black,
    citecolor=black,
    filecolor=black,
    urlcolor=black,
}

\usepackage[longnamesfirst]{natbib}
\usepackage{ulem}
\usepackage{color}

\usepackage{setspace}
\usepackage{pgfplots}
\usepackage{subfig}
\usepackage{epstopdf}

\usepackage{multirow}

\newcommand\sbullet[1][.5]{\mathbin{\vcenter{\hbox{\scalebox{#1}{$\bullet$}}}}}

\usepackage{tabularray-2021}
\UseTblrLibrary{booktabs}
\usepackage{threeparttable}

\usepackage{array}

\RequirePackage{amsthm}

\theoremstyle{definition}

\newtheorem{theorem}{Theorem}
\newtheorem{corollary}{Corollary}

\newtheorem{Aassumption}{Assumption}

\newcommand{\Lin}{\ensuremath{\mathbb{L}}}

\makeatletter
\newcommand*{\indep}{
  \mathbin{
    \mathpalette{\@indep}{}
  }
}
\newcommand*{\nindep}{
  \mathbin{
    \mathpalette{\@indep}{\not}
                               
  }
}
\newcommand*{\@indep}[2]{
  \sbox0{$#1\perp\m@th$}
  \sbox2{$#1=$}
  \sbox4{$#1\vcenter{}$}
  \rlap{\copy0}
  \dimen@=\dimexpr\ht2-\ht4-.2pt\relax
  \kern\dimen@
  {#2}
  \kern\dimen@
  \copy0 
} 
\makeatother

\def\bias{B}
 
\def\tildeB{\mathcal{B}_\text{A3fail}}

\usepackage{mathtools}

\usepackage{caption}
\newcommand{\GG}[1]{}

\title{\textbf{The Effect of Omitted Variables on \\ the Sign of Regression Coefficients}\footnote{We thank audiences at various seminars and conferences, as well as the editor, the referees, Paul Diegert, Rob Garlick, Arik Levinson, and Muyang Ren for helpful conversations and comments. We thank Paul Diegert, Jack Duhon, Hongchang Guo, Eszter Kiss, Julia Ma, Daria Soboleva, Jordan Woltjer, and Shuhan Zou for excellent research assistance. Masten thanks the National Science Foundation for research support under Grant 1943138.}}
\author{Matthew A. Masten\footnote{Department of Economics, Duke University,
        \texttt{matt.masten@duke.edu}} \qquad Alexandre Poirier\thanks{
    Department of Economics, Georgetown University,
 \texttt{alexandre.poirier@georgetown.edu}}
}
\date{May 8th, 2026}
\begin{document}
\maketitle
\begin{abstract}
We show that, depending on how the impact of omitted variables is measured, it can be substantially \emph{easier} for omitted variables to flip coefficient signs than to drive them to zero. This behavior occurs with ``Oster's delta'' (\citealt{Oster2019}), a widely reported robustness measure. Consequently, any time this measure is large---suggesting that omitted variables may be unimportant---a much smaller value reverses the sign of the parameter of interest. We propose a modified measure of robustness to address this concern. We illustrate our results in four empirical applications and two meta-analyses. We implement our methods in the companion Stata module \texttt{regsensitivity}.
\end{abstract}



\onehalfspacing
\normalsize
\section{Introduction}\label{sec:introduction}

The analysis of causality often relies on untestable assumptions, like the absence of unobserved omitted variables that could bias one's findings. A large literature in statistics and econometrics now provides many tools that allow researchers to perform sensitivity analyses to assess the importance of these assumptions. Many of these methods use \emph{breakdown points} as quantitative measures of the robustness of one's conclusions to departures from the baseline identifying assumptions. These points tell us how much we can depart from our identifying assumptions before our empirical conclusions break down in some way. 

In this paper, we distinguish between two specific kinds of breakdown points. The first is the \emph{explain away} breakdown point. It answers the question
\begin{itemize}
\item[] What is the smallest value of the sensitivity parameter required for the data to be consistent with a zero causal effect?
\end{itemize}
The second is the \emph{sign change} breakdown point. It answers the question
\begin{itemize}
\item[] What is the smallest value of the sensitivity parameter required for the data to be consistent with a causal effect that has a different sign from the causal effect we found in the baseline model?
\end{itemize}
These two breakdown points are often equal, but are not generally equivalent. In particular, the sign change breakdown point can be \emph{smaller} than the explain away breakdown point. Hence it can be easier to reverse the sign of one's results than to drive them to zero. This can occur when the omitted variable bias is discontinuous in the sensitivity parameter, allowing the value of the bias adjusted estimand to jump across the horizontal axis at zero as the sensitivity parameter varies. This behavior is illustrated in Figure \ref{fig:intro}, which we discuss more below. Such discontinuities can arise in regression analysis because the sensitivity parameters often involve covariance and variance terms, which lead to nonlinear restrictions on the value of the bias.

In this paper, we study the relationship between these two kinds of breakdown points in coefficient stability analyses used to assess the importance of omitted variables in linear regression analysis. These analyses examine how coefficients change when additional regressors are added, and have a long history in empirical economics. Folk wisdom holds that if the additional covariates increase the R-squared substantially, but the coefficient nonetheless does not change much, then we can be more confident that any further omitted variable bias is small. For example, see the textbook discussion on pages 74--78 of \cite{AngristPischke2015}. This idea was first formalized by \cite{AltonjiElderTaber2005a, AltonjiElderTaber2005b, AltonjiElderTaber2008}. It was later substantially extended by \cite{Oster2019} (hereafter Oster), an extremely influential paper with about 5500 Google Scholar citations as of December 2025. \citet[\emph{AER}, page 2706]{FinkelsteinEtAl2021} describe it as providing ``the now-standard methodology'' and ``the standard approach'' to adjusting for selection on unobservables.\footnote{Other papers on regression sensitivity include \cite{Mauro1990}, \cite{MurphyTopel1990}, \cite{Frank2000}, \cite{Imbens2003}, \cite{AltonjiElderTaber2005a}, \cite{Clarke2009}, \cite{BellowsMiguel2009}, \cite{HosmanHansenHolland2010}, \cite{GonzalezMiguel2015}, \cite{Krauth2016}, \cite{CinelliHazlett2020}, \cite{DMP2025ident, DMP2025axiom} for linear models and \cite{RosenbaumRubin1983sensitivity}, \cite{Rosenbaum1995, Rosenbaum2002}, \cite{RobinsRotnitzkyScharfstein2000}, \cite*{AltonjiElderTaber2005a, AltonjiElderTaber2008}, \cite{MastenPoirier2018, MastenPoirier2019BF}, and \cite{AET2019} for nonlinear models.} Moreover, from 2019--2021, 25 papers published in the top five economics journals formally compute and report the results from this method, while another 6 papers informally reference this method in support of their analysis. This implies that, in an average year, more papers use this method than two-stage least squares (based on the data collected by \citealt{BlandholEtAl2022}). We discuss in detail how Oster's \citeyearpar{Oster2019} results are used in empirical economics in Appendix \ref{sec:inPractice}. 

In Section \ref{sec:OsterReview} we review Oster's \citeyearpar{Oster2019} method. The main sensitivity parameter for this method, denoted by $\delta$, is commonly interpreted as the ratio of the magnitude of selection on unobservables to the magnitude of selection on observables. We then study breakdown points for this sensitivity parameter in Section \ref{sec:OsterSignFlip}. Oster's Proposition 2 characterizes the explain away breakdown point. Following her recommendation, empirical researchers now routinely report estimates of this explain away breakdown point as a measure of robustness to omitted variables. We show that this explain away breakdown point is very different from the sign change breakdown point. In fact, we show that any time the explain away breakdown point is large---suggesting that omitted variables may be unimportant---a much smaller value of her sensitivity parameter can actually change the sign of the parameter of interest. In particular, we prove that the sign change breakdown point is bounded above by $1$. This is a substantial concern since 1 is often viewed as the cutoff for a robust result. Consequently, if we maintain 1 as the cutoff, our result implies that \emph{no empirical results are robust to sign changes}, using Oster's method. Put differently: The true causal effect can have the opposite sign as the baseline estimand any time the magnitude of selection on unobservables is at least as large as the magnitude of selection on observables.

\begin{figure}[t]
\centering
\includegraphics[width=0.475\textwidth]{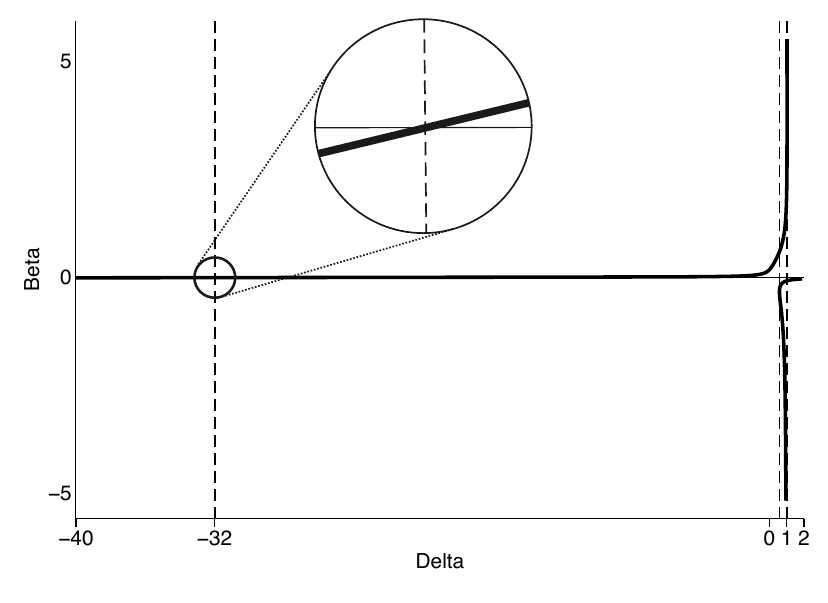}
\includegraphics[width=0.475\textwidth]{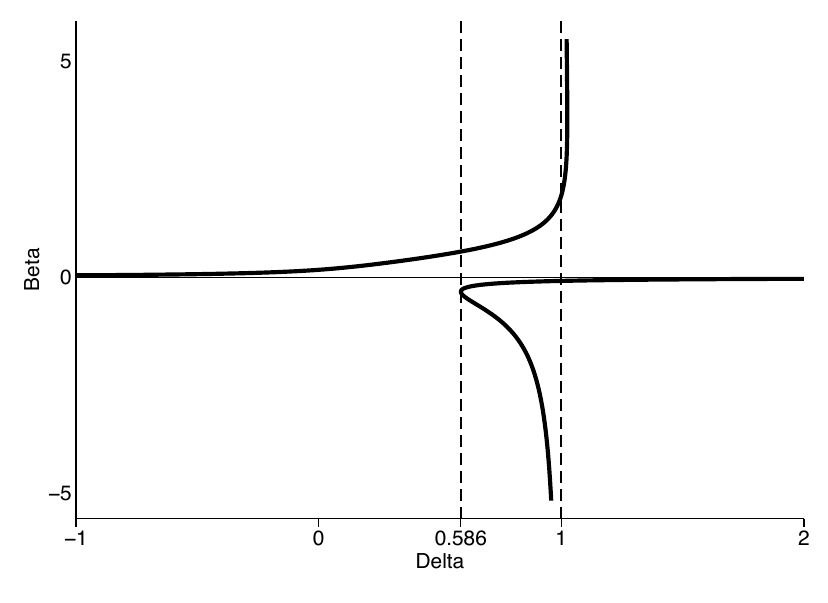}
\caption{\footnotesize Example of Oster's identified set $\mathcal{B}_I(\delta, R_\text{long}^2)$ for the regression coefficient $\beta$ as a function of the sensitivity parameter $\delta$ (see Section \ref{sec:OsterReview} for a definition of this set and a description of the $R_\text{long}^2$ parameter). As we discuss in Section \ref{sec:osterIdentification}, this identified set always has at most three elements. Both plots show the same set, but with different horizontal axis ranges. In this example, the explain away breakdown point is $| -32 |$. This is shown in the magnified region on the left plot. In contrast, the sign change breakdown point is 0.586, as shown in the right plot. The plot also shows the vertical asymptote at $\delta = 1$ which we discuss in Section \ref{sec:deltaSignChangeBP}. This figure is based on our empirical application to \cite{SatyanathVoigtlanderVoth2017}; see Section \ref{sec:empirical}.\label{fig:intro}}
\end{figure}

To build intuition for our main results in Section \ref{sec:OsterSignFlip}, consider Figure \ref{fig:intro}, which is based on our empirical application to \citet*{SatyanathVoigtlanderVoth2017}. Both plots show the same curve, but with different horizontal axis ranges, to highlight different aspects which we discuss below. This curve is an example from our empirical application in Section \ref{sec:empirical}. It shows the pairs of regression coefficient values $\beta$ (on the vertical axis) and Oster's main sensitivity parameter $\delta$ (on the horizontal axis) that are consistent with the data and assumptions. Importantly, as shown in the right plot, for some values of $\delta$ there are multiple values of $\beta$ that are consistent with the data. This occurs because $\delta$ is defined to be a ratio of two regression coefficients, and hence knowledge of $\delta$'s exact value implies that $\beta$ must satisfy a fairly complicated nonlinear constraint. Specifically, it must be the solution to a cubic equation (see equation \eqref{eq:sol_to_cubic_eqn} on page \pageref{eq:sol_to_cubic_eqn}), which can have multiple roots. This explains the source of the discontinuity in the value of the omitted variable bias that we discussed earlier, and why for this sensitivity parameter it is easier to flip the coefficient's sign than to obtain an exact zero.

Concretely, the baseline estimate is obtained at $\delta = 0$. In this empirical example, the baseline estimate is positive. How robust is this estimate to the presence of omitted variables? In the left plot we see that we must go all the way to $\delta = -32$ before the value $\beta = 0$ is consistent with the data and assumptions. This is the number that researchers commonly report as a measure of robustness. However, from the right plot, we see that we only need to go to $\delta = 0.586$ to find a value of $\beta$ that is negative. This value $|\delta| = 0.586$ is the sign change breakdown point. \cite{Oster2019} does not discuss this breakdown point, its companion Stata package \texttt{psacalc} does not compute it, and none of the papers we survey in Appendix \ref{sec:inPractice} report it. Yet, as we see here, the sign change breakdown point is substantially smaller than $| -32 |$, the explain away breakdown point. Indeed, its magnitude is below $1$, the commonly used cutoff for a robust result. Consequently, in this empirical example, focusing on the sign change breakdown point overturns the authors' conclusion of robustness: Only a small amount of selection on unobservables relative to observables is necessary to flip the sign of the baseline estimate.

Importantly, \emph{all} of the coefficient values plotted in Figure \ref{fig:intro}, even the seemingly extreme values, are equally consistent with the data and assumptions. Under the assumptions of Oster's Proposition 2, which is used to obtain these plots, there is no justification for ignoring any of these values. The only way to eliminate values is by making additional assumptions. One reasonable assumption, which we consider in Section \ref{sec:modified}, is to restrict the magnitude of the omitted variable bias. With this additional assumption, we can justifiably eliminate some of the most extreme values. This assumption does not overturn our main conclusions, however. For example, in Figure \ref{fig:intro} the coefficient value consistent with $\delta = 0.586$ is not extreme at all---it is the same order of magnitude as the baseline estimate. So in this example, the sign change breakdown point is still very different from the explain away breakdown point, even with a substantial restriction on the magnitude of the omitted variable bias. Moreover, magnitude restrictions require researchers to choose a tuning parameter---just how much should they restrict the magnitude? If the magnitude is restricted too much, then this assumption is essentially equivalent to assuming a priori that there is no omitted variable bias, in which case a sensitivity analysis would be unnecessary. But as soon as reasonably sized magnitudes are allowed, the concerns we discuss in sections \ref{sec:OsterSignFlip} and \ref{sec:biasAdjustment} can arise.

\begin{figure}[t]
\centering
\raisebox{-0.5\height}{\includegraphics[width=0.35\textwidth]{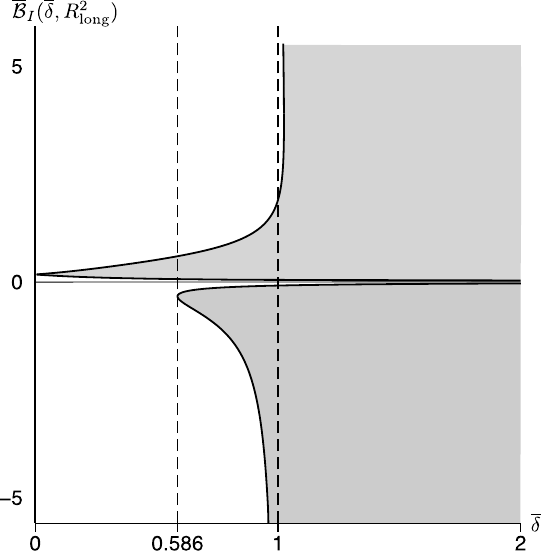}}
\hspace{0.05\textwidth}
\raisebox{-0.5\height}{\includegraphics[width=0.48\textwidth]{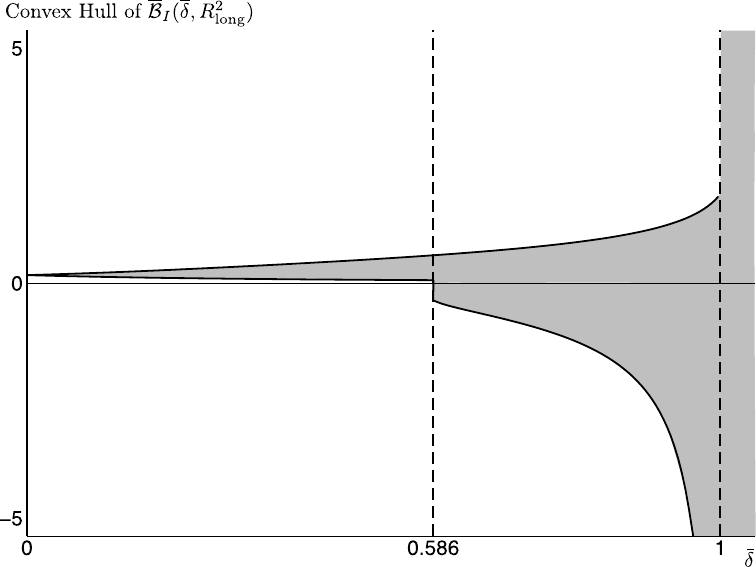}}
\caption{\footnotesize Left: The cumulative identified set for the example in Figure \ref{fig:intro}. Here we see that for large values of $\bar{\delta}$ the identified set includes every value in $\R$ except a small interval around 0. The set includes 0 for $\bar{\delta} \geq 32$, the explain away breakdown point. But the sign change breakdown point occurs much earlier, at $\delta = 0.586$. Moreover, note that the identified set at $\bar{\delta} = 1$ is $(-\infty, -0.0855] \cup [0.0432,1.8947]$. At $\bar{\delta} = 1.05$ it is $(-\infty,-0.0798] \cup [0.0415,\infty)$. Right: The convex hull of the left plot. This figure is based on our empirical application to \cite{SatyanathVoigtlanderVoth2017}; see Section \ref{sec:empirical}. \label{fig:osterCumulative}}
\end{figure}

Researchers also commonly report selection bias adjusted estimates. In section \ref{sec:biasAdjustment} we show that the most commonly used adjustment is extremely sensitive to perturbations of Oster's sensitivity parameter. To see why, consider again Figure \ref{fig:intro}. There is always an asymptote at $\delta = 1$. This implies that bias adjustments computed at $\delta$ exactly equal to one---which is currently standard empirical practice---are \emph{not} representative of bias adjustments computed at nearby values of $\delta$. 

To address the issues described in sections \ref{sec:OsterSignFlip} and \ref{sec:biasAdjustment}, we propose several simple modifications to the methods in \cite{Oster2019}. First, we recommend gathering all values of the coefficient consistent with a range of sensitivity parameters, and presenting this set. The left plot in Figure \ref{fig:osterCumulative} gives an example: Here we gather all $\beta$'s consistent with $| \delta |$ less than some fixed $\bar{\delta}$. We then plot that set as a function of $\bar{\delta}$. The graph of these values of $\beta$ is discontinuous in $\bar{\delta}$ at $\bar{\delta} = 0.586$, which explains the difference in the sign change breakdown and explain away breakdown points. As we saw in Figure \ref{fig:intro}, nearly all values of $\beta$, including arbitrarily large negative values, are consistent with an upper bound of $\bar{\delta} = 1.05$, but a small range around 0 is excluded from the set, and the value 0 does not become part of this set until $\bar{\delta}$ equals 32. Researchers who want to focus only on the largest and smallest possible values of $\beta$ can instead present the convex hull, which is shown in the right plot of Figure \ref{fig:osterCumulative}.

Second, we develop the sign change breakdown point for Oster's main sensitivity parameter and recommend that researchers report this rather than, or in addition to, the explain away breakdown point. Furthermore, in Section \ref{sec:modified}, we provide extensions to these two recommendations to formally incorporate the magnitude restrictions discussed above. We show that the sign change breakdown point can be strictly larger than 1 when a magnitude restriction is imposed. Therefore it is possible to recover findings of robustness even without changing the cutoff for robustness. This is not always the case however, as we saw in Figure \ref{fig:intro}.

\nocite{SatyanathVoigtlanderVoth2017data}
\nocite{SatyanathVoigtlanderVoth2017}

In Section \ref{sec:empirical} we illustrate our results using data from Satyanath et al.\ (2017a,b, \emph{Journal of Political Economy}), who study the effect of social capital on the rise of the Nazi party. There we show just how different the sign change and explain away breakdown points can be in practice. We then show how magnitude restrictions can sometimes help retain findings of robustness. We also illustrate the extreme sensitivity of bias adjusted estimates to perturbations of the sensitivity parameter. In Appendix \ref{sec:extraEmpirical_applications} we study three additional empirical applications in detail, all published in the \emph{American Economic Review}.

To see that these four applications are not outliers, in Section \ref{sec:metaAnalysis} we conduct two meta-analyses covering 58 different empirical applications published in the top five economics journals plus the \emph{AEJ: Applied Economics}. We show that the distinction between explain away and sign change breakdown points leads to substantively different conclusions about robustness to omitted variables in the majority of empirical applications. In particular, suppose researchers were to continue the current standard practice of reporting explain away breakdown points, but wanted to interpret them as sign change breakdown points. We show that in more than 50\% of applications, this interpretation is only valid if researchers \emph{already} know the sign of the coefficient of interest before performing the sensitivity analysis. But in that case, it would not be necessary to do a sensitivity analysis, since the result is known a priori.

Thus far we have discussed ways of assessing omitted variable bias by using $\delta$ as the main sensitivity parameter. We also briefly discuss several critiques of $\delta$ itself in Section \ref{sec:delta_discussion}. Finally, in Section \ref{sec:conclusion} we conclude with a summary of our recommendations.

\section{Setup and Identification Analysis}\label{sec:OsterReview}

In this section we review the main identification results in \cite{Oster2019}, including her breakdown analysis. Throughout this section, as well as sections \ref{sec:OsterSignFlip}--\ref{sec:delta_discussion}, we abstract from sampling uncertainty and assume the population data is known. We briefly discuss estimation and inference in Appendix \ref{sec:estimation}.

\subsection{Regressions of Interest}

Let $Y$ be an outcome variable, $X$ a treatment variable, $W_1$ a vector of observed covariates, and $W_2$ a vector of unobserved covariates. The following assumption ensures that the OLS estimands we will consider are well defined.

\begin{Aassumption}\label{assn:posdef}
$\var(Y,X,W_1,W_2)$ is finite. $\var(X,W_1,W_2)$ and $\var(Y,X,W_1)$ are positive definite.
\end{Aassumption}

\ref{assn:posdef} allows $Y$ to be perfectly collinear with $(X,W_1,W_2)$, but not with $(X,W_1)$. Collinearity of $Y$ with $(X,W_1)$ would violate \ref{assn:nonzero} below. 

Consider three OLS regressions:
\begin{enumerate}
\item The \textbf{short regression} of
\[
	Y \text{ on } (1,X).
\]
Let $\beta_\text{short}$ denote the corresponding population regression coefficient on $X$. Let $R_\text{short}^2$ denote the corresponding R-squared.

\item The \textbf{medium regression} of
\[
	Y \text{ on } (1,X,W_1).
\]
Let $\beta_\text{med}$ denote the corresponding population regression coefficient on $X$. Let $R_\text{med}^2$ denote the corresponding R-squared. Let $\gamma_{1,\text{med}}$ denote the coefficients on $W_1$.

\item The \textbf{long regression} of
\[
	Y \text{ on } (1,X,W_1,W_2).
\]
Let $\beta_\text{long}$ denote the corresponding population regression coefficient on $X$. Let $R_\text{long}^2$ denote the corresponding R-squared. Note that $R_\text{long}^2 \in [R_\text{med}^2,1]$. Let $(\gamma_{1,\text{long}}, \gamma_{2,\text{long}})$ denote the coefficients on $W \coloneqq (W_1, W_2)$.
\end{enumerate}
It is well known that linear regression coefficients can be given causal interpretations under various identifying assumptions (e.g., \citealt{AngristPischke2009}). In this paper, we take it as given that $\beta_\text{long}$ is the parameter of interest. This focus on $\beta_\text{long}$ is often because we have in fact made assumptions that make it a causal parameter, but our technical results do not rely on those motivating assumptions. So, taking interest in $\beta_\text{long}$ as given, write the long regression as
\begin{equation}\label{eq:proj1}
	Y = \beta_\text{long} X + \gamma_{1,\text{long}}' W_1 + \gamma_{2,\text{long}}' W_2 + Y^{\perp X,W}
\end{equation}
where $Y^{\perp X,W}$ is defined to be the OLS residual plus the intercept term, and hence is uncorrelated with each component of $(X,W_1,W_2)$ by construction. The identification problem is that $W_2$ is not observed. Instead, we only assume that the joint distribution of $(Y,X,W_1)$ is known. This allows us to point identify $\beta_\text{short}$ and $\beta_\text{med}$. These coefficients usually do not equal $\beta_\text{long}$, however. The difference between $\beta_\text{long}$ and either $\beta_\text{short}$ or $\beta_\text{med}$ is omitted variable bias. We next consider assumptions that will help us bound the magnitude of this omitted variable bias.

\subsection{Comparing Selection on Observables and on Unobservables}\label{sec:defOfDelta}

Following the analysis of \citet[pages 175--176]{AltonjiElderTaber2005a}, \cite{Oster2019} recommends that we measure the magnitude of selection on unobservables via the parameter
\begin{equation}\label{eq:deltaMainDefinition}
	\delta \coloneqq \frac{ \cov(X, \gamma_{2,\text{long}}' W_2) }{ \var(\gamma_{2,\text{long}}'W_2) } \Bigg/ \frac{\cov(X,\gamma_{1,\text{long}}'W_1) }{\var(\gamma_{1,\text{long}}'W_1)}.
\end{equation}
To allow for cases where the denominator is zero, our formal analysis uses the following assumption.

\begin{Aassumption}\label{assn:delta}
There is a known $\delta \in \R$ such that
\begin{equation}\label{eq:mainDeltaEquation}
	\delta \cdot \frac{\cov(X,\gamma_{1,\text{long}}'W_1) }{\var(\gamma_{1,\text{long}}'W_1)} = \frac{ \cov(X, \gamma_{2,\text{long}}' W_2) }{ \var(\gamma_{2,\text{long}}' W_2) }.
\end{equation}
\end{Aassumption}

The following assumption, combined with \ref{assn:posdef}, ensures that we are not dividing by zero in equation \eqref{eq:mainDeltaEquation}.

\begin{Aassumption}\label{assn:nonzero}
	$\gamma_{1,\text{long}} \neq 0$ and $\gamma_{2,\text{long}} \neq 0$.
\end{Aassumption}

The parameter $\delta$ depends on two terms:
\begin{enumerate}
\item (``Selection on unobservables'') Regress $X$ on $(1,\gamma_{2,\text{long}}'W_2)$ and get the coefficient on the index $\gamma_{2,\text{long}}'W_2$. 

\item (``Selection on observables'') Regress $X$ on $(1,\gamma_{1,\text{long}}'W_1)$ and get the coefficient on the index $\gamma_{1,\text{long}}'W_1$. 
\end{enumerate}
$\delta$ is the ratio of these two regression coefficients. $\delta$ is not known from the data since it depends on $W_2$, which is not observed, and $(\gamma_{1,\text{long}},\gamma_{2,\text{long}})$, which are also unknown. Instead, we will study what can be said about $\beta_\text{long}$ under assumptions on $\delta$. Several papers propose alternative sensitivity parameters, including \cite{CinelliHazlett2020} and \cite{DMP2025ident}. In this paper we instead use the same sensitivity parameter $\delta$ as \cite{Oster2019}, since our focus here is to study the implications of using different kinds of breakdown points. We give more discussion of $\delta$'s interpretation in Section \ref{sec:delta_discussion} below.

\subsubsection*{Baseline covariates}

$\delta$ is defined to calibrate the magnitude of selection on unobservables by using the observed covariates $W_1$. Typically we will have additional covariates $W_0$ which we want to include in the analysis, but \emph{not} use for calibration. Denote these baseline controls by $W_0$. To incorporate these covariates, replace $(X,W_1,W_2)$ with $(X^{\perp W_0}, W_1^{\perp W_0}, W_2^{\perp W_0})$, the variables after residualizing with respect to $W_0$, throughout the entire analysis. See Section 3.4 of \cite{DMP2025ident} for a justification of this approach and further discussion. For simplicity, we omit $W_0$ throughout our theoretical analysis.

\subsection{Identification of $\beta_\text{long}$}\label{sec:osterIdentification}

In this section we state a version of the main and most general identification result in \cite{Oster2019}. That paper also analyzes several special cases which we do not discuss for brevity. The main substantive assumption is the following.

\begin{Aassumption}[Exogenous controls]\label{assn:exog}
	$\cov(W_1,W_2) = 0$.
\end{Aassumption}

\ref{assn:exog} requires that all of the observed covariates are uncorrelated with the omitted variable. Appendix A.1 in \cite{Oster2019} briefly discusses one approach to relaxing \ref{assn:exog}, which we discuss in Section \ref{sec:delta_discussion}. Also see Appendix C of \cite{DMP2025axiom} for additional details. \cite{DMP2025ident} develop an alternative approach to sensitivity analysis that does not rely on \ref{assn:exog}. In this paper we maintain \ref{assn:exog} throughout because our main focus is to study different notions of breakdown. For simplicity, from here on we assume $W_2$ is a scalar.  All of our results can be generalized to vector $W_2$ by working with the scalar index $\gamma_{2,\text{long}}'W_2$ everywhere (as in the proof of Theorem 2 in \citealt{DMP2025ident}, for example).

To characterize the possible values for $\beta_\text{long}$, it is useful to consider the regression of $X$ on $(1,W_1)$. Let $\pi_1$ denote the coefficient on $W_1$ and let $X^{\perp W_1} \coloneqq X - \pi_1'W_1$.

Under the above assumptions, Oster's Proposition 2 shows that $\beta_\text{long}$ must lie in the set
\begin{align}\label{eq:sol_to_cubic_eqn}
	\mathcal{B}(\delta, R_\text{long}^2) \coloneqq \{ b \in \R: f(\beta_\text{med} - b,\delta,R_\text{long}^2) = 0 \},
\end{align}
where $f(\bias,\delta, R_\text{long}^2) \coloneqq f_0(\bias) + \delta \cdot f_1(\bias, R_\text{long}^2)$ with
\begin{align*}
	f_0(\bias) &\coloneqq -\bias\var(X^{\perp W_1})\left(\var(\gamma_{1,\text{med}}'W_1) + 2\bias\cov(X,\gamma_{1,\text{med}}'W_1) + \bias^2\var(\pi_1'W_1)\right) \\[1em]
	f_1(\bias,R_\text{long}^2)
	&\coloneqq (R_\text{long}^2 - R_\text{med}^2)\var(Y) \cov(X, \gamma_{1,\text{med}}' W_1)
	+ \bias (R_\text{long}^2 - R_\text{med}^2) \var(Y) \var(\pi_1'W_1) \\
	&\qquad + \bias^2 \var(X^{\perp W_1})\cov(X,\gamma_{1,\text{med}}'W_1)
	+ \bias^3\var(X^{\perp W_1})\var(\pi_1'W_1).
\end{align*}
As \cite{Oster2019} discusses, the set in \eqref{eq:sol_to_cubic_eqn} are the roots of a cubic polynomial. Hence this set has at most three elements. Figure \ref{fig:intro} shows an example of this set, as a function of $\delta$. Our Theorem \ref{thm:IDset} in Appendix \ref{sec:technical} extends this analysis to show that $\mathcal{B}(\delta, R_\text{long}^2)$ is essentially sharp. That is, almost any element in this set is in fact a feasible value of $\beta_\text{long}$. Formally, we show that it equals the identified set (denoted by $\mathcal{B}_I(\delta,R_\text{long}^2)$), after subtracting off the set $\tildeB \coloneqq \{b \in \R: \gamma_{1,\text{med}} + (\beta_\text{med} - b)\pi_1 = 0\}$:
\begin{align}\label{eq:ID_set_expression}
	\mathcal{B}_I(\delta,R_\text{long}^2) = \mathcal{B}(\delta,R_\text{long}^2) \setminus \tildeB.
\end{align}
The set we subtract off are values of $\beta_\text{long}$ that are not consistent with \ref{assn:nonzero}; this set is typically empty, however. This sharpness is important to ensure that breakdown analysis is not unnecessarily conservative. Finally, in Appendix \ref{sec:technical} we also show how the analysis simplifies to the baseline case when $\delta = 0$, which gives $\beta_\text{long} = \beta_\text{med}$.

\subsection{Oster's Breakdown Analysis}

Thus far we have characterized the identified set for $\beta_\text{long}$ when $\delta$ and $R_\text{long}^2$ are known. Next we consider breakdown analysis. In this section we review the main breakdown point result in \cite{Oster2019}. We then compare it with the sign change breakdown point in Section \ref{sec:OsterSignFlip}.

Let $\beta_\text{hypo} \in \R$. Define the \emph{exact value} breakdown point
\[
	\delta^{\text{bp},=}(\beta_\text{hypo},R_\text{long}^2)
	\coloneqq \inf \{ | \delta | : \delta \in \R, \beta_\text{hypo} \in \mathcal{B}_I(\delta, R_\text{long}^2) \}.
\]
This is the smallest magnitude of the sensitivity parameter $\delta$ that is compatible with $\beta_\text{long} = \beta_\text{hypo}$. That is, for all $\delta$ with $| \delta | < \delta^{\text{bp},=}(\beta_\text{hypo},R_\text{long}^2)$, we know that $\beta_\text{hypo}$ is not in the identified set $\mathcal{B}_I(\delta, R_\text{long}^2)$. The next result gives a closed form solution for this value.

\begin{theorem}\label{thm:OsterProp3}
Suppose the joint distribution of $(Y,X,W_1)$ is known, \ref{assn:posdef}, \ref{assn:nonzero}, and \ref{assn:exog} hold, $R_\text{long}^2 \in (R_\text{med}^2,1]$ is known, $\beta_\text{long} = \beta_\text{hypo}$ for a known $\beta_\text{hypo} \in \R \setminus \tildeB$, and $\beta_\text{short} \neq \beta_\text{med}$. Then the solution $\delta$ to equation \eqref{eq:mainDeltaEquation} is point identified and equals
\[
	\frac{-f_0(\beta_\text{med} - \beta_\text{hypo})}{f_1(\beta_\text{med} - \beta_\text{hypo}, R_\text{long}^2)}.
\]
Consequently, the absolute value of this term equals $\delta^{\text{bp},=}(\beta_\text{hypo},R_\text{long}^2)$.
\end{theorem}

The proof of Theorem \ref{thm:OsterProp3} and all other formal results can be found in Appendix \ref{sec:proofs}. Theorem \ref{thm:OsterProp3} is very similar to Oster's Proposition 3, with only minor differences. As Oster remarked, this result shows that, even though knowledge of $\delta$ does not point identify $\beta_\text{long}$, knowledge of $\beta_\text{long}$ does point identify $\delta$. Figure \ref{fig:intro} illustrates why: Fixing $\delta$ is equivalent to placing a vertical line on the plot. This vertical line may intersect the displayed curve up to three times, since $\mathcal{B}_I(\delta,R_\text{long}^2)$ can have up to three elements. Fixing the value $\beta_\text{long}$ to be equal to some known $\beta_\text{hypo}$, however, is equivalent to placing a horizontal line on the plot. This line can only intersect the displayed curve once. The absolute value of the horizontal coordinate at this intersection is precisely $\delta^{\text{bp},=}(\beta_\text{hypo}, R_\text{long}^2)$.

Setting $\beta_\text{hypo} = 0$, we get the \emph{explain away} breakdown point
\[
	\delta^{\text{bp},\text{explain away}}(R_\text{long}^2)
	\coloneqq
	\delta^{\text{bp},=}(0,R_\text{long}^2)
	=
	\frac{|f_0(\beta_\text{med})|}{|f_1(\beta_\text{med}, R_\text{long}^2)|}.
\]
Its signed value, $-f_0(\beta_\text{med})/f_1(\beta_\text{med} , R_\text{long}^2)$, is often called \emph{Oster's delta}. This is the unique value of $\delta$ that is compatible with an exactly zero long regression coefficient, $\beta_\text{long} = 0$. Estimates of this value are commonly reported as a measure of the robustness of the baseline model $(\delta = 0)$ to the presence of omitted variables. 

\section{Measuring the Robustness of the Sign of $\beta_\text{long}$ to Omitted Variables}\label{sec:OsterSignFlip}

In this section we analyze the robustness of conclusions about the sign of $\beta_\text{long}$ to the presence of omitted variables. In Section \ref{sec:deltaSignChangeBP} we first show that the sign change breakdown point for $\delta$ can never be larger than 1. This is a concern since $| \delta | = 1$ is often considered the cutoff for robustness, and therefore this result implies that conclusions about the sign of $\beta_\text{long}$ are always non-robust. In Section \ref{sec:modified} we propose a simple modification to Oster's analysis that allows the sign change breakdown point to be larger than 1.

\subsection{The Sign Change Breakdown Point for $\delta$}\label{sec:deltaSignChangeBP}

Oster's delta answers the following question:
\begin{itemize}
\item[] Suppose we knew the exact value of $R_\text{long}^2$. What is the smallest value of $|\delta|$ that is required for the data to be consistent with $\beta_\text{long} = 0$? 
\end{itemize}
We call this the \emph{explain away} breakdown point. Researchers may also be interested in answering the following question:
\begin{itemize}
\item[] Suppose we knew the exact value of $R_\text{long}^2$. What is the smallest value of $|\delta|$ such that $\beta_\text{long}$ has a different sign from $\beta_\text{med}$?
\end{itemize}
We call this the \emph{sign change} breakdown point. As mentioned earlier, the answers to these two questions can be different. In fact, as we next show, the sign change breakdown point is typically much smaller than the explain away breakdown point.

To formalize this claim, define the sign change breakdown point for a given value of $R_\text{long}^2$ as
\[
	\delta^{\text{bp},\text{sign}}(R_\text{long}^2) \coloneqq 
	\begin{cases}
		\delta^{\text{bp},>}(R_\text{long}^2)
			&\text{if $\beta_\text{med} > 0$} \\
		\delta^{\text{bp},<}(R_\text{long}^2)
			&\text{if $\beta_\text{med} < 0$}
	\end{cases}
\]
where
\begin{align*}
	\delta^{\text{bp},>}(R_\text{long}^2)
	&\coloneqq
	\inf \{ | \delta | : \delta \in \R, b \leq 0 \text{ for some } b \in \mathcal{B}_I(\delta, R_\text{long}^2) \}\\[0.5em]
	\delta^{\text{bp},<}(R_\text{long}^2)
	&\coloneqq \inf \{ | \delta | : \delta \in \R, b \geq 0 \text{ for some } b \in \mathcal{B}_I(\delta, R_\text{long}^2) \}.
\end{align*}

We now state our main theoretical result.

\begin{theorem}\label{thm:osterExplainAwayBounded}
Suppose \ref{assn:posdef}--\ref{assn:exog} hold. Suppose $\beta_\text{short} \neq \beta_\text{med}$. Then for any $R_\text{long}^2 \in (R_\text{med}^2, 1]$,
\[
	\delta^{\text{bp},\text{sign}}(R_\text{long}^2) \leq 1.
\]
\end{theorem}

Here we assume $\beta_\text{short} \neq \beta_\text{med}$ for simplicity only. Theorem \ref{thm:osterExplainAwayBounded} shows that, in Oster's sensitivity analysis, the sign change breakdown point \emph{can never be larger than one}. In practice, researchers typically consider the value $|\delta| = 1$ to be the cutoff for determining the robustness of a result. That is, they interpret values of $|\delta|$ greater than 1 as robust and values less than 1 as sensitive. Based on this criterion, Theorem \ref{thm:osterExplainAwayBounded} therefore shows that the \emph{sign} of the parameter of interest is \emph{never robust} (assuming we require $|\delta|$ strictly larger than 1 to conclude that the result is robust). In contrast, the explain away breakdown point $\delta^{\text{bp},\text{explain away}}(R_\text{long}^2)$ is often substantially larger than one. We illustrate this difference empirically in Section \ref{sec:empirical}.

Figure \ref{fig:intro} illustrates this result: The graph $\{ (\delta,b) \in \R^2 : f(\beta_\text{med} - b, \delta, R_\text{long}^2) = 0 \}$ has an asymptote at $\delta = 1$. Consequently, values of $\delta$ close to one---but not equal to one---yield identified sets that contain arbitrarily large or arbitrarily small values. Hence for $\delta$ values near one we cannot point identify the sign of $\beta_\text{long}$. However, much larger values of $\delta$ may be needed before the graph crosses the horizontal axis $b = 0$ exactly. This explains why the explain away breakdown point can often be much larger than the sign change breakdown point for this model.

\subsubsection*{Intuition}

The asymptote at $\delta = 1$ arises because the sensitivity analysis allows the treatment $X$ and the covariates $(W_1,W_2)$ to be nearly multi-collinear. And, as we explain here, this near multi-collinearity implies that the omitted variable bias will be large while at the same time $\delta$ will be close to 1. We prove this formally in Proposition \ref{prop:formal_asym_intuition} in Appendix \ref{sec:intuitionAppendix}. To see it intuitively, consider a distribution of $(Y,X,W_1,W_2)$ that is consistent with the identified covariance matrix $\var(Y,X_1,W_1)$ but where the omitted variable $W_2$ almost fully accounts for the variation in treatment after adjusting for the observed variables $W_1$. There are three main steps:
\begin{enumerate}
\item In this case, $X$, $W_1$, and $W_2$ are close to multi-collinear. Consequently, the coefficients in the long regression, $(\beta_\text{long}, \gamma_{1,\text{long}}, \gamma_{2,\text{long}})$, will be large. We prove that this is a general feature of near multi-collinear OLS estimands in Appendix \ref{sec:OLSmultiCollinear}. If $\beta_\text{long}$ is large then the OVB must be large because $\beta_\text{med}$ is fixed (that is, its value is not affected by the omitted variable $W_2$).

\item $\delta$'s measure of selection on unobservables is the coefficient from regressing $X$ on $\gamma_{2,\text{long}} W_2$:
\begin{align}\label{eq:SOU_coeff}
	\frac{\cov(X,\gamma_{2,\text{long}}W_2)}{\var(\gamma_{2,\text{long}}W_2)}.
\end{align}
$\gamma_{2,\text{long}}$, as one of the long regression coefficients, will also be large because of near multi-collinearity. Consequently, this measure of selection on unobservables will be close to zero since regression coefficients scale down when the covariate is scaled up. More formally, note that $\cov(X,\gamma_{2,\text{long}}W_2)$ and $\var(\gamma_{2,\text{long}}W_2)$ both diverge to infinity as we approach perfect multi-collinearity since $|\gamma_{2,\text{long}}| \rightarrow \infty$ in this scenario. However their ratio (the regression coefficient in \eqref{eq:SOU_coeff}) converges to zero since its numerator is linear in $\gamma_{2,\text{long}}$ while its denominator is quadratic in $\gamma_{2,\text{long}}$.

\item For the same reason, $\delta$'s measure of selection on observables will also be close to zero, since $\gamma_{1,\text{long}}$ will be large. Since the numerator and denominator of $\delta$ are both small, it is not obvious how the ratio will behave. However, it turns out that both terms converge to zero at the same rate, which implies that $\delta$ converges to a constant. This follows from the OVB formula for the coefficient on $W_1$:
\[
	\gamma_{1,\text{long}} - \gamma_{1,\text{med}} = (\beta_\text{med} - \beta_\text{long}) C_1
\]
where $C_1$ is a constant, and the OVB formula for the coefficient on $W_2$:
\[
	\gamma_{2,\text{long}} - 0 =  (\beta_\text{med} - \beta_\text{long}) C_2
\]
where $C_2$ is a value that converges to a constant. We derive these equations in Lemma \ref{lemma:regfacts}. Thus both $\gamma_{2,\text{long}}$ and $\gamma_{1,\text{long}}$ diverge at the same rate, controlled by the divergence of $\beta_\text{long}$. The divergence of these terms is what drives the numerator and denominator of $\delta$ to zero. Finally, some additional regression algebra shows that the constant limit of $\delta$ is one.
\end{enumerate}
Thus in the first step we see that the OVB will be large while in the third step we see that $\delta$ will be close to one. This is precisely what happens along the asymptote at $\delta = 1$. This followed from (a) analyzing the behavior of near multi-collinear OLS coefficients and (b) exploring the specific structure of $\delta$. Thus the asymptote arises because $(X,W_1,W_2)$ are allowed to be arbitrarily mutli-collinear. Researchers may be interested in explicitly ruling out this collinearity. This is exactly what the alternative methods of \cite{CinelliHazlett2020} and \cite{DMP2025ident} do. Both papers include sensitivity parameters that explicitly bound the joint explanatory power of the observed and unobserved covariates in the first stage. The sensitivity analyses proposed by both papers then lead to sign change and explain away breakdown points that are equal.

\subsection{Incorporating Magnitude Restrictions on the Omitted Variable Bias}\label{sec:modified}

We have just shown that the sign change breakdown point for $\delta$ can never be larger than one. This negative result can be overturned by adding additional identifying assumptions. For example, there are usually values of the coefficient that are extremely implausible, such as a 500\% return to education. Any procedure that produced such an estimate would therefore suggest that there is a problem with the procedure itself, the data, or both. In this section we show that if researchers are willing to explicitly impose this kind of bound then we can obtain sign change breakdown points larger than one. Specifically, consider the following assumption.

\begin{Aassumption}\label{assn:magnRest}
$| \beta_\text{long} - \beta_\text{med} | \leq M$ for some known $M \geq 0$.
\end{Aassumption}

\ref{assn:magnRest} says that the magnitude of the omitted variable bias is no larger than $M$. This assumption can be viewed as formalizing Oster's \citeyearpar{Oster2019} statement that researchers may be ``willing to assume that the bias is fairly small'' (page 194). The value $M$ should be chosen to reflect the researcher's beliefs about the largest possible bias in their baseline estimand $\beta_\text{med}$. For example, consider a regression of log earnings $Y$ on years of schooling $X$ and other observed covariates $W_1$. Suppose our estimate of $\beta_\text{med}$ is 0.1. Then $M=1$ would say that the true coefficient on schooling---as obtained after adjusting for the unobserved confounders---lies between $-0.9$ and 1.1. This interval includes most estimates (e.g., \citealt{Card2001} Table II) so in this application researchers would likely view \ref{assn:magnRest} with $M=1$ as a plausible assumption.

In general, a simple approach to selecting $M$ is to let it be a multiple of $\beta_\text{med}$. That is, let $M = p \cdot | \beta_\text{med} |$ for $p \geq 0$. For this choice, \ref{assn:magnRest} is equivalent to
\[
	\beta_\text{long} \in \big[ \beta_\text{med} - p \cdot|\beta_\text{med}|, \beta_\text{med} + p\cdot|\beta_\text{med}| \big].
\]
All elements of this set have the same sign if $p < 1$. For this reason, since we are interested in robustness to sign changes, $p$ should be chosen to be larger than 1. In the returns to schooling example with $\beta_\text{med} = 0.1$, $M=1$ is equivalent to the choice $p=10$. We recommend that researchers consider a range of magnitude restrictions, as we illustrate empirically in Section \ref{sec:empirical}. For example, in Section \ref{sec:empirical} we consider $p = 2$, $3$, and $10$. 

The next result shows the impact of this magnitude restriction on identification of $\beta_\text{long}$.

\begin{corollary}\label{cor:magnitudeRestrictions}
Suppose the joint distribution of $(Y,X,W_1)$ is known. Suppose \ref{assn:posdef}--\ref{assn:magnRest} hold. Suppose $R_\text{long}^2 \in (R_\text{med}^2,1]$ is known. Recall that \ref{assn:delta} specifies that $\delta \in \R$ is known. Denote the identified set for $\beta_\text{long}$ under these assumptions by $\mathcal{B}_I(\delta, R_\text{long}^2, M)$. Then
\[
	\mathcal{B}_I(\delta, R_\text{long}^2, M) = \mathcal{B}_I(\delta, R_\text{long}^2) \cap [\beta_\text{med} - M, \beta_\text{med} + M].
\]
\end{corollary}

Let $\delta^{\text{bp},\text{sign}}(R_\text{long}^2, M)$ denote the corresponding sign change breakdown point. For $0 \leq M < \infty$, it is now possible for $\delta^{\text{bp},\text{sign}}(R_\text{long}^2, M) > 1$. We illustrate this empirically in Section \ref{sec:empirical}. 

Thus far we have followed \cite{Oster2019} by assuming the exact value of $\delta$, the measure of selection is known. We extend previous results by considering a known \textit{bound} on the measure of selection, rather than its exact value. Specifically, we make the following assumption.

\begin{Aassumption}\label{assump:deltaBarBound}
There is a known constant $\bar{\delta} \geq 0$ such that $| \delta | \leq \bar{\delta}$ for $\delta$ that satisfies equation \eqref{eq:mainDeltaEquation}.
\end{Aassumption}

This leads to the following identification result.

\begin{corollary}\label{cor:cum_IDset}
Suppose the joint distribution of $(Y,X,W_1)$ is known. Suppose \ref{assn:posdef}, \ref{assn:nonzero}, \ref{assn:exog}, and \ref{assump:deltaBarBound} hold. Suppose $R_\text{long}^2 \in (R_\text{med}^2,1]$ is known. Denote the identified set for $\beta_\text{long}$ under these assumptions by $\overline{\mathcal{B}}_I(\bar{\delta}, R_\text{long}^2)$. Then
\[
	\overline{\mathcal{B}}_I(\bar{\delta}, R_\text{long}^2) = \bigcup_{\delta : | \delta | \leq \bar{\delta}} \mathcal{B}_I(\delta,R_\text{long}^2).
\]
Suppose \ref{assn:magnRest} also holds. Denote the identified set for $\beta_\text{long}$ under this as well as the earlier assumptions by $\overline{\mathcal{B}}_I(\bar{\delta}, R_\text{long}^2, M)$. Then
\[
	\overline{\mathcal{B}}_I(\bar{\delta}, R_\text{long}^2, M) = \bigcup_{\delta : | \delta | \leq \bar{\delta}} \mathcal{B}_I(\delta,R_\text{long}^2, M).
\]
\end{corollary}

This corollary simply notes that we take the union of the identified sets from equation \eqref{eq:ID_set_expression} over all allowed values of the sensitivity parameter $\delta$ with magnitude less than $\bar{\delta}$ to obtain the identified set under \ref{assump:deltaBarBound}. 

\section{The Sensitivity of Bias Adjusted Estimands}\label{sec:biasAdjustment}

Thus far we have focused on breakdown analysis. Researchers also commonly compute point estimates that adjust for omitted variable bias. See our survey in Appendix \ref{sec:inPractice} for examples from the empirical economics literature. The most commonly used bias adjustment is an estimate of
\begin{equation}\label{eq:OsterP1}
	\beta_\text{long} =
	\beta_\text{med} + (\beta_\text{med} - \beta_\text{short}) \frac{R_\text{long}^2 - R_\text{med}^2}{R_\text{med}^2 - R_\text{short}^2}.
\end{equation}
Oster's Proposition 1 shows that this equation holds if $\delta = 1$ and an auxiliary assumption holds (see \ref{assump:OsterA2} in Appendix \ref{sec:BiasCorrectionsAppendix}). However, equation \eqref{eq:OsterP1} does not hold, even in an approximate sense, if $\delta$ is very close to one, but not exactly equal to one. Instead, $\beta_\text{long}$ can be arbitrarily large even for values of $\delta$ very close to one. As discussed earlier, this is due to the vertical asymptote in the identified set that occurs precisely at $\delta = 1$. Consequently, bias adjustments based on equation \eqref{eq:OsterP1} do not accurately represent the magnitude of omitted variable bias for values of $\delta$ close to one. We illustrate this discrepancy in Table \ref{table:biasAdjustmentSensitivity} in our empirical application. We discuss this issue in more detail in Appendix \ref{sec:BiasCorrectionsAppendix}. We also discuss several other issues with how bias adjustment is done in practice in Appendix \ref{sec:commonPitfallsAppendix}.

Rather than performing bias adjustments, researchers can plot the identified set $\overline{\mathcal{B}}_I(\bar{\delta}, R_\text{long}^2)$ for $\beta_\text{long}$ as a function of $\bar{\delta} \geq 0$. We described this set in Section \ref{sec:modified}. This set describes the values of $\beta_\text{long}$ consistent with any amount of selection on unobservables $\delta$ with magnitude less than $\bar{\delta}$. Hence, unlike a single point estimand for $\delta = 1$, it shows how the feasible values of $\beta_\text{long}$ vary for a range of $\delta$ values. See Figure \ref{fig:osterCumulative} for an example.

\section{Interpreting $\delta$ and Additional Critiques}\label{sec:delta_discussion}

We have focused on the properties of various ways of assessing omitted variable bias that use $\delta$ as the main sensitivity parameter. The previous literature has also critiqued the use of $\delta$ itself as a sensitivity parameter. We briefly discuss those and several other critiques in this section.

\subsection{$\delta$ is a Double Ratio}\label{sec:doubleRatioDelta}

\citet[Section 6.3]{CinelliHazlett2020} show that $\delta$ can be interpreted as a double ratio, and use this result to argue that $\delta$ does not accurately measure selection on unobservables versus selection on observables. We summarize their analysis here. For simplicity, consider the case where $W_1$ and $W_2$ are scalar and are normalized to have variance 1. Maintain assumptions \ref{assn:posdef}--\ref{assn:exog} from earlier. Let
\begin{align}\label{eq:proj_eqn}
	X &= \pi_1 W_1 + \pi_2 W_2 + X^{\perp W},
\end{align}
where $(\pi_1,\pi_2)$ are the coefficients on $(W_1,W_2)$ in the linear projection of $X$ on $(1,W_1,W_2)$, and $X^{\perp W}$ is the intercept plus the projection residual. Note that $\pi_1$ is also the coefficient on $W_1$ in the linear projection of $X$ on $(1,W_1)$, since we assumed $W_1$ and $W_2$ are uncorrelated.

Substituting equation \eqref{eq:proj_eqn} into the definition of $\delta$ (equation \eqref{eq:deltaMainDefinition}) shows that $\delta$ can be written as a double ratio:
\begin{align*}
	\delta &= \frac{ \cov(X, \gamma_{2,\text{long}} W_2) }{ \var(\gamma_{2,\text{long}} W_2) } \Bigg/ \frac{\cov(X,\gamma_{1,\text{long}} W_1) }{\var(\gamma_{1,\text{long}} W_1)} = \frac{\pi_2}{\pi_1} \Big/ \frac{\gamma_{2,\text{long}}}{\gamma_{1,\text{long}}}.
\end{align*}
The numerator measures the relative strength of two associations: \emph{treatment} and the omitted variable ($\pi_2$), and treatment and the observed variable ($\pi_1$). Similarly, the denominator measures the relative strength of two associations: the \emph{outcome} and the omitted variable ($\gamma_{2,\text{long}}$), and the outcome and the observed variable ($\gamma_{1,\text{long}}$). Hence $\delta$ compares the relative importance of unobservables and observables for treatment to their relative importance for outcomes.

This property that $\delta$ is a ratio of ratios leads to some counter-intuitive interpretations of the magnitude of $\delta$. For example, suppose $\gamma_{2,\text{long}} = c \cdot \gamma_{1,\text{long}}$ and $\pi_2 = c \cdot \pi_1$ for $c = 10$. Since both $W_1$ and $W_2$ are normalized to have unit variance, this value of $c$ indicates that $W_2$ has an association with $X$ (and $Y$) that is 10 times stronger than that of $W_1$ with $X$ (and $Y$). However, the implied value of $\delta$ is 1, which is traditionally interpreted as meaning that there is ``equal selection'' based on $W_2$ and on $W_1$. This value of $\delta$ is unchanged if $c = 1/10$, in which case $W_2$'s association with $X$ and $Y$ is 10 times \textit{weaker} than $W_1$'s association with those same variables. In fact, the value of $\delta$ is equal to 1 for all $c \neq 0$ in this example. As \cite{CinelliHazlett2020} point out, the interpretation of $\delta$ in some papers is at odds with its double ratio property. For example \citet[page 192]{Oster2019} states that $\delta = 1$ implies that ``the unobservable[s] and observables are equally related to the treatment.'' Their paper recommends continuing to use $R_\text{long}^2$ as one of the sensitivity parameters, but that $\delta$ should be replaced with $R^2_{X \sim W_2 \sbullet W_1}$, where $R^2_{X \sim W_2 \sbullet W_1}$ is the unknown R-squared in a regression of $X$ on $W_2$, after partialling out $W_1$. Moreover, the sensitivity analysis that uses this alternative parameterization has identical sign change and explain away breakdown points.

\subsection{Exogeneity of the Controls}

All of the sensitivity analyses we discussed in sections \ref{sec:OsterReview}--\ref{sec:biasAdjustment} maintain Assumption \ref{assn:exog}. That assumption requires all observed variables to be uncorrelated with all omitted variables. In many applications, this assumption would be viewed as strong---most observed variables are themselves correlated and it is unlikely that the omitted variables happen to be uncorrelated with all of the observed covariates, especially when many observed covariates are included in the analysis. Oster argues this assumption can be relaxed by redefining $W_2$ to be $W_2^{\perp W_1} := W_2 - \cov(W_2,W_1)\var(W_1)^{-1}W_1$, which is uncorrelated with $W_1$ by definition. However, as \cite{DMP2025axiom} argue, this redefinition of $W_2$ substantively changes the interpretation of $\delta$. The interpretation changes because the replacement of $W_2$ by $W_2^{\perp W_1}$ requires replacing the coefficients $(\gamma_{1,\text{long}},\gamma_{2,\text{long}})$ with $(\widetilde{\gamma}_1,\gamma_{2,\text{long}})$ where $\widetilde{\gamma}_1 := \gamma_{1,\text{long}} + \var(W_1)^{-1}\cov(W_1,W_2)\gamma_{2,\text{long}}$. The redefined $\delta$ is thus not only a function of the associations between $(W_1,W_2)$ and $X$ and $Y$, as discussed above, but also of the unknown covariance between $W_2$ and $W_1$. Concretely, for scalar $W_1$ and $W_2$, the redefined $\delta$ is
\begin{align*}
	\widetilde{\delta}
	&\coloneqq \frac{\cov(X,\gamma_{2,\text{long}} W_2^{\perp W_1})}{\var(\gamma_{2,\text{long}}W_2^{\perp W_1})} \Bigg/ \frac{\cov(X,\widetilde{\gamma}_1 W_1)}{\var(\widetilde{\gamma}_1 W_1)}\\
	&= \frac{\pi_2}{\pi_1  + \pi_2 \frac{\cov(W_2,W_1)}{\var(W_1)}} \Bigg/ \frac{\gamma_{2,\text{long}}}{\gamma_{1,\text{long}}  + \gamma_{2,\text{long}} \frac{\cov(W_2,W_1)}{\var(W_1)}}.
\end{align*}
This approach implies that applied researchers have to either (a) assume unobserved confounders are uncorrelated with all controls, or (b) reason about a parameter that depends on this unknown covariance. (a) is typically considered implausible, while (b) is challenging since this redefined $\delta$ is difficult to interpret.

This concern can be addressed by using an alternative parameterization of the omitted variable bias formula, such as those in \cite{CinelliHazlett2020} or \cite{DMP2025ident}. The parameterization in \cite{CinelliHazlett2020} continues to use $R_\text{long}^2$ but replaces $\delta$ with $R_{X \sim W_2 \sbullet W_1}^2$ (see Section \ref{sec:doubleRatioDelta} above). The parameterization in \cite{DMP2025ident} replaces both $\delta$ and $R_\text{long}^2$ with $r_X^2 \coloneqq \var(\pi_2'W_2)/\var(\pi_1'W_1)$ and $r_Y^2 \coloneqq \var(\gamma_{2,\text{long}}'W_2)/\var(\gamma_{1,\text{long}}'W_1)$, which are measures of the relative importance of $W_2$ and $W_1$ in the selection and outcome equation, respectively. In addition to addressing this concern about endogenous controls, both parameterizations also lead to sensitivity analyses where the sign change and explain away breakdown points are the same.

\subsection{The Role of $R_\text{long}^2$}\label{sec:R2longChoice}

All of the analysis above assumes the exact value of $R_\text{long}^2$ is known. Researchers should also assess how their conclusions depend on the choice of $R_\text{long}^2$. This is not common practice, however; see our survey analysis in Appendix \ref{sec:inPractice}. Instead, researchers commonly focus on one of two choices, either $1.3 \widehat{R}_\text{med}^2$ or 1. The choice $R_\text{long}^2 = 1$ is often thought to be the most conservative. However, \citet[propositions 6--8]{Basu2022} shows that this is not correct. Specifically, he shows that the explain away breakdown point $\delta^{\text{bp},=}(\beta_\text{hypo},R_\text{long}^2)$ is \emph{not} necessarily decreasing in $R_\text{long}^2$. Consequently, according to this measure of robustness, \emph{smaller} values of $R_\text{long}^2$ could lead to smaller values of the breakdown point, and hence less robust results. This fact can be overlooked if researchers only use a single value of $R_\text{long}^2$. For researchers who want to use the sensitivity parameters $\delta$ and $R_\text{long}^2$ we recommend plotting estimates of the identified sets $\overline{\mathcal{B}}_I(\bar{\delta}, R_\text{long}^2)$ for a range of $R_\text{long}^2$ values. Although we omit this analysis for brevity, researchers could also consider replacing the known $R_\text{long}^2$ assumption with an assumption that it is only known to be bounded, $R_\text{long}^2 \in [R_\text{med}^2, \overline{R}_\text{long}^2]$ for a known value of $\overline{R}_\text{long}^2$, and then reporting the identified set for $\beta_\text{long}$ as a function of $\overline{R}_\text{long}^2$ and the upper bound $\overline{\delta}$ on $\delta$.

\subsection{Optimistic Bias Correction}

Finally, \citet[page 220]{DeLucaMagnusPeracchi2019} note that, when there are multiple elements of \eqref{eq:sol_to_cubic_eqn} (values of $\beta_\text{long}$ consistent with the data and assumptions), by default the Stata package \texttt{psacalc} only reports the element that is \emph{closest} to $\beta_\text{med}$. Since the other elements are equally consistent with the data and assumptions, they note that there is no a priori reason to believe that the true magnitude of omitted variable bias happens to be the smallest one. This concern is also addressed by plotting the set $\overline{\mathcal{B}}_I(\bar{\delta}, R_\text{long}^2)$ rather than performing bias corrections.

\def\mystrut{\rule{0pt}{1.25\normalbaselineskip}}
\section{Empirical Application: Social Capital and the Rise of the Nazi Party}\label{sec:empirical}

Social capital is typically viewed as a good which leads to positive outcomes for countries. \citet*[\emph{JPE}]{SatyanathVoigtlanderVoth2017} instead argued that ``social capital can undermine and help to destroy a democratic system'' (page 482) by studying the rise of the Nazi Party (the National Socialist German Worker's Party [NSDAP]) in 1933. To do this, they analyzed a dataset of $n=229$ towns/cities in Germany. The outcome variable $Y$ is a measure of how many people in that city joined the Nazi party. The treatment variable $X$ is a measure of social capital, based on ``association density,'' the number of social clubs and associations per 1,000 city inhabitants. They consider three variations of this measure based on the type of association included: All associations, Civic only, and Military only. They also include a variety of socioeconomic and political covariates.
 
Their paper has a variety of analyses. We focus on their main results, which use a selection on observables identification strategy paired with linear models. That is, they estimate linear regressions of $Y$ on $X$, an intercept, and additional observed covariates. To assess the robustness of those results to omitted variables, they applied \cite{Oster2019} in their Subsection 4.D, with additional details in their Appendix G. We re-examine that analysis in light of our theoretical results in Section \ref{sec:OsterSignFlip}. See Appendix \ref{sec:extraEmpirical_details} for extra details and results for this empirical application.

\subsubsection*{Breakdown Analysis}

\begin{table}[!t]
\centering
\SetTblrInner[talltblr]{rowsep=0pt}
\resizebox{0.8\textwidth}{!}{
\footnotesize
\begin{talltblr}[
  caption = {Different Types of Breakdown Points for Assessing the Robustness of \cite{SatyanathVoigtlanderVoth2017} Results to Omitted Variables.\label{table:mainTable1}},
  remark{Note} = {For all specifications, the $W_0$ includes the baseline controls while the calibration covariates $W_1$ includes the socioeconomic and political controls. See Appendix \ref{sec:extraEmpirical_details} for a complete list of these variables. For comparison with column (1), in column (2) we also report the sign of the estimate obtained from Theorem \ref{thm:OsterProp3}. Columns (1) and (2) show $\widehat{\delta}^{\text{bp},\text{explain away}}(R_\text{long}^2)$. Column (3) shows the sign change breakdown point $\widehat{\delta}^{\text{bp},\text{sign}}(R_\text{long}^2, +\infty) = \widehat{\delta}^{\text{bp},\text{sign}}(R_\text{long}^2)$. Columns (4)--(6) show $\widehat{\delta}^{\text{bp},\text{sign}}(R_\text{long}^2, M)$ for various values of $M$.
  },
]{c | c *{6}{>{\centering}p{0.075\textwidth}}}
  \toprule
\mystrut
 & & \multicolumn{2}{c}{Explain Away} & \multicolumn{4}{c}{Sign Change with $M$ equal to}  \\[0.4em]
& & Reported & Correct & $+ \infty$ & $10 | \widehat{\beta}_\text{med} |$ & $3 | \widehat{\beta}_\text{med} | $ & $2 | \widehat{\beta}_\text{med} |$ \\[4pt]
  & $R_\text{long}^2$ & (1) & (2) & (3) & (4) & (5)  & (6)\\[4pt]
\hline
\mystrut
\multirow{2}{*}{All} & $1$ & $<0$ & $-32.0$ & 0.586 & 0.586 & 0.586 & 0.676  \\
 & $1.3 \widehat{R}_\text{med}^2$ & - & $-331.5$ & 0.953 & 0.954 & 1.476 & 2.829   \\[0.4em] 
\hline
\mystrut
\multirow{2}{*}{Civic} &  $1$ & $<0$ & 0.99 & 0.736 & 0.736 & 0.736 & 0.745 \\
 & $1.3 \widehat{R}_\text{med}^2$ & - & 9.78 & 0.994 &1.036  & 1.806 & 2.986  \\[0.4em] 
\hline
\mystrut
\multirow{2}{*}{Military} &  $1$ & $<0$ & 1.42 & 0.899 & 0.899 & 0.909 & 0.977  \\
 & $1.3 \widehat{R}_\text{med}^2$& - & 10.65 & 0.993 & 1.035 & 1.821 & 3.058    \\[2pt] 
\bottomrule
\end{talltblr}
}
\end{table}

Table \ref{table:mainTable1} shows the main breakdown results. This table replicates and extends the bottom half of Panel B in Table A.27 of Appendix G in \cite{SatyanathVoigtlanderVoth2017}. The three panels of Table \ref{table:mainTable1} correspond to the three different versions of treatment. All of the cell entries are different kinds of breakdown points. First consider columns (1) and (2). These show estimates of Oster's explain away breakdown point. We present estimates for two different choices of $R_\text{long}^2$, either 1 or Oster's rule of thumb $1.3 \widehat{R}_\text{med}^2$. Column (1) shows the values reported by \cite{SatyanathVoigtlanderVoth2017}. As described in Appendix \ref{sec:commonMistake}, these values are incorrect. Column (2) reports the correct values. If we judged robustness by comparing the magnitude of these values of the explain away breakdown point to 1, then all of the results would be deemed robust except one, civic associations with $R_\text{long}^2 = 1$, which gives an explain away breakdown point of 0.98.

However, as we have discussed in this paper, the explain away breakdown point can be very different from the sign change breakdown point. Columns (3)--(6) show estimated sign change breakdown points, for four increasingly restrictive choices of the constraint $M$ on the magnitude of OVB: $+ \infty$ (no constraint), $10 | \widehat{\beta}_\text{med} |$, $3 | \widehat{\beta}_\text{med} |$, and $2 | \widehat{\beta}_\text{med} |$. First consider column (3), which does not impose any magnitude restriction. As implied by Theorem \ref{thm:osterExplainAwayBounded}, we see that all of the point estimates are bounded above by 1. 

As explained in Section \ref{sec:modified}, we can obtain sign change breakdown points larger than 1 if we impose additional restrictions. Columns (4)--(6) show the estimated sign change breakdown points when we restrict the magnitude of OVB to be below $M$. For the relatively unrestrictive value of $M = 10 | \widehat{\beta}_\text{med} |$, this additional restriction does not change the sign change breakdown points much. However, it is enough to push two of the breakdown points slightly above 1. For the more restrictive values $M = 3 | \widehat{\beta}_\text{med} |$ and $M = 2 | \widehat{\beta}_\text{med} |$, the breakdown points are now all moderately larger than 1 for the rule of thumb choice $R_\text{long}^2 = 1.3 \widehat{R}_\text{med}^2$. For $R_\text{long}^2 = 1$, however, the magnitude restrictions do not affect the values of the sign change breakdown points very much. Finally, notice that even with the strictest magnitude restriction, the sign change breakdown points (column 6) are substantially smaller than the explain away breakdown points (column 2).

\subsubsection*{Selection Bias Adjustments}

\begin{table}[!t]
\centering
\resizebox{0.82\textwidth}{!}{
\begin{talltblr}[caption = {The Sensitivity of Bias Adjustments. \label{table:biasAdjustmentSensitivity}}, remark{Note} = {All estimates are for the ``all associations'' specification from \cite{SatyanathVoigtlanderVoth2017}. All sensitivity methods use $R_\text{long}^2 = 1$. See Section \ref{sec:BiasCorrectionsAppendix} for a discussion and definition of \ref{assump:OsterA2}.}]
{p{0.48\textwidth}>{\centering} *{1}{>{\centering}p{0.18\textwidth}} *{1}{>{\centering}p{0.40\textwidth}}
}
  \toprule
Method & Assumptions & $\beta_\text{long}$ Values \\[2pt]
  \midrule
  \mystrut
  \multicolumn{3}{l}{A. Current Practice} \\[4pt]
\hline
\mystrut
Baseline Estimate & $\delta = 0$ & 0.17 \\
Oster's Prop 1 (our Prop \ref{prop:OsterP1}) & $\delta = 1$, \ref{assump:OsterA2} & 0.532 \\[4pt]
\midrule
\mystrut
  \multicolumn{3}{l}{B. Sensitivity to $\delta$ Assumptions} \\[4pt]
 \hline
 \mystrut
Oster's Prop 2 (our Thm \ref{thm:IDset}), $\widehat{\mathcal{B}}_I(\delta, R_\text{long}^2)$ & $\delta = 1$ & $\{ -0.0855, \,1.8947 \}$ \\
Oster's Prop 2, $\widehat{\mathcal{B}}_I(\delta, R_\text{long}^2)$ & $\delta = 0.99$ & $\{ -18.66, \, -0.0868, \, 1.736 \}$ \\
Oster's Prop 2, $\widehat{\mathcal{B}}_I(\delta, R_\text{long}^2)$ & $\delta = 1.01$ & $\{  -0.0843,  \,  2.133, \,  15.64 \}$ \\
Corollary \ref{cor:cum_IDset}, $\widehat{\overline{\mathcal{B}}}_I(\bar{\delta}, R_\text{long}^2)$ & $| \delta | \leq \bar{\delta} = 1$ & $(-\infty, -0.0855] \cup [0.0432,1.8947]$ \\
Corollary \ref{cor:cum_IDset}, $\widehat{\overline{\mathcal{B}}}_I(\bar{\delta}, R_\text{long}^2)$ & $| \delta | \leq \bar{\delta} = 1.05$ & $(-\infty,-0.0798] \cup [0.0415,\infty)$ \\[4pt]
\bottomrule
\end{talltblr}
} 
\end{table}

\cite{SatyanathVoigtlanderVoth2017} did not present selection bias adjusted estimates; as we document in Appendix \ref{sec:inPractice}, it is rare for empirical researchers to present both breakdown analysis and bias adjustments. To illustrate the results of Section \ref{sec:biasAdjustment}, here we also use bias adjustments to analyze the sensitivity of the results in \cite{SatyanathVoigtlanderVoth2017}.

Table \ref{table:biasAdjustmentSensitivity} summarizes our bias adjustment results. For brevity, here we focus on just the ``all associations'' treatment variable and the specification in column (2) of appendix Table \ref{table:mainTable0}, and we only consider $R_\text{long}^2 = 1$. Our findings are similar for the other treatment variables, specifications, and $R_\text{long}^2$ values. For simplicity we also do not present confidence sets for these point estimates. Panel A represents current empirical practice. The first row shows the baseline estimate. The second row shows the selection bias adjusted estimate obtained from Oster's Proposition 1 (our Prop.\ \ref{prop:OsterP1}), which assumes $\delta = 1$ and requires an additional assumption (\ref{assump:OsterA2}, discussed in Appendix \ref{sec:BiasCorrectionsAppendix}). Based on this estimate, the researcher would conclude that selection bias only works in their favor---it has the same sign as their baseline estimate and indeed is even larger.

However, now consider the results in Panel B. If we maintain $\delta = 1$ but drop \ref{assump:OsterA2}, then there are now two values of the coefficient that are consistent with the assumptions and the data, and one of them is negative. Hence even this minor change from current practice overturns the robustness finding based on Oster's Proposition 1 alone. Now, a researcher may claim that the value $-0.0855$ is close to zero, and perhaps is not statistically significantly different from zero, in which case one may argue that Oster's Proposition 2 with $\delta = 1$ still weakly supports a positive effect. So next we consider perturbing $\delta$ a slight amount: If we set $\delta = 0.99$ then the estimated identified set now includes a large negative value, $-18.66$. Likewise, if we set $\delta = 1.01$ then the estimated identified set includes a large positive value, $15.64$. Thus we see that the magnitude of the bias adjustment is very sensitive to the specific choice of $\delta$. Finally, the last two rows of Panel B show the estimated identified sets when we replace the assumption that $\delta$ is known exactly with the assumption that we only have a bound $\bar{\delta}$ on its magnitude (\ref{assump:deltaBarBound}). Here we see that for $\bar{\delta} = 1$, almost every negative value is consistent with the data and the assumptions, and many positive values are as well, including some bigger and smaller than the baseline estimate. If we increase this bound slightly to $\bar{\delta} = 1.05$ then arbitrarily large values of the coefficient are possible too.

\subsubsection*{Identified Set Plots and Main Empirical Recommendations}

To better understand the estimates in tables \ref{table:mainTable1} and \ref{table:biasAdjustmentSensitivity}, it is useful to plot the estimated identified set $\widehat{\mathcal{B}}_I(\delta,R_\text{long}^2)$ as a function of $\delta$. Figure \ref{fig:intro}, which we discussed in the introduction, does this for the ``all associations'' specification from \cite{SatyanathVoigtlanderVoth2017} with $R_\text{long}^2 = 1$; this corresponds to the first row of Table \ref{table:mainTable1}. Figure \ref{fig:osterCumulative} plots the corresponding estimated cumulative identified set as a function of $\bar{\delta}$, which assumes that $| \delta | \leq \bar{\delta}$, as well as its convex hull. These bounds provide a modified sensitivity analysis method that accurately represents the impact of omitted variables on the range of possible coefficient values. Our main recommendation is that empirical researchers present these set estimates. Note also that the magnitude constraint \ref{assn:magnRest} can be imposed by simply superimposing horizontal lines at the largest and smallest feasible values on this plot.

Researchers may also want to summarize the information shown in plots like Figure \ref{fig:osterCumulative} of the cumulative identified set. We recommend two ways in which researchers can do this, both of which modify and extend current practice as documented in Appendix \ref{sec:inPractice}.
\begin{enumerate}
\item First, in lieu of computing bias adjusted estimates based on equation \eqref{eq:OsterP1}, researchers can estimate the cumulative identified set $\widehat{\overline{\mathcal{B}}}_I(\overline{\delta},R_\text{long}^2)$ for a reference value of $\overline{\delta}$, such as $\overline{\delta} = 1$. Then they can select an empirically plausible value of $M$, the largest allowed value of the omitted variable bias. Then remove any elements from $\widehat{\overline{\mathcal{B}}}_I(\overline{\delta},R_\text{long}^2)$ that lay outside the range $[\widehat{\beta}_\text{med} - M, \widehat{\beta}_\text{med} + M]$, because these are deemed a priori implausible. The set of all remaining elements is then the estimated set of long regression coefficient values consistent with $| \delta | \leq \overline{\delta}$, the choice of $R_\text{long}^2$, and the restriction that the magnitude of OVB is no larger than $M$. This set can be interpreted as a set of bias corrected coefficients. For example, consider the specification in the second row of Table \ref{table:mainTable1}, with $M = 2 | \widehat{\beta}_\text{med} | = 2 \times 0.17 = 0.34$ (column (6)). In this case, $\widehat{\overline{\mathcal{B}}}_I(1, 1.3 \widehat{R}_\text{med}^2) \cap[-0.17, 0.51] = \big( (-\infty, -0.84] \cup [0.16,0.19] \big) \cap [-0.17, 0.51] = [0.16, 0.19]$. Thus researchers can state:
\begin{quote}
``Assuming the omitted variable bias is no larger than $0.34$, and that selection on unobservables is no larger than selection on observables, the bias corrected set of coefficient estimates is $[0.16,0.19]$.''
\end{quote}

\item Second, in lieu of computing explain away breakdown points, researchers can estimate the sign change breakdown point for an empirically plausible value of $M$. For example, consider again the specification in the second row of table \ref{table:mainTable1}, with $M = 2 | \widehat{\beta}_\text{med} | = 0.34$ (column (6)). Here the estimated sign change breakdown point is 2.83. Thus researchers can state:
\begin{quote}
``Assuming the omitted variable bias is no larger than 0.34, we estimate that selection on unobservables would need to be at least 2.83 times as large as selection on observables to change the sign of our baseline result.''
\end{quote}
\end{enumerate}
Whenever summary statements like these are made, we recommend that researchers discuss their choice of $M$, report their choice of $R_\text{long}^2$ (and consider exploring additional choices, see Section \ref{sec:R2longChoice}), and also still present identified set plots like in Figure \ref{fig:osterCumulative}, if only in an appendix.

\section{Explaining Away versus Sign Changes in Practice: Two Meta-Analyses}\label{sec:metaAnalysis}

In Section \ref{sec:empirical} we illustrated the difference between explain away and sign change breakdown points in a single application. In this section, we study this difference in 58 different empirical applications by conducting two complementary meta-analyses. Overall, we show that this distinction between explaining away and sign changes leads to substantively different conclusions about robustness to omitted variables in the majority of empirical applications.

\subsection{Data}

The first meta-analysis is based on the papers we survey in Appendix \ref{sec:inPractice} that were published in top five economics journals. We were able to replicate 12 of the 23 papers that implemented Oster's methods; the other 11 papers either had private data or other hurdles that prevented a full replication. Among the 12 papers we replicated, we gathered all regressions where the authors implemented Oster's methods, giving 194 total regressions. Within each paper, the robustness results for many of those regressions will be strongly correlated; for example, because they use the same outcome and treatment variables but with different controls. Hence we manually selected a subset of regressions to focus on, such as only including the longest specification per table. This gives our primary sample of 34 regressions from 12 papers.\footnote{Data sources: \cite{Arbatli2020data}, 
\cite{BFG2020data}, 
\cite{BertrandEtAl2020data}, 
\cite{DippelHeblich2021data}, 
\cite{Enke2020data}, 
\cite{Eugster2019data}, 
\cite{GavazzaEtAl2019data}, 
\cite{Gregg2020data}, 
\cite{GrosfeldEtAl2020data}, 
\cite{Heldring2021data}, 
\cite{Squicciarini2020data}, 
\cite{Tabellini2020data}}

The second meta-analysis is based on papers used in the meta-analysis of \citet[sections 2 and 5]{Oster2019}. This dataset has 141 regressions from 55 different papers that were published in either a top 5 economics journal or the \emph{AEJ: Applied Economics} between 2008 and 2013, and which reported results from coefficient stability analyses. See \cite{Oster2019,Oster2019data} for more details on the construction of this dataset. Since this dataset already only contains only a handful of regressions per paper, we did not further restrict the sample. This second dataset is useful both because it provides a larger sample of regressions and also because the inclusion criteria are qualitatively different, since none of these papers implemented Oster's methods because they preceded her paper.

\subsection{Results}

We present two sets of results. First, consider Table \ref{meta_table_1}. For each regression we estimated three breakdown points. First is the explain away breakdown point, also known as the absolute value of ``Oster's delta,'' which represents current practice. Then we estimated two versions of the sign change breakdown point: One without any restriction on the magnitude of the OVB ($M = +\infty$) and the other which assumes the OVB is at most ten times the baseline estimate ($M = 10 \cdot | \widehat{\beta}_\text{med} |$). Table \ref{meta_table_1} shows percentiles of the distribution of these breakdown points across the various regressions for each of the two samples described above (panels A and B), and for the two most common choices of $R_\text{long}^2$ (Oster's rule of thumb ($\min \{1.3 \widehat{R}_\text{med}^2, 1 \}$), and 1). 

\def\mystrut{\rule{0pt}{1.25\normalbaselineskip}}
\begin{table}[!t]
\centering
\SetTblrInner[talltblr]{rowsep=0pt}
\footnotesize
\begin{talltblr}[
  caption = {Meta-Analysis: Distribution of various breakdown points.\label{meta_table_1}},
]{c l | *{5}{>{\raggedleft\arraybackslash}p{0.05\textwidth}}}
    \toprule
    \mystrut
    & & \multicolumn{5}{c}{Percentiles} \\[0.5em]
    \multicolumn{1}{c}{$R^2_\text{long}$} & \multicolumn{1}{c |}{Breakdown Point Type}  & 10\% & 25\% & 50\% & 75\% & 90\% \\[0.4em]
\midrule
\multicolumn{7}{l}{Panel A: Our Sample of Top 5 Papers} \mystrut \\[0.5em]
\hline
\mystrut
    $\min \{ 1.3 \widehat{R}_\text{med}^2, 1 \}$ 
    & Explain away & 1.07 & 1.43 & 2.46 & 6.37 & 13.22 \\
    & Sign change ($M = 10 | \widehat{\beta}_\text{med} |$)  & 0.78 & 1.03 & 1.21 & 1.47 & 1.66 \\
    & Sign change ($M = +\infty$) & 0.78 & 1.0 & 1.0 & 1.0 & 1.0 \\[0.4em]
    \midrule
    \mystrut
    1 
    & Explain away & 1.17 & 1.37 & 1.9 & 2.4 & 8.82 \\
    & Sign change ($M = 10 | \widehat{\beta}_\text{med} |$)   & 0.21 & 0.22 & 0.88 & 1.17 & 1.37 \\
    & Sign change ($M = +\infty$) & 0.21 & 0.22 & 0.88 & 1.0 & 1.0 \\[0.4em]
 \midrule
\multicolumn{7}{l}{Panel B: Oster's Sample of Top 5 Papers + \emph{AEJ:Applied}} \mystrut \\[0.5em]
\hline
\mystrut
    $\min \{ 1.3 \widehat{R}_\text{med}^2, 1 \}$ 
    & Explain away & 1.29 & 2.65 & 9.0 & 32.54 & 74.14 \\
    & Sign change ($M = 10 | \widehat{\beta}_\text{med} |$)  & 0.95 & 1.15 & 2.21 & 8.53 & 31.52 \\
    & Sign change ($M = +\infty$) & 0.75 & 0.96 & 1.0 & 1.0 & 1.0 \\[0.4em]
    \midrule
    \mystrut
    1 
    & Explain away & 1.15 & 1.76 & 2.93 & 6.76 & 24.14 \\
    & Sign change ($M = 10 | \widehat{\beta}_\text{med} |$)   & 0.88 & 1.13 & 1.79 & 3.48 & 5.24 \\
    & Sign change ($M = +\infty$) & 0.42 & 0.94 & 1.0 & 1.0 & 1.0 \\[0.4em]
    \bottomrule
    \end{talltblr}
\begin{tablenotes}[center]
\footnotesize
\item[] \hspace{-2.1mm} \emph{Note}: Both samples only include results considered robust according to standard practice (the explain away breakdown point is weakly larger than 1), where we use the value of $R_\text{long}^2$ specified in the corresponding panel). Panel A: Top: $N=29$ regressions from 11 papers. Bottom: $N=11$ regressions from 6 papers. Panel B: Top: $N=111$ regressions from 47 papers. Bottom: $N=58$ regressions from 31 papers.
\end{tablenotes}
\end{table}

Here we further restrict the sample to all regressions that would be considered robust according to standard practice, meaning that the explain away breakdown point is larger than 1. By definition, the sign change breakdown point is always weakly smaller than the explain away breakdown point. Hence results that are not robust according to the explain away breakdown point will also be non-robust according to the sign change breakdown point. So this sample restriction allows us to answer the question: Among all results considered robust by standard practice, what is the impact of using sign change breakdown points to assess robustness?

First consider the sample in Panel A. As shown in Theorem \ref{thm:osterExplainAwayBounded}, the sign change breakdown point without any magnitude restrictions is never larger than 1. This result can be seen in the last line of all panels. For the rule of thumb choice of $R_\text{long}^2$, most of these breakdown points exactly equal 1. For $R_\text{long}^2 = 1$, more than half of the regressions are no longer robust when using $1$ as the robustness threshold point. In practice, researchers almost always go beyond binary robust / not-robust conclusions and also report the magnitude of the explain away breakdown point, interpreting larger magnitudes as evidence of greater robustness.  The second row of each sub-panel shows that even with a reasonably strong restriction on the magnitude of OVB, the sign change breakdown points are substantially smaller than the explain away breakdown points. For example, the median is cut in half. Near the top of the distribution, the 90th percentile shrinks 87\% from 13.22 down to 1.66 for the rule of thumb $R_\text{long}^2$, and similarly for $R_\text{long}^2 = 1$. The results for Panel B are qualitatively similar: The sign change breakdown points are substantially smaller in magnitude than the explain away breakdown points. Overall, these results show that in most empirical applications, common conclusions like
\begin{quote}
``selection on unobservables would have to be at least 4.1 times stronger than selection on observables to explain away the relationship'' (Squicciarini \citeyear{Squicciarini2020}, \emph{AER}, page 20) [Note: This computation uses $R_\text{long}^2 = 1$]
\end{quote}
will often greatly exaggerate the magnitude of robustness, since they are based on explain away rather than sign change breakdown points (in this application, with $M = 10 | \widehat{\beta}_\text{med} |$ the sign change breakdown point is instead only 1.07; see Table \ref{table:squic}).

\begin{table}[!t]
\centering
\SetTblrInner[talltblr]{rowsep=0pt}
\footnotesize
\begin{talltblr}[
  caption = {Meta-Analysis: Distribution of $\widehat{p}^\text{bp}$, 
  the largest magnitude of OVB (relative to $\widehat{\beta}_\text{med}$) we can allow for to guarantee that the estimated sign change breakdown point is at least as large as the estimated explain away breakdown point.\label{meta_table_2}},
]{c | *{5}{>{\raggedleft\arraybackslash}p{0.06\textwidth}}}
    \toprule
    \mystrut
    & \multicolumn{5}{c}{Percentiles} \\[0.5em]
    \multicolumn{1}{c}{$R^2_\text{long}$}  & 10\% & 25\% & 50\% & 75\% & 90\% \\[0.4em]
\midrule
\multicolumn{6}{l}{Panel A: Our Sample of Top 5 Papers} \mystrut \\[0.5em]
\hline
\mystrut
    $\min \{ 1.3 \widehat{R}_\text{med}^2, 1 \}$ & 1.0 & 1.0 & 1.0 & 1.3 & 67.2 \\
[0.4em]
    1 &   1.0 & 1.0 & 1.0 & 13.8 & 21.3 \\
[0.4em]
 \midrule
\multicolumn{6}{l}{Panel B: Oster's Sample of Top 5 Papers + \emph{AEJ:Applied}} \mystrut \\[0.5em]
\hline
\mystrut
    $\min \{ 1.3 \widehat{R}_\text{med}^2, 1 \}$ & 1.0 & 1.0 & 1.0 & 2.7 & 11.4 \\
[0.4em]
    1 & 1.0 & 1.0 & 1.0 & 3.9 & 61.6 \\
[0.4em]
    \bottomrule
    \end{talltblr}
\begin{tablenotes}[center]
\footnotesize
\item[] \hspace{-2.1mm} \emph{Note}: Both samples only include results considered robust according to standard practice (the explain away breakdown point is weakly larger than 1, where we use the value of $R_\text{long}^2$ specified in the corresponding panel). Panel A: First row: $N=29$ regressions from 11 papers. Second row: $N=11$ regressions from 6 papers. Panel B: First row: $N=111$ regressions from 47 papers. Second row: $N=58$ regressions from 31 papers.
\end{tablenotes}
\end{table}

Next consider Table \ref{meta_table_2}. Here we ask the question: 
\begin{quote}
Consider a researcher who does not want to change their current practice at all, and continues to report the explain away breakdown point. What is the weakest possible restriction on the magnitude of OVB that is necessary to correctly interpret the explain away breakdown point as a sign change breakdown point?
\end{quote}
To make magnitudes of OVB comparable across different regressions, we work with magnitudes relative to the baseline estimates $\widehat{\beta}_\text{med}$. Specifically, for the $i$th regression, let $\widehat{p}_i^\text{bp}$ be the largest magnitude of OVB (defined as a proportion of the baseline estimate) such that the estimated sign change breakdown point is at least as large as the estimated explain away breakdown point in this regression.\footnote{Formally, we define $\widehat{p}_i^\text{bp} \coloneqq \sup \{ p \geq 1 :  | \widehat{\delta}_i^{\text{bp},\text{sign}}(R_\text{long}^2, p \cdot | \widehat{\beta}_{\text{med},i} |) | \geq | \widehat{\delta}^{\text{bp},\text{explain away}}_i(R_\text{long}^2) | \}$.} The two panels in Table \ref{meta_table_2} give the distribution of $\widehat{p}_i^\text{bp}$ across regressions $i$ in our meta-analysis samples for different choices of $R_\text{long}^2$. 

Here we see that, across both samples and both choices of $R_\text{long}^2$, more than 50\% of regressions require that the \emph{sign} of the parameter of interest $\beta_\text{long}$ must be known a priori if we want to interpret the explain away breakdown point as a sign change breakdown point. That follows since---recall from Section \ref{sec:modified}---the choice $M = \beta_\text{med}$ (equivalently, $p=1$) implies that the magnitude restriction \ref{assn:magnRest} is equivalent to $\beta_\text{long} \in [0, 2 \beta_\text{med}]$ (when the baseline estimand is positive; the negative case is analogous). If the sign of $\beta_\text{long}$ was known a priori, however, then there would be no need to perform a sensitivity analysis to assess whether omitted variables could flip the sign of one's baseline results; they cannot by assumption.

For regressions above the 50th percentile, the implicit magnitude restriction is typically larger than $p = 1$, but is often still quite restrictive. For example, the 75th percentile in Panel A is 1.3 under the rule of thumb $R_\text{long}^2$ choice. To interpret this, suppose $\widehat{\beta}_\text{med} = 2$. Then $p = 1.3$ corresponds to the assumption that $\beta_\text{long} \in [-0.6, 4.6]$. Additional related results for these meta-analyses are in Appendix \ref{sec:extra_meta-analyses}.

\section{Conclusion}\label{sec:conclusion}

In this paper we distinguished between two kinds of breakdown points: sign change and explain away. Using theoretical results that apply regardless of the data generating process (Section \ref{sec:OsterSignFlip}), four detailed empirical applications (Section \ref{sec:empirical} and Appendix \ref{sec:extraEmpirical_applications}), and two meta-analyses covering 58 papers (Section \ref{sec:metaAnalysis}), we showed that the distinction between these two types of breakdown points matters in most empirical settings. In particular, explain away breakdown points often substantially overstate the robustness of results to omitted variables that can flip their sign. Given these results, we recommend that researchers plot the estimated identified set for the coefficient of interest under a variety of different assumptions about omitted variables; for example, see Figure \ref{fig:osterCumulative}. We also recommend reporting sign change breakdown points as robustness summary statistics, such as in Table \ref{table:mainTable1}. Both of these recommendations are easily implemented by using our companion Stata package \texttt{regsensitivity}.

\singlespacing
\bibliographystyle{econometrica}
\bibliography{OsterSignChange}

\makeatletter\@input{MP2aux.tex}\makeatother
\end{document}

\maketitle

\singlespacing
\appendix
\allowdisplaybreaks

This document contains various supplementary appendices. Appendix \ref{sec:inPractice} contains our survey of regression sensitivity analysis in practice. Appendix \ref{sec:extraEmpirical_details} has further details on our main empirical application. Appendix \ref{sec:extraEmpirical_applications} describes three additional detailed empirical applications. Appendix \ref{sec:extra_meta-analyses} contains additional results for our two meta-analyses. Appendix \ref{sec:technical} describes our formal identification results. Appendix \ref{sec:commonPitfallsAppendix} discusses some common pitfalls with current empirical practice. Appendix \ref{sec:estimation} discusses estimation and inference. Appendix \ref{sec:proofs} gives all proofs. Appendix I lists all papers included in our empirical survey of Appendix \ref{sec:inPractice}.

\section{A Survey of Regression Sensitivity Analysis in Practice}\label{sec:inPractice}

To understand how empirical researchers are using \cite{Oster2019} in practice, we performed a survey of the empirical economics literature. We constructed the sample by looking at all papers published in the three-year period from 2019 to 2021, which cite \cite{Oster2019}, and which were published in either \emph{AER}, \emph{JPE}, \emph{QJE}, \emph{REStud}, or \emph{ECMA}. This yields 35 papers. One of these is an econometric theory paper which cited \cite{Oster2019} as a related method only, and so we drop it from the sample. Three of these papers used the methods in \cite{Oster2019}, but only in working paper drafts, not in the final published version. We drop these three from the sample as well. This leaves 31 papers remaining. These papers are listed in Appendix I. 

This sample size is similar to the survey of \cite{BlandholEtAl2022}, who found 112 papers that use 2SLS in the same five journals, over the 19-year span from January 2000 to October 2018. In fact, since we are only looking at 3 years of data, Oster's \citeyearpar{Oster2019} method is more popular than 2SLS, on average, since 112 papers over 19 years is only about 18 papers every three years.\footnote{\cite{BlandholEtAl2022} gathered their data from Web of Science, whereas we use Google Scholar. Using our search criteria, Web of Science finds 28 papers. So this difference in counts is not due to differences in the original data sources.} This is perhaps not surprising since Oster's results can be applied to regressions based on many different identification strategies, including unconfoundedness, instrumental variables, and difference-in-differences.

There are three different ways in which researchers apply the results of \cite{Oster2019}: (1) Report the explain away breakdown point, (2) Report a bias adjusted estimand, and (3) Informally discuss coefficient stability and cite \cite{Oster2019} as justification. Table \ref{table:osterEmpiricalSurvey} summarizes the distribution of these three methods across the 31 papers in our dataset. About 80\% of the papers use and report the results of one of the two formal sensitivity analysis methods. Among these, roughly half only report breakdown points, and about half report bias adjusted estimands. Only a few papers do both. The remaining roughly 20\% of papers only informally discuss coefficient stability and refer to \cite{Oster2019} in their discussion. Some of those discussions suggest that they may have used one of the formal methods, but if so they did not report the findings.

\def\mystrut{\rule{0pt}{1.25\normalbaselineskip}}
\begin{table}[!t]
\centering
\SetTblrInner[talltblr]{rowsep=0pt}
\resizebox{0.75\textwidth}{!}{
\begin{talltblr}[
  caption = {How Empirical Researchers Use \cite{Oster2019}. \label{table:osterEmpiricalSurvey}},
  remark{Note} = {Sample contains 31 papers published in \emph{AER}, \emph{QJE}, \emph{JPE}, \emph{ECMA}, or \emph{REStud} between 2019--2021 and which use \cite{Oster2019}.},
]
{
p{0.48\textwidth}>{\centering} | *{1}{>{\centering}p{0.20\textwidth}} *{1}{>{\centering}p{0.15\textwidth}}
}
  \toprule
\mystrut
& Proportion & Count \\[0.4em]
\hline
\mystrut
Total & 100\% & 31 \\[0.4em]
\hline
\mystrut
1. Compute Explain Away Breakdown Points & 45\% & 14 \\
\hfill And vary $R_\text{long}^2$ & 21\% (conditional) & 3 \\[1em]
2. Compute Bias Adjusted Estimates & 42\% & 13 \\
\hfill And vary $\delta$ & 15\% (conditional) & 2 \\
\hfill And vary $R_\text{long}^2$ & 46\% (conditional) & 6 \\[1em]
3. Discuss Informally Only & 19\% & 6 \\[1em]
Both 1 and 2 & 6\% & 2 \\[0.4em]
\bottomrule
\end{talltblr}
}
\end{table}

Consider the papers that report the explain away breakdown point. Very few also vary $R_\text{long}^2$. Most use either $1.3 \widehat{R}_\text{med}^2$ or 1. One paper also used $2 \widehat{R}_\text{med}^2$. One paper presented the function $\widehat{\delta}^{\text{bp},\text{explain away}}(R_\text{long}^2)$ as $R_\text{long}^2$ varied from $\widehat{R}_\text{med}^2$ to 1. When discussing their results, researchers used a variety of phrases to describe their estimates of $\delta^{\text{bp},\text{explain away}}(R_\text{long}^2)$: It is the amount of selection on unobservables relative to observables needed...
\begin{itemize}

\item[] ``to fully explain away the coefficient on the explanatory variable of interest'' (Arbatl{\i}, Ashraf, Galor, and Klemp 2020, \emph{ECMA}, page 747)

\item[] ``to drive the results observed'' (Valencia Caicedo 2019, \emph{QJE}, online appendix page 6)

\item[] ``to reverse the sign of the results'' (Hoffman and Tadelis 2021, \emph{JPE}, page 272)

\end{itemize}
to give a few examples from our survey. As we have shown in this paper, the last quote is incorrect, because the sign change and explain away breakdown points are not equivalent. The first quote correctly describes it as an explain away breakdown point. The middle quote's interpretation depends on what it means ``to drive the results''; other authors often use similar language like making estimates ``spurious'' or ``go away.'' If these statements refer to the conclusion that the effect is not exactly zero (but might be positive or negative), then these descriptions are correct. If, however, these statements refer to the sign of the main finding, then these descriptions are incorrect. Our paper clarifies that empirical researchers must specify whether they are interested in robustness to sign changes, or merely to an exact zero.

Next consider the papers that report bias adjusted estimands. For any given specification, all of these papers report a scalar bias adjusted estimand. Most papers explicitly say that this is because they use Oster's Proposition 1 (see our Proposition \ref{prop:OsterP1}) and assume $\delta = 1$. Others do not clearly explain which result they are using to obtain their bias correction. None of the papers that use Oster's Proposition 1 discuss the validity of Assumption \ref{assump:OsterA2}. Eleven of the thirteen papers use just the single value $\delta = 1$. One of those papers considers $\delta = 0.5$ only (and no other values), but also uses the incorrect equation \eqref{eq:OsterP1approx} as described in Appendix \ref{sec:commonMistake}. Two papers consider multiple $\delta$ values, but one of them also uses the incorrect equation \eqref{eq:OsterP1approx}. The other paper uses a variation of Oster's results, which we have not formally analyzed, although we conjecture that the concerns we raise in our paper will apply to that extension as well. About half of the papers (7 of 13) only consider one value of $R_\text{long}^2$, always either $1.3 \widehat{R}_\text{med}^2$ or 1. In the other half (6 of 13), five papers also consider the second value $2 \widehat{R}_\text{med}^2$, while one paper considers up to 4 different values.

\section{Additional Details for Empirical Application of Section \ref{sec:empirical}}\label{sec:extraEmpirical_details}

\subsection{Data}

Section 3 of \cite{SatyanathVoigtlanderVoth2017} gives a full description of the data. Here we focus on the aspects relevant to our replication. \cite{SatyanathVoigtlanderVoth2017} use three different measures of social capital. All three are based on ``association density,'' the number of social clubs and associations per 1,000 city inhabitants. There are 22,127 different associations in the sample. The authors classify these into three types:
\begin{enumerate}
\item \emph{Military} associations: Stahlhelm (``steel helmet''), veterans associations.

\item \emph{Civic} associations (have ``a clearly nonmilitaristic/nationalist outlook''): Animal breeding, music, chess, hiking, women’s, citizens’, and homeland (Heimat) clubs, and some others (predominantly civic clubs, many with an artistic or creative pursuit such as gardening, theatre, or photography).

\item \emph{Other}: Sports, choirs, gymnastics, shooting, students/fraternities, lodges, youth, oldfellows, hunting, gentlemen, corps.
\end{enumerate}
$X$ is then defined to be the density of all associations, of civic associations only, or of military associations only.

Next consider the control variables. Recall from Section \ref{sec:defOfDelta} that we distinguish between two kinds of observed control variables: Those which we use for comparison in the sensitivity analysis ($W_1$) and those which we do not ($W_0$). The baseline controls $W_0$ are population, share Catholic, and share blue-collar, all defined for 1925. The comparison covariates $W_1$ are 
\begin{itemize}
\item Socioeconomic controls: share of Jews (1925), share unemployed (1933), welfare recipients per 1000 (1933), war participants per 1000 (1933), social insurance pensioners per 1000 (1933), log average income tax payment (1933), log average property tax payment (1933).

\item Political controls: Hitler speeches per 1000 (1932), average DNVP votes 1920--1928, average DVP votes 1920--1928, average SPD votes 1920--1928, average KPD votes 1920--1928.
\end{itemize}
In their Appendix G, \cite{SatyanathVoigtlanderVoth2017} also consider using the variables in $W_0$ for calibration; we omit this analysis for brevity. In their baseline analysis, they also include state fixed effects as covariates. They do not include these in their sensitivity analysis, however.

\subsection{Baseline Model Results}

\def\mystrut{\rule{0pt}{1.25\normalbaselineskip}}

\begin{table}[t]
\centering
\resizebox{0.55\textwidth}{!}{
\footnotesize
\begin{talltblr}[caption = {Baseline results for \cite{SatyanathVoigtlanderVoth2017}. \label{table:mainTable0}},  remark{Note} = {These estimates replicate part of Table 3 in \cite{SatyanathVoigtlanderVoth2017}: Column (1) replicates Panel A, columns (4)--(6). Column (3) replicates Panel B, columns (4)--(6). Panel A, column (2) replicates Panel B, column (3). Panels B and C, column (2) were not presented in Table 3 of \cite{SatyanathVoigtlanderVoth2017}. 
}]
{p{0.20\textwidth}>{\centering} *{3}{>{\centering}p{0.10\textwidth}
}
}
  \toprule
  \multicolumn{4}{c}{Dependent Variable: Nazi Party Entry, 1925--January 1933} \\[2pt]
  \midrule
  & (1) & (2) & (3) \\
  \midrule
  \multicolumn{4}{l}{A. Treatment = All Associations} \\[4pt]  
& 0.160 & 0.172 & 0.087 \\
& (0.054) & (0.052) & (0.050) \\[2pt]
  \midrule
  \multicolumn{4}{l}{B. Treatment = Civic Associations} \\[4pt]
& 0.429 & 0.441 & 0.284 \\
& (0.132) & (0.125) & (0.108) \\[2pt]
  \midrule
  \multicolumn{4}{l}{C. Treatment =  Military Associations} \\[4pt]
& 0.829 & 0.853 & 0.613 \\
& (0.268) & (0.276) & (0.357) \\[2pt]
  \midrule
  Controls: & & & \\[2pt]
  \ \ Baseline              & X & X & X \\
  \ \ Socioeconomic        &   & X & X \\
  \ \ Political            &   & X & X \\
  \ \ State Fixed Effects  &   &   & X \\
\bottomrule
\end{talltblr}
} 
\end{table}

Table \ref{table:mainTable0} shows the baseline regression results. Column (1) shows the estimated coefficient $\widehat{\beta}_\text{short}$ on $X$ from OLS of $Y$ on $(1,X,W_0)$, along with its standard error. Panels A--C show this coefficient for the three different definitions of the treatment variable $X$. First consider Panel A. This is their most broad definition of treatment, accounting for all kinds of associations. We see that the estimated coefficient is positive and statistically significantly different from zero at the 1\% level. Thus, adjusting for the baseline controls $W_0$, cities with more social capital, as measured by density of all associations, had higher rates of new membership in the Nazi party.

Does this result reflect a broad causal effect of social capital on entry into the Nazi party? One counter-argument is that it may merely reflect the effect of a specific kind of social capital---membership in military-specific associations. Indeed, in Panel C we see that, when restricting attention to military associations, the estimated coefficient is again positive and statistically significantly different from zero at the 1\% level. However, suppose we restrict attention to civic associations, which are defined to be non-militaristic and non-nationalistic. Panel B shows these estimates. There we continue to see a positive point estimate that is statistically significantly different from zero at the 1\% level. This finding suggests that the effect is not driven solely by military-specific social capital.

Are these results driven by omitted variable bias? The authors consider a variety of alternative methods to answer this question. We focus on just two: (i) Using additional controls and (ii) Using formal econometric methods for sensitivity analysis. Column (2) of Table \ref{table:mainTable0} includes the additional controls $W_1$ in the regressions. Specifically, it shows the estimated coefficient $\widehat{\beta}_\text{med}$ on $X$ from OLS of $Y$ on $(1,X,W_0,W_1)$, along with its standard error. Panels A--C again show this coefficient for the three different versions of treatment. We again see that all of the point estimates are positive and statistically significantly different from zero at the 1\% level. Column (3) adds state fixed effects. The point estimates are smaller but still positive and statistically significantly different from zero at the 10\% level. On page 500, the authors explain that this drop in precision arises from the fact that some of the historical policies which drove variation in association density often varied at the state level.

\subsection{Plots of Identified Sets for Fixed $\delta$}

Here we plot estimated identified sets $\widehat{\mathcal{B}}_I(\delta,R_\text{long}^2)$ as a function of $\delta$ for the three choices of treatment variables and two choices of $R_\text{long}^2$. Figure \ref{fig:empirical1} shows these plots. These plots show several of the features we discussed earlier in the paper. For example, at $\delta = 0$, the identified set is the singleton containing $\widehat{\beta}_\text{med}$. This is displayed as the vertical intercept. The horizontal intercept, in contrast, is the explain away breakdown point, as reported in column (2) of Table \ref{table:mainTable1}. Only two of these intercepts are visible given the range of the plots, however.

As $\delta$ varies, the identified set can contain up to three elements. For example, in the top left graph, there are two feasible values of $\beta_\text{long}$ at $\delta = 1$. We can also see the vertical asymptote at $\delta = 1$, which implies that for $\delta$'s very close to 1, the identified set will contain either very large positive values, very large negative values, or both. This explains why the identified set under the restriction $| \delta | \leq \bar{\delta}$ is bounded and finite for $\bar{\delta} < 1$ but unbounded for $\bar{\delta} > 1$. It also explains why the bias adjustments in Table \ref{table:biasAdjustmentSensitivity} are so sensitive to the exact choice of $\delta$.

\begin{figure}
\caption{Estimated identified sets $\widehat{\mathcal{B}}_I(\delta,R_\text{long}^2)$ for assessing the robustness of \cite{SatyanathVoigtlanderVoth2017} results to omitted variables. Left column: $R_\text{long}^2 = 1$. Right column: $R_\text{long}^2 = 1.3 \widehat{R}_\text{med}^2$. Top row: Treatment variable includes all associations. Middle row: Treatment variable includes civic associations only. Bottom row: Treatment variable includes military associations only. \label{fig:empirical1}}
\centering
\includegraphics[width=0.475\textwidth]{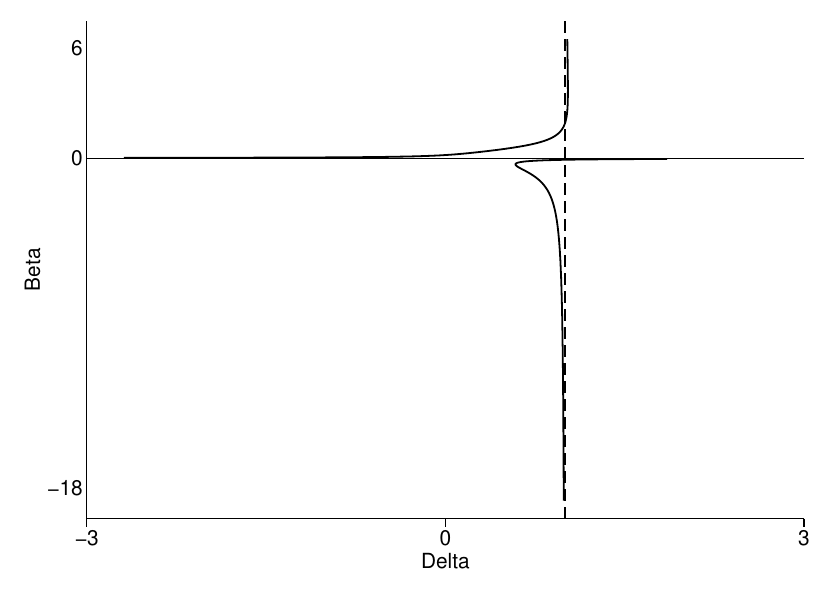}
\hspace{0.025\textwidth}
\includegraphics[width=0.475\textwidth]{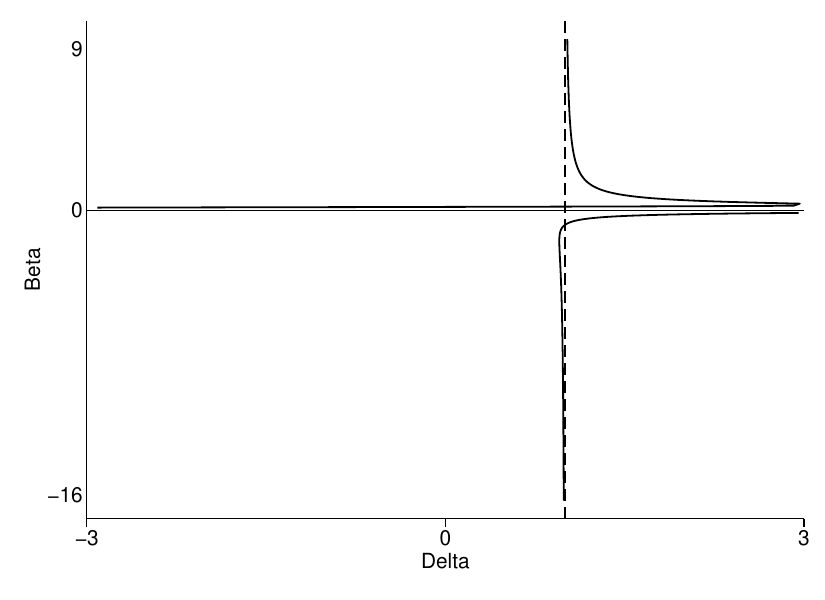} \\
\includegraphics[width=0.475\textwidth]{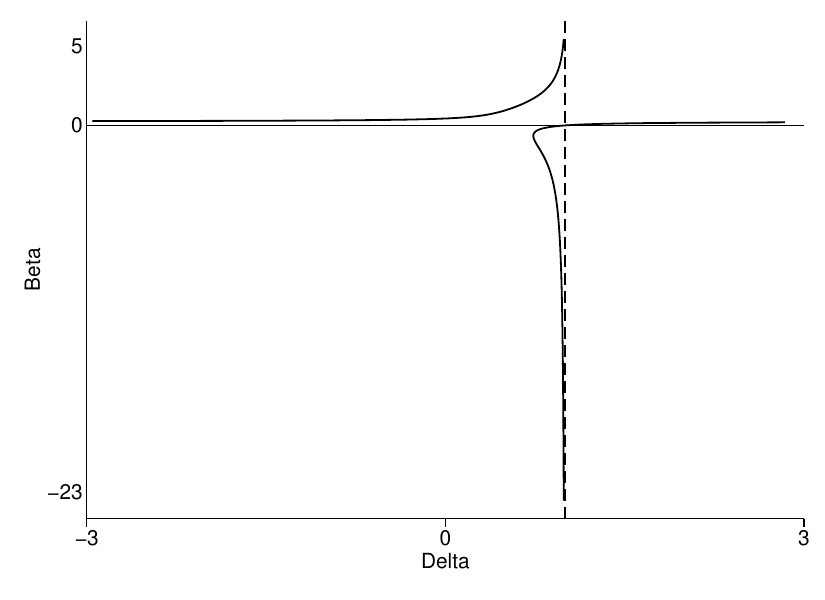}
\hspace{0.025\textwidth}
\includegraphics[width=0.475\textwidth]{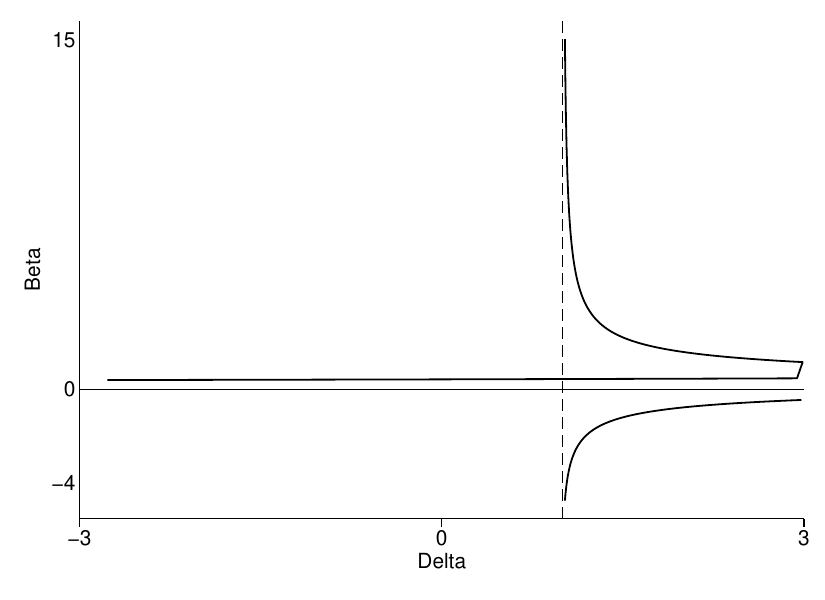} \\
\includegraphics[width=0.475\textwidth]{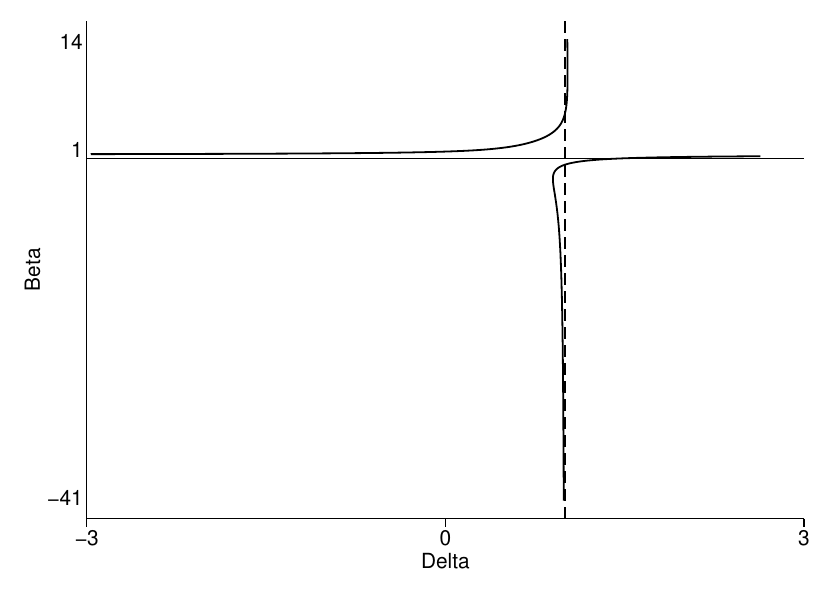}
\hspace{0.025\textwidth}
\includegraphics[width=0.475\textwidth]{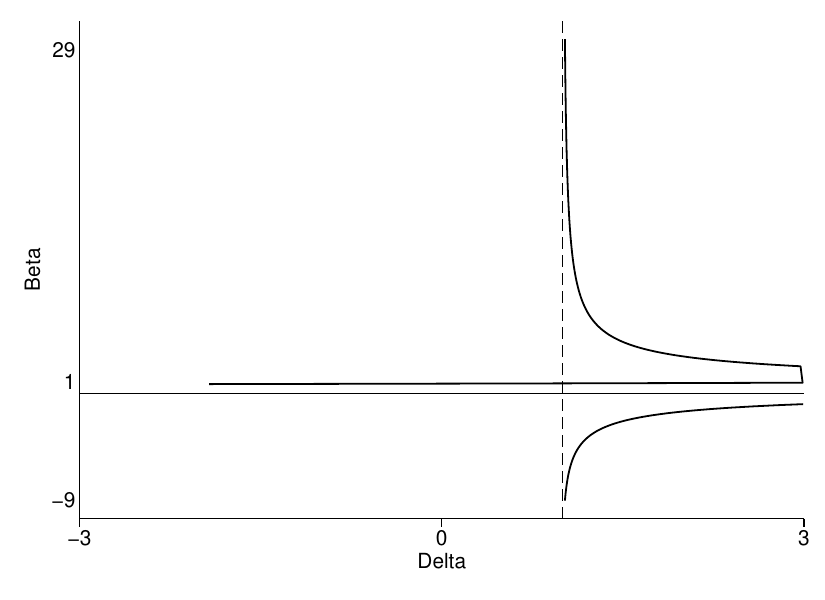}
\end{figure}

\subsection{Empirical Conclusions}

By only examining the explain away breakdown point, \cite{SatyanathVoigtlanderVoth2017} conclude that
\begin{quote}
``selection on unobservables would have to be substantially stronger than selection on observables for our main result to be overturned.'' (online appendix page 44)
\end{quote}
However, the language ``for our main result to be overturned'' suggests that they might also be interested in the robustness of their point estimates to sign changes. Indeed, if the true coefficient $\beta_\text{long}$ was negative, then higher social capital would lead to a \emph{decrease} in Nazi party membership, which would support the opposite of their conclusion that social capital supported the Nazi party and therefore helped to undermine an existing democracy.

When examining the relative robustness of their results for the three versions of the treatment variable, a second concern arises. Based on the explain away breakdown point, they say:
\begin{quote}
``Note also that for the logic of our argument, the coefficient for military associations is not the most important---what matters is that the civic clubs and associations have an important effect on Nazi Party entry. This is never in doubt in the Altonji-Elder-Taber/Oster exercise.'' (online appendix page 44)
\end{quote}
When examining the sign change breakdown points, however, we see that the coefficient for military associations is the most robust, among the three versions of treatment. The coefficient for civic clubs and associations, in contrast, is not so obviously robust to omitted variables. Thus the sign change breakdown points suggest a more tentative conclusion about robustness than the definitive one drawn by the authors.

\section{Additional Empirical Applications}\label{sec:extraEmpirical_applications}

Beyond our empirical application in Section \ref{sec:empirical}, we also considered all six papers published in the \emph{American Economic Review} from our Appendix \ref{sec:inPractice} survey which formally implemented the \cite{Oster2019} method. Among these papers, three had complete, publicly available replication data and code. In this section, we apply our methods to the data from these three papers. For brevity we only present the results; see the original papers for a discussion of the full empirical context, analysis, and data. Two of the three papers only presented breakdown point estimates, while the third paper only presented bias adjustments.

\subsubsection*{1. Squicciarini (2020, \emph{AER})}

\def\mystrut{\rule{0pt}{1.25\normalbaselineskip}}

\begin{table}[!t]
\centering
\SetTblrInner[talltblr]{rowsep=0pt}
\resizebox{0.80\textwidth}{!}{
\begin{talltblr}[
  caption = {Different Types of Breakdown Points for Assessing the Robustness of \cite{Squicciarini2020} Results to Omitted Variables.\label{table:squic}},
  remark{Note} = {See the original paper for variables included in $W_1$. $R_\text{long}^2 = 1$ for all results in columns (1)--(5). Column (1) shows $\widehat{\delta}^{\text{bp},\text{explain away}}(R_\text{long}^2)$. Column (2) shows the sign change breakdown point $\widehat{\delta}^{\text{bp},\text{sign}}(R_\text{long}^2, +\infty) = \widehat{\delta}^{\text{bp},\text{sign}}(R_\text{long}^2)$. Columns (3)--(5) show $\widehat{\delta}^{\text{bp},\text{sign}}(R_\text{long}^2, M)$.},
]{c | *{1}{>{\centering}p{0.145\textwidth}} *{4}{>{\centering}p{0.065\textwidth}} } 
  \toprule
\mystrut
 & Explain Away & \multicolumn{4}{c}{Sign Change with $M$ equal to}  \\[0.4em]
& & $+ \infty$ & $10 | \widehat{\beta}_\text{med} |$ & $3 | \widehat{\beta}_\text{med} |$ & $2 | \widehat{\beta}_\text{med} |$ & \\[4pt]
Outcome variable  & (1) & (2) & (3) & (4) & (5)  \\[4pt]
\hline
\mystrut
Share Ind. Workers &9.0950&0.9982&1.3683&3.4494&4.6986 \\
Machines &4.0699&0.9968&1.0744&1.8306&2.4460 \\[2pt]
\bottomrule
\end{talltblr}
}
\end{table}

In this paper, the author estimates explain away breakdown points for two different outcome variables, and reports them in Table A.11 of the online appendix. Based on these estimates, the author concludes
\begin{quote}
``that selection on unobservables would have to be at least 4.1 times stronger than selection on observables to explain away the relationship between the [treatment variable] and the outcome variables, making it unlikely that unobserved factors are confounding the results.'' (page 20)
\end{quote}
In Table \ref{table:squic}, we replicate and extend this analysis. Column (1) shows the explain away breakdown point estimates as reported and correctly computed by the original author. Column (2) shows the estimated sign change breakdown points. Consistent with our Theorem \ref{thm:osterExplainAwayBounded}, these are less than one. However, they are quite close to one. In columns (3)--(5) we add magnitude restrictions. Column (3) shows that even with a fairly loose magnitude restriction we are able to obtain sign change breakdown points above one. And, especially for the first outcome variable, we can substantially increase the breakdown point with further magnitude restrictions.

Overall, our analysis here shows that although the explain away breakdown points are substantially larger than the sign change breakdown points, if one uses the conventional cutoffs for robustness and imposes some magnitude restrictions on the OVB, the author's conclusion that their results are robust to omitted variables continues to hold.

\subsubsection*{2. Dippel and Heblich (2021, \emph{AER})}

In Table 2 on page 491, the authors report estimates of the explain away breakdown point. On that page, they say
\begin{quote}
``To get a sense for potential selection on unobservables in the full and matched sample, we report Oster's $\delta$ in the bottom of Table 2. Oster’s $\delta$ measures how large the bias from unobservables would have to be relative to bias from observable \emph{to imply a true value of $\beta = 0$} in equation (1). In the full sample, $\delta$ is consistently around 0.5. In the matched sample by contrast, $\delta$ ranges between 1.7 and 1.9. This implies a strong sense of robustness in the matched sample, where selection on unobservables would have to be almost twice as strong as selection on observables to \emph{make our core estimate go away}.'' [emphasis added]
\end{quote}
Our extended analysis, shown in Table \ref{table:Dippel}, is qualitatively similar: The results appear robust for the matched sample but not the full sample. Specifically, first consider the full sample. Consider the first line, the core controls; the results are similar for the other two choices of controls. In column (1) we see that the estimated explain away breakdown point is 0.5170. This is also the sign change breakdown point, as shown in column (2). Given that this value is substantially less than 1, where the asymptote in the identified set appears, imposing magnitude constraints (columns (3)--(5)) does not affect the value of the sign change breakdown point. All of these results together suggest that the results for the full sample are not robust.

Next consider the matched sample. Consider the first line, the core controls; the results are similar for the other two choices of controls. Here we see the estimated explain away breakdown point is 1.7073. This is not the sign change breakdown point, however. In this case, the sign change breakdown point is exactly 1. Nonetheless, with even a loose magnitude constraint---assuming the OVB is bounded above by $10 | \widehat{\beta}_\text{med}|$---the sign change and explain away breakdown points coincide exactly. Further magnitude restrictions are not helpful because these restrictions only matter insofar as they bring the sign change and explain away breakdown points closer together. Overall, these results suggest that the results for the matched sample are robust.

\begin{table}[!t]
\centering
\SetTblrInner[talltblr]{rowsep=0pt}
\resizebox{0.80\textwidth}{!}{
\begin{talltblr}[
  caption = {Different Types of Breakdown Points for Assessing the Robustness of \cite{DippelHeblich2021} Results to Omitted Variables.\label{table:Dippel}},
  remark{Note} = {See original paper for variables included in each choice of $W_1$. $R_\text{long}^2 = 1$ for all results in columns (1)--(5). Column (1) shows $\widehat{\delta}^{\text{bp},\text{explain away}}(R_\text{long}^2)$. Column (2) shows the sign change breakdown point $\widehat{\delta}^{\text{bp},\text{sign}}(R_\text{long}^2, +\infty) = \widehat{\delta}^{\text{bp},\text{sign}}(R_\text{long}^2)$. Columns (3)--(5) show $\widehat{\delta}^{\text{bp},\text{sign}}(R_\text{long}^2, M)$.},
]{c c |  *{1}{>{\centering}p{0.145\textwidth}} *{4}{>{\centering}p{0.065\textwidth}}}
  \toprule
\mystrut
 & & Explain Away & \multicolumn{4}{c}{Sign Change with $M$ equal to}  \\[0.4em]
& & & $+ \infty$ & $10 | \widehat{\beta}_\text{med} |$ & $3 | \widehat{\beta}_\text{med} |$ & $2 | \widehat{\beta}_\text{med} |$ \\[4pt]
Sample & $W_1$ & (1) & (2) & (3) & (4) & (5) \\[4pt]
\hline
\mystrut
& Core controls &  0.5170&0.5170&0.5170&0.5170&0.5170 \\
Full & + other Vselect &  0.5187 & 0.5187&0.5187&0.5187&0.5187 \\
& All controls &  0.5156 &0.5156&0.5156&0.5156&0.5156 \\[0.4em] 
\hline
\mystrut
& Core controls & 1.7073& 1.0000&1.7073&1.7073&1.7073 \\
Matched & + other Vselect  & 1.9654& 1.0000&1.9654&1.9654&1.9654 \\
& All controls  & 1.9554 &1.0000&1.9554&1.9554&1.9554 \\[2pt] 
\bottomrule
\end{talltblr}
}
\end{table}

\begin{table}[!t]
\centering
\resizebox{0.98\textwidth}{!}{
\footnotesize
\begin{talltblr}[caption = {The Sensitivity of Bias Adjustments in \cite{Gregg2020}. \label{table:GreggCombined}}, remark{Note} = {All sensitivity methods use $R_\text{long}^2 = 1$. Panel I: Specification 1 corresponds to column (1) of Panel B in Table 3 of Gregg (2020) while specification 2 corresponds to column (4). Panel II: Specification 1 corresponds to column (2) of Panel B in Table 3 of Gregg (2020) while specification 2 corresponds to column (5). Panel III: Specification 1 corresponds to column (3) of Panel B in Table 3 of Gregg (2020) while specification 2 corresponds to column (6). See Section \ref{sec:BiasCorrectionsAppendix} for a discussion and definition of \ref{assump:OsterA2}.}]
{p{0.40\textwidth}>{\centering} *{1}{>{\centering}p{0.15\textwidth}} *{1}{>{\centering}p{0.40\textwidth}} *{1}{>{\centering}p{0.40\textwidth}}}
  \toprule
 &  & \multicolumn{2}{c}{$\beta_\text{long}$ Values} \\[2pt]
Method & Assumptions & Specification 1  & Specification 2  \\[4pt]
\midrule
\mystrut
  \multicolumn{4}{l}{\textbf{I. Outcome variable: }$\log(\text{R}/\text{L})$} \\[4pt]
 \hline
 \mystrut
  \multicolumn{4}{l}{A. Current Practice} \\[4pt]
\hline
\mystrut
\hspace{3mm}Baseline Estimate & $\delta = 0$
& 0.495 
& 0.451 \\
\hspace{3mm}Oster's Prop 1 (our Prop \ref{prop:OsterP1}) & $\delta = 1$, \ref{assump:OsterA2}
& 0.434 
& 0.367 \\[4pt]
\hline
\mystrut
  \multicolumn{4}{l}{B. Sensitivity to $\delta$ Assumptions} \\[4pt]
 \hline
 \mystrut
\hspace{3mm}Oster's Prop 2 (our Thm \ref{thm:IDset}), $\widehat{\mathcal{B}}_I(\delta, R_\text{long}^2)$ & $\delta = 1$ 
& $\{ 0.412,   197 \}$
& $\{ 0.331, 152 \}$ \\
\hspace{3mm}Oster's Prop 2, $\widehat{\mathcal{B}}_I(\delta, R_\text{long}^2)$ & $\delta = 0.9999$ 
& $\{ 0.412, 210, 3200 \}$
& $\{ 0.331, 158, 3817 \}$ \\
\hspace{3mm}Oster's Prop 2, $\widehat{\mathcal{B}}_I(\delta, R_\text{long}^2)$ & $\delta = 1.01$ 
& $\{  -99.8,   0.411,   67.1 \}$ 
& $\{ -99.0, 0.329, 60.7 \}$ \\
\hspace{3mm}Corollary \ref{cor:cum_IDset}, $\widehat{\overline{\mathcal{B}}}_I(\bar{\delta}, R_\text{long}^2)$ & $| \delta | \leq \bar{\delta} = 1$ 
& $[0.412, 0.550] \cup [197,\infty)$ 
& $[0.331, 0.526] \cup [152,\infty)$ \\
\hspace{3mm}Corollary \ref{cor:cum_IDset}, $\widehat{\overline{\mathcal{B}}}_I(\bar{\delta}, R_\text{long}^2)$ & $| \delta | \leq \bar{\delta} = 1.05$ 
& $(-\infty,-39.2] \cup [0.407, 0.553] \cup [33.8,\infty)$ 
& $(-\infty,-37.9] \cup [0.323, 0.529] \cup [31.4,\infty)$ \\[0.4em]
\midrule
\mystrut
  \multicolumn{4}{l}{\textbf{II. Outcome variable: }$\log(\text{HP}/\text{L})$} \\[4pt]
\hline
  \mystrut
  \multicolumn{4}{l}{A. Current Practice} \\[4pt]
\hline
\mystrut
\hspace{3mm}Baseline Estimate & $\delta = 0$
& 0.250 
& 0.216 \\
\hspace{3mm}Oster's Prop 1 (our Prop \ref{prop:OsterP1}) & $\delta = 1$, \ref{assump:OsterA2}
& 0.804
& 0.782 \\[4pt]
\hline
\mystrut
  \multicolumn{4}{l}{B. Sensitivity to $\delta$ Assumptions} \\[4pt]
 \hline
 \mystrut
\hspace{3mm}Oster's Prop 2 (our Thm \ref{thm:IDset}), $\widehat{\mathcal{B}}_I(\delta, R_\text{long}^2)$ & $\delta = 1$ 
& $\{ -25.8,   0.9532 \}$
& $\{ -39.4,   0.9525 \}$ \\
\hspace{3mm}Oster's Prop 2, $\widehat{\mathcal{B}}_I(\delta, R_\text{long}^2)$ & $\delta = 0.99$ 
& $\{ -367,  -27.4,   0.9449 \}$
& $\{ -393,  -43.4,   0.9440 \}$ \\
\hspace{3mm}Oster's Prop 2, $\widehat{\mathcal{B}}_I(\delta, R_\text{long}^2)$ & $\delta = 1.01$ 
& $\{-24.4, 0.9615,411\}$ 
& $\{ -36.6,   0.961,   464 \}$ \\
\hspace{3mm}Corollary \ref{cor:cum_IDset}, $\widehat{\overline{\mathcal{B}}}_I(\bar{\delta}, R_\text{long}^2)$ & $| \delta | \leq \bar{\delta} = 1$ 
& $(-\infty,-25.8] \cup [-0.301, 0.9532]$
& $(-\infty,-39.4] \cup [-0.358, 0.9525]$ \\
\hspace{3mm}Corollary \ref{cor:cum_IDset}, $\widehat{\overline{\mathcal{B}}}_I(\bar{\delta}, R_\text{long}^2)$ & $| \delta | \leq \bar{\delta} = 1.05$ 
& $(-\infty,-21.0] \cup [-0.326, 0.995] \cup [95.0,\infty)$ 
& $(-\infty,-30.2] \cup [-0.383, 0.996] \cup [112,\infty)$ \\[0.4em]
\midrule
\mystrut
  \multicolumn{4}{l}{\textbf{III. Outcome variable: }TFP} \\[4pt]
\hline
  \mystrut
  \multicolumn{4}{l}{A. Current Practice} \\[4pt]
\hline
\mystrut
\hspace{3mm}Baseline Estimate & $\delta = 0$
& 0.102 
& 0.015 \\
\hspace{3mm}Oster's Prop 1 (our Prop \ref{prop:OsterP1}) & $\delta = 1$, \ref{assump:OsterA2}
& 0.132
& -0.141 \\[4pt]
\hline
\mystrut
  \multicolumn{4}{l}{B. Sensitivity to $\delta$ Assumptions} \\[4pt]
 \hline
 \mystrut
\hspace{3mm}Oster's Prop 2 (our Thm \ref{thm:IDset}), $\widehat{\mathcal{B}}_I(\delta, R_\text{long}^2)$ & $\delta = 1$ 
& $\{ -212, 0.163 \}$
& $\{ -0.876,   24.3 \}$ \\
\hspace{3mm}Oster's Prop 2, $\widehat{\mathcal{B}}_I(\delta, R_\text{long}^2)$ & $\delta = 0.99999$ 
& $\{-7090,  -218,   0.163 \}$
& $\{ -0.876,   24.3,   22452 \}$ \\
\hspace{3mm}Oster's Prop 2, $\widehat{\mathcal{B}}_I(\delta, R_\text{long}^2)$ & $\delta = 1.01$ 
& $\{ -35.6,  0.164,   43.0 \}$ 
& $\{-35.9,  -0.922,   14.6 \}$ \\
\hspace{3mm}Corollary \ref{cor:cum_IDset}, $\widehat{\overline{\mathcal{B}}}_I(\bar{\delta}, R_\text{long}^2)$ & $| \delta | \leq \bar{\delta} = 1$ 
& $(-\infty,-212] \cup [0.079, 0.163]$
& $[-0.876, 0.115] \cup [24.3, \infty)$ \\
\hspace{3mm}Corollary \ref{cor:cum_IDset}, $\widehat{\overline{\mathcal{B}}}_I(\bar{\delta}, R_\text{long}^2)$ & $| \delta | \leq \bar{\delta} = 1.05$ 
& $(-\infty,-16.5] \cup [.079, 0.168] \cup [18.0,\infty)$ 
& $(-\infty,-11.0] \cup [-1.16, 0.118] \cup [7.93,\infty)$ \\[2pt] 
\bottomrule
\end{talltblr}
} 
\end{table}

\subsubsection*{3. Gregg (2020, \emph{AER})}

This paper uses Oster's results in tables 3, 5, and 7 of the main paper. For brevity we only consider the results in Table 3. The author considers three different outcome variables, and two choices of sample and calibration covariates $W_1$ for each outcome variable. Unlike the two papers we considered above, Gregg does not estimate breakdown points. Instead, they only present bias adjusted estimates. Specifically, following Oster's (2019) recommendation and the defaults in the accompanying Stata package \texttt{psacalc}, Panel B of Table 3 of \cite{Gregg2020} reports the single element of the estimated identified set obtained from Oster's Proposition 2 (our Theorem \ref{thm:IDset}) which is closest to $\widehat{\beta}_\text{med}$ (this is an estimate of the value $\beta^*(\delta, R_\text{long}^2)$ defined in Section \ref{sec:ostersBoundingSet}). Using these estimates, Gregg provides the following analysis:
\begin{quote}
Some suggestive evidence on the nature of selection is provided by the Oster (2017) bound listed for each regression in panel B of Table 3. Here I report the estimate of the coefficient on corporation in the case where the selection on unobservables is equal to the selection on observables (when $\delta$ equals 1). The observables in this case are the province, industry, and year controls. These coefficients suggest some positive selection in revenue per worker, but possibly some negative selection in machine power per worker and TFP. Columns 4--6 introduce controls for the factory’s age, which makes little difference for the revenue per worker or power per worker regressions. However, factory age is highly correlated with TFP, since age proxies for survival. Furthermore, with the age control included, the estimated Oster bound indicates positive selection in the TFP regression, likely since selection on the newly included observable, factory age, was substantial.'' (pages 414--415)
\end{quote}
In this analysis Gregg uses Oster's method to assess the direction of the selection bias, by comparing the bias adjusted estimates with the baseline estimates. Overall, Gregg concludes that there is positive selection bias, meaning that the baseline estimates are too large. In our analysis, we show that Gregg's conclusion about the direction of selection bias no longer holds once we account for the sensitivity of the bias adjusted estimands obtained using Oster's method. 

To see this, we replicate and extend her estimates for each of the three outcome variables in Table \ref{table:GreggCombined}. First consider Panel I, for the $\log(\text{R}/\text{L})$ outcome variable. Consider specification 1. Gregg compared the element of the set $\{ 0.412, 197\}$ that is closest to the baseline estimate 0.495 with that baseline estimate. This gives $0.412 < 0.495$, suggesting that the baseline estimate is too large. Similarly, in specification 2 we see that $0.331 < 0.451$, suggesting again that the baseline estimate is too large. There are two main concerns with this comparison, however. First, it ignores the other element of the estimated Proposition 2 identified set. For example, in specification 1 the baseline estimate $0.495$ is smaller than $197$, the other element of the estimated identified set. This comparison suggests that the baseline estimate is \emph{too small}.

Now, as in our discussion of magnitude restrictions in Section \ref{sec:modified}, one might be willing to eliminate this second element by imposing an additional assumption that restrictions the magnitude of the coefficient. This does not salvage the argument, however, because there is a second issue: In the fourth line of Panel B we show the estimated cumulative identified set under the assumption that $| \delta | \leq 1$. This set is a union of two intervals. Even if we ignore the upper interval, this set is $[0.412, 0.550]$. This set contains values that are smaller \emph{and larger} than the baseline estimate of 0.495. Thus if we allow for $\delta$ to take any value so long as its magnitude is not larger than 1, rather than assuming that $\delta$ is exactly equal to 1, then we can no longer conclude anything about the direction of the selection bias. A similar analysis holds for the second specification, and the other two outcome variables.

\section{Additional Results for the Meta-Analyses of Section \ref{sec:metaAnalysis}}\label{sec:extra_meta-analyses}

This section presents an extended version of Table \ref{meta_table_2} in Section \ref{sec:metaAnalysis}. Table \ref{meta_table_2} presented the distribution of $\widehat{p}_i^\text{bp}$, the largest magnitude of OVB (defined as a proportion of the baseline estimate) such that the estimated sign change breakdown point is at least as large as the estimated explain away breakdown point, across both meta-analysis samples and two values of $R^2_\text{long}$. 

Here we first generalize the definition of $\widehat{p}_i^\text{bp}$ to allow us to compare the estimated sign change breakdown point with $\delta$ values other than the estimated explain away breakdown point. Specifically, for the $i$th regression, let $\widehat{p}_i^\text{bp}(\delta_\text{hypo})$ be the largest magnitude of OVB (defined as a proportion of the baseline estimate) such that the estimated sign change breakdown point is at least as large as $| \delta_\text{hypo} |$, where $\delta_\text{hypo}$ is a hypothesized value that is allowed to vary.\footnote{Formally, we define $\widehat{p}_i^\text{bp}(\delta_\text{hypo}) \coloneqq \sup \{ p \geq 1 :  | \widehat{\delta}_i^{\text{bp},\text{sign}}(R_\text{long}^2, p \cdot | \widehat{\beta}_{\text{med},i} |) | \geq | \delta_\text{hypo} | \}$.} For example, setting  $\delta_\text{hypo} = \widehat{\delta}_i^{\text{bp},\text{explain away}}(R_\text{long}^2)$ yields $\widehat{p}_i^\text{bp}(\delta_\text{hypo}) = \widehat{p}_i^\text{bp}$, the values whose distribution is reported in Table \ref{meta_table_2}. These values answer the question: When can the explain away breakdown point be interpreted as a sign change breakdown point? Table \ref{meta_table_2_appendix} again shows these values, in the rows that set $\delta_\text{hypo} = \widehat{\delta}^{\text{bp},\text{explain away}}(R_\text{long}^2)$. Table \ref{meta_table_2_appendix} also includes three new values, however: $\delta_{\text{hypo}}$ equal to 2, 1.05, and 1.

\begin{table}[!t]
\centering
\SetTblrInner[talltblr]{rowsep=0pt}
\footnotesize
\begin{talltblr}[
  caption = {Meta-Analysis: Distribution of $\widehat{p}^\text{bp}$, the largest magnitude of OVB (relative to $\widehat{\beta}_\text{med}$) we can allow for to guarantee that the sign change breakdown point is at least as large as $\delta_\text{hypo}$.\label{meta_table_2_appendix}},
]{c c | *{5}{>{\raggedleft\arraybackslash}p{0.06\textwidth}}}
    \toprule
    \mystrut
    & & \multicolumn{5}{c}{Percentiles} \\[0.5em]
    \multicolumn{1}{c}{$R^2_\text{long}$} & \multicolumn{1}{c |}{$\delta_\text{hypo}$}  & 10\% & 25\% & 50\% & 75\% & 90\% \\[0.4em]
\midrule
\multicolumn{7}{l}{Panel A: Our Sample of Top 5 Papers} \mystrut \\[0.5em]
\hline
\mystrut
    $\min \{ 1.3 \widehat{R}_\text{med}^2, 1 \}$ &  $\widehat{\delta}^{\text{bp},\text{explain away}}(R_\text{long}^2)$ & 1.0 & 1.0 & 1.0 & 1.3 & 67.2 \\
    & 2 & 1.0 & 1.4 & 3.4 & 6.9 & 67.2 \\
    & 1.05 & 2.0 & 9.2 & 25.2 & 54.0 & 198.9 \\
    & 1 & 2.1 & 709.8 & $\infty$ & $\infty$ & $\infty$ \\[0.4em]
    \midrule
    \mystrut
    1 &  $\widehat{\delta}^{\text{bp},\text{explain away}}(R_\text{long}^2)$ & 1.0 & 1.0 & 1.0 & 13.8 & 21.3 \\
    & 2 & 1.0 & 1.0 & 1.3 & 13.8 & 21.3 \\
    & 1.05 & 1.2 & 1.2 & 4.5 & 53.2 & 147.8 \\
    & 1 & 1.2 & 1.2 & 4.9 & $\infty$ & $\infty$ \\[0.4em]
 \midrule
\multicolumn{7}{l}{Panel B: Oster's Sample of Top 5 Papers + \emph{AEJ:Applied}} \mystrut \\[0.5em]
\hline
\mystrut
    $\min \{ 1.3 \widehat{R}_\text{med}^2, 1 \}$ &  $\widehat{\delta}^{\text{bp},\text{explain away}}(R_\text{long}^2)$ & 1.0 & 1.0 & 1.0 & 2.7 & 11.4 \\
    & 2 & 1.0 & 3.3 & 12.8 & 58.6 & 116.0 \\
    & 1.05 & 3.7 & 16.8 & 70.7 & 255.0 & 771.6 \\
    & 1 & 5.3 & 59.6 & 1978.6 & $\infty$ & $\infty$ \\[0.4em]
    \midrule
    \mystrut
    1 & $\widehat{\delta}^{\text{bp},\text{explain away}}(R_\text{long}^2)$ & 1.0 & 1.0 & 1.0 & 3.9 & 61.6 \\
    & 2 & 1.0 & 2.1 & 10.8 & 45.8 & 211.8 \\
    & 1.05 & 1.5 & 16.3 & 50.5 & 192.3 & 760.0 \\
    & 1 & 1.6 & 32.1 & 216.4 & $\infty$ & $\infty$ \\[0.4em]
    \bottomrule
    \end{talltblr}
\begin{tablenotes}[center]
\footnotesize
\item[] \hspace{-2.1mm} \emph{Note}: Both samples only include results considered robust according to standard practice (the explain away breakdown point is weakly larger than 1, where we use the value of $R_\text{long}^2$ specified in the corresponding panel). Panel A: Top: $N=29$ regressions from 11 papers. Bottom: $N=11$ regressions from 6 papers. Panel B: Top: $N=111$ regressions from 47 papers. Bottom: $N=58$ regressions from 31 papers. The value $\infty$ means our optimizer to compute $\widehat{p}_i^\text{bp}$ stopped at the preset upper bound of 1 million, which can be interpreted as effectively infinite. For $\delta_\text{hypo} = 2$ and 1.05, we actually use $\min \{2,  \widehat{\delta}^{\text{bp},\text{explain away}}(R_\text{long}^2) \}$ and $\min \{1.05, \widehat{\delta}^{\text{bp},\text{explain away}}(R_\text{long}^2) \}$ to ensure that the distributions in each sub-panel are stochastically ranked from top to bottom.
\end{tablenotes}
\end{table}

We can use these additional values $\widehat{p}_i^\text{bp}(\delta_\text{hypo})$ to answer several questions beyond our primary question of when explain away breakdown points can be interpreted as sign change breakdown points. For example, what is the weakest possible restriction on the magnitude of OVB to ensure that the sign change breakdown point is at least 1? $\widehat{p}_i(\delta_\text{hypo})$ with $\delta_\text{hypo} = 1$ answers this question. The bottom row of all sub-panels shows the distribution of our estimates of these magnitudes. Here we see that even for this very weak conclusion, at least 10\% of regressions still require substantial magnitude restrictions. Most papers, however, require a fairly weak, or essentially no magnitude restriction in order to obtain a sign change breakdown point equal to or larger than 1. However, because of the asymptote at $\delta = 1$, these $\widehat{p}_i^\text{bp}$ values may not be informative about values for $\delta$ cutoffs slightly above 1. So the middle rows of each sub-panel show estimates of $\widehat{p}_i^\text{bp}(\delta_\text{hypo})$ for values of $\delta_\text{hypo}$ close to, but slightly above 1. In these middle rows we see that obtaining sign change breakdown points which are meaningfully larger than 1 often requires substantial magnitude restrictions.

\section{Formal Identification Results}\label{sec:technical}

In this section, we provide additional details on the paper's formal analysis.

\subsubsection*{Notation Remark}

For random vectors $A$ and $B$, define $A^{\perp B} \coloneqq A - [\var(B)^{-1} \cov(B,A)]' B$. This is the sum of the residual from a linear projection of $A$ onto $(1,B)$ and the intercept in that projection.

\subsection{Generalization of \cite{Oster2019} Proposition 2}

The following result formally states our theorem deriving the identified set of $\beta_\text{long}$, a generalization of Oster's \citeyearpar{Oster2019} Proposition 2.

\begin{theorem}\label{thm:IDset}
Suppose the joint distribution of $(Y,X,W_1)$ is known. Suppose \ref{assn:posdef}--\ref{assn:exog} hold. Suppose $R_\text{long}^2 \in (R_\text{med}^2,1]$ is known. Then the identified set for $\beta_\text{long}$ is
\[
	\mathcal{B}_I(\delta,R_\text{long}^2) = \{ b \in \R: f(\beta_\text{med} - b,\delta,R_\text{long}^2) = 0 \} \setminus \tildeB
	\tag{\ref{eq:ID_set_expression}}
\]
where
\[
	\tildeB \coloneqq \{b \in \R: \gamma_{1,\text{med}} + (\beta_\text{med} - b)\pi_1 = 0\}
\]
and $f(\bias,\delta, R_\text{long}^2) \coloneqq f_0(\bias) + \delta \cdot f_1(\bias, R_\text{long}^2)$ with
\begin{align*}
	f_0(\bias) &\coloneqq -\bias\var(X^{\perp W_1})\left(\var(\gamma_{1,\text{med}}'W_1) + 2\bias\cov(X,\gamma_{1,\text{med}}'W_1) + \bias^2\var(\pi_1'W_1)\right) \\[1em]
	f_1(\bias,R_\text{long}^2)
	&\coloneqq (R_\text{long}^2 - R_\text{med}^2)\var(Y) \cov(X, \gamma_{1,\text{med}}' W_1)
	+ \bias (R_\text{long}^2 - R_\text{med}^2) \var(Y) \var(\pi_1'W_1) \\
	&\qquad + \bias^2 \var(X^{\perp W_1})\cov(X,\gamma_{1,\text{med}}'W_1)
	+ \bias^3\var(X^{\perp W_1})\var(\pi_1'W_1).
\end{align*}
\end{theorem}

There are several differences between Theorem \ref{thm:IDset} and Oster's \citeyearpar{Oster2019} Proposition 2. First, we do not impose the normalization that $\var(W_1)$ is diagonal. Second, Oster's exogenous controls assumption is stated as $\cov(W_1, \gamma_{2,\text{long}}' W_2) = 0$, which is equivalent to our \ref{assn:exog} given that we are assuming for simplicity $W_2$ is scalar throughout the analysis, and since $\gamma_{2,\text{long}} \neq 0$ by \ref{assn:nonzero}.

Third, and most importantly, we show the identified set equals the solutions to the cubic equation except those that lie in the set $\tildeB$, which was not present in \cite{Oster2019}. Here we explain how to interpret $\tildeB$ and how to characterize when it is empty. By Lemma \ref{lemma:regfacts} in Appendix \ref{sec:proofs},
\begin{equation}\label{eq:osterA2eq}
	\gamma_{1,\text{long}} = \gamma_{1,\text{med}} + (\beta_\text{med} - \beta_\text{long}) \pi_1.
\end{equation}
So $\tildeB$ is the set of $\beta_\text{long}$ values which would imply $\gamma_{1,\text{long}} = 0$, which violates \ref{assn:nonzero}; recall that $\delta$ is not defined when \ref{assn:nonzero} fails. This explains why such values must be removed. Next, notice that $\tildeB$ is nonempty if and only if there exists a $C \in \R$ such that $\gamma_{1,\text{med}} = C \cdot \pi_1$: see Proposition \ref{prop:checkingOsterA2} below. That is, if and only if $\pi_1$ and $\gamma_{1,\text{med}}$ are proportional to each other. When $W_1$ is a vector, as is usually the case in empirical practice, this proportionality condition will typically not hold, and hence $\tildeB$ will be empty. When $W_1$ is a scalar, however, $\pi_1$ and $\gamma_{1,\text{med}}$ are scalars and hence are always proportional when $\pi_1 \neq 0$. In this case, the singleton $\tildeB = \{ \beta_\text{med} + \gamma_{1,\text{med}} / \pi_1 \}$ is not in the identified set. While the appearance of the set $\tildeB$ in Theorem \ref{thm:IDset} is mostly technical, it is nonetheless important, for example, when proving that the identified set is a singleton under the stronger assumptions in Proposition \ref{prop:OsterP1} below.

Theorem \ref{thm:IDset} assumed $R_\text{long}^2 > R_\text{med}^2$. The following result explains why: When $R_\text{long}^2 = R_\text{med}^2$ the omitted variable bias is zero, as shown in the following proposition.

\begin{proposition}\label{prop:R2longEqualR2med}
Suppose Assumption \ref{assn:posdef} holds. Suppose $R_\text{long}^2 = R_\text{med}^2$. Then $\beta_\text{long} = \beta_\text{med}$.
\end{proposition}

Finally, the analysis simplifies to the baseline case when $\delta = 0$.

\begin{corollary}\label{cor:baselineCase}
Suppose assumptions \ref{assn:posdef}--\ref{assn:exog} hold. Suppose $\delta = 0$ and $\gamma_{1,\text{med}} \neq 0$. Then $\beta_\text{long} = \beta_\text{med}$.
\end{corollary}

\subsection{Adding Assumptions to Point Identify $\beta_\text{long}$}\label{sec:BiasCorrectionsAppendix}

Equation \eqref{eq:ID_set_expression} generally does not point identify $\beta_\text{long}$, which means there is not a unique way to compute a bias adjusted estimand. However, as Oster shows, the identified set simplifies to a singleton under two additional assumptions. The first assumption is $\delta = 1$. The second assumption is the following, which is equivalent to Oster's Assumption 2 (page 192).\footnote{It is unclear to us whether Oster's \citeyearpar{Oster2019} Assumption 2 is meant to be an assumption relating $\gamma_{1,\text{med}}$ to $\pi_1$ or relating $\gamma_{1,\text{long}}$ to $\pi_1$ (using our notation), since that paper uses the same notation $\Psi$ to denote both $\gamma_{1,\text{long}}$ and $\gamma_{1,\text{med}}$ in different places (for example, compare her equation (1) with the definition of $\psi_i$ in the statement of her Assumption 2). In general, $\gamma_{1,\text{med}} \neq \gamma_{1,\text{long}}$. Regardless, our Proposition \ref{prop:checkingOsterA2} shows that the two forms of this assumption are equivalent under the exogenous controls assumption \ref{assn:exog}. Also note that \cite{Oster2019} does not discuss the testability of this assumption.}

\begin{Aassumption}\label{assump:OsterA2}
There is a constant $C \in \R$ such that $\gamma_{1,\text{long}} = C \cdot \pi_1$.
\end{Aassumption}

Recall that $\pi_1$ denotes the coefficient on $W_1$ in the regression of $X$ on $(1,W_1)$. \ref{assump:OsterA2} holds immediately when $W_1$ is a scalar and $\pi_1 \neq 0$, but otherwise is a restriction. In fact, the following result implies that \ref{assump:OsterA2} is refutable and confirmable (as defined in \citealt{Breusch1986}) under exogenous controls. Hence the data alone tells us whether it holds.

\begin{proposition}\label{prop:checkingOsterA2}
Suppose \ref{assn:posdef} and \ref{assn:exog} hold. Then \ref{assump:OsterA2} holds if and only if there is a constant $C_\text{med} \in \R$ such that $\gamma_{1,\text{med}} = C_\text{med} \cdot \pi_1$.
\end{proposition}

For vector $W_1$, \cite{Oster2019} said \ref{assump:OsterA2} ``is very unlikely to hold except in pathological cases'' (page 192). Proposition \ref{prop:checkingOsterA2} shows that researchers do not need to guess about this assumption---it can be directly tested in the data. Such a hypothesis test can be constructed using standard minimum distance methods, by comparing $\widehat{\gamma}_{1,\text{med}}$ with multiples of $\widehat{\pi}_1$; we omit a full development for brevity.

Under these two additional assumptions, $\delta = 1$ and \ref{assump:OsterA2}, we can now show that $\beta_\text{long}$ is point identified.

\begin{proposition}\label{prop:OsterP1}
Suppose the joint distribution of $(Y,X,W_1)$ is known, \ref{assn:posdef}, \ref{assn:nonzero}, \ref{assn:exog}, and \ref{assump:OsterA2} hold, \ref{assn:delta} holds with $\delta = 1$, $R_\text{long}^2 \in (R_\text{med}^2,1]$ is known, and $\beta_\text{short} \neq \beta_\text{med}$. Then $\beta_\text{long}$ is point identified and
\[
	\beta_\text{long} =
	\beta_\text{med} + (\beta_\text{med} - \beta_\text{short}) \frac{R_\text{long}^2 - R_\text{med}^2}{R_\text{med}^2 - R_\text{short}^2}.
	\tag{\ref{eq:OsterP1}}
\]
\end{proposition}

This result is a version of Oster's Proposition 1. Oster's \citeyearpar{Oster2019} proof is slightly incomplete, however, since it only shows that the value in equation \eqref{eq:OsterP1} is consistent with the data. It does not show that this is the only value that is consistent with the data. Our proof of Proposition \ref{prop:OsterP1} shows this last step. In this proof we show how the result follows directly from our Theorem \ref{thm:IDset}. The cubic function $f(\beta_\text{med} - b,\delta,R_\text{long}^2)$ is quadratic with two solutions when $\delta = 1$  and one of the two solutions is eliminated by subtracting $\tildeB$. The remaining solution is equation \eqref{eq:OsterP1}. Finally, in this result we assume $\beta_\text{short} \neq \beta_\text{med}$ to ensure that $R_\text{short}^2 \neq R_\text{med}^2$ and hence that the right hand side of \eqref{eq:OsterP1} is well defined.

\subsubsection*{An Alternative Assumption}

Proposition \ref{prop:OsterP1} showed that $\beta_\text{long}$ takes on a single value when $\delta = 1$ and \ref{assump:OsterA2} holds. As mentioned earlier, Oster says \ref{assump:OsterA2} is ``unlikely to hold''. Without \ref{assump:OsterA2}, there are usually two values of $\beta_\text{long}$ consistent when $\delta = 1$. Oster briefly suggests an alternative to \ref{assump:OsterA2} that also implies a single value of $\beta_\text{long}$ is consistent with $\delta = 1$. This alternative, her Assumption 3 on page 194, can be stated as follows.

\begin{Aassumption}\label{assump:OsterA3}
$\cov(X,\gamma_{1,\text{med}}'W_1)$ and $\cov(X,\gamma_{1,\text{long}}'W_1)$ have the same sign.
\end{Aassumption}

In words, the omitted variable bias does not change the sign of the covariance between the treatment and the index of the observables across the medium and long regressions. Using equations \eqref{eq:lemma1-1} and \eqref{eq:lemma1-3} below, we can rewrite \ref{assump:OsterA3} as
\begin{align*}
	\sign(\beta_\text{short} - \beta_\text{med}) &= \sign\left(\frac{\beta_\text{short} - \beta_\text{med}}{R^2_{X \sim W_1}} + \beta_\text{med} - \beta_\text{long}\right).
\end{align*}
This restriction is therefore a simple inequality assumption on the omitted variable bias:
\begin{align*}
	\beta_\text{med} - \beta_\text{long} &\geq \frac{\beta_\text{med} - \beta_\text{short}}{R^2_{X \sim W_1}}
	\qquad \text{when} \quad \beta_\text{med} - \beta_\text{short} \leq 0,\\
	\beta_\text{med} - \beta_\text{long} &\leq \frac{\beta_\text{med} - \beta_\text{short}}{R^2_{X \sim W_1}}
	\qquad \text{when} \quad \beta_\text{med} - \beta_\text{short} \geq 0.
\end{align*}
Thus, \ref{assump:OsterA3} together with an assumed value for $\delta$ (not necessarily 1), will rule out certain values for the omitted variable bias. Consequently, the identified set for $\beta_\text{long}$ will generally have fewer elements when \ref{assump:OsterA3} is imposed. Since $(\beta_\text{short}, \beta_\text{med}, R^2_{X \sim W_1})$ are all point identified, this assumption can also be easily implemented in practice.

\subsection{Properties of R-squared}

Here we state without proof five properties of the R-squared and partial R-squared that we use in Section \ref{sec:intuitionAppendix}, but which are not necessarily widely known. Let $A$ be a random variable and let $B$ and $C$ be random vectors, all with finite second moments. Let $\var(A,B,C)$ be positive definite.
\begin{align}
	R^2_{A \sim B \sbullet C} &= \frac{R^2_{A \sim B,C} - R^2_{A \sim C}}{1 - R^2_{A \sim C}}\label{eq:R2_1}\\
	\var(A^{\perp C}) &= \var(A)(1 - R^2_{A \sim C})\label{eq:R2_2}\\
	R^2_{A \sim B,C} &= R^2_{A \sim B^{\perp C}} + R^2_{A \sim C}\label{eq:R2_3}\\
	R^2_{A \sim B, C} &= R^2_{A \sim B} + R^2_{A \sim C} \text{ when } \cov(B,C) = 0\label{eq:R2_4}\\
	\var(A)(R^2_{A \sim B,C} - R^2_{A \sim C}) &= \var(\gamma_B'B^{\perp C}),\label{eq:R2_5}
\end{align}
where $\gamma_B$ is the coefficient on $B$ in the linear regression of $A$ on $(1,B,C)$.

\subsection{Formal Results on the Asymptote at $\delta = 1$}\label{sec:intuitionAppendix}

We now provide a formal justification for our intuitive explanation behind the presence of the asymptote at $\delta = 1$ in Section \ref{sec:deltaSignChangeBP}. We show that, for fixed $R^2_\text{long} \in (R^2_\text{med},1]$, the following two conditions are \emph{equivalent}:
\begin{enumerate}
\item $(X,W_1,W_2)$ are nearly multicollinear.

\item $|\beta_\text{long}|$ is arbitrarily large and $\delta$ is arbitrarily close to 1.
\end{enumerate}
To formalize this statement, we consider a sequence of data generating processes for $(Y,X,W_1,W_2)$ indexed by $j \in \{1,2,\ldots\}$. Since our analysis only depends on the covariance matrix of these variables, we will work with the indexed covariance matrix $\Sigma_j \coloneqq \varj(Y,X,W_1,W_2)$. Assume that dgps along this sequence are consistent with $\var(Y,X,W_1)$, the known variance matrix of observed variables. We index any aspect of the joint distribution of $(Y,X,W_1,W_2)$ that depends on $W_2$ with $j$ to indicate that their value depends on $\Sigma_j$. For example, we let
\begin{align*}
	\delta_j &\coloneqq \frac{\covj(X,\gamma_{2,\text{long},j}W_2)}{\varj(\gamma_{2,\text{long},j}W_2)} \bigg/ \frac{\cov(X,\gamma_{1,\text{long},j}'W_1)}{\var(\gamma_{1,\text{long},j}'W_1)} .
\end{align*}
Note that the covariance and variance terms which do not involve $W_2$ are not indexed by $j$ since they are taken over the distribution of $\var(Y,X,W_1)$, which we assume does not vary with $j$.

We formalize the idea of ``near multicollinearity'' between $X$, $W_1$, and $W_2$ as $R^2_{X \sim W_1,W_2,j} \rightarrow 1$ as $j\rightarrow \infty$. This means that a linear combination of $W_1$ and $W_2$ approaches $X$ in a mean-square sense as $j \rightarrow \infty$. The notions of ``$|\beta_\text{long}|$ is arbitrarily large'' and ``$\delta$ is arbitrarily close to 1'' are formalized as $(|\beta_{\text{long},j}|,\delta_j) \rightarrow (+\infty,1)$ as $j\rightarrow\infty$.

\begin{proposition}\label{prop:formal_asym_intuition}
	Fix $\var(Y,X,W_1)$ and $R^2_\text{long} \in ( R^2_\text{med},1]$. For $j \in \{1,2,\ldots\}$, let $\Sigma_j \coloneqq \varj(Y,X,W_1,W_2)$ be a sequence of variance matrices consistent with $\var(Y,X,W_1)$, assumptions \ref{assn:posdef}, \ref{assn:nonzero}, and \ref{assn:exog}, and with $R^2_{Y \sim X,W_1,W_2,j} = R^2_\text{long}$.  Then, as $j\rightarrow\infty$,
\begin{align*}
	R_{X \sim W_1, W_2,j}^2 \rightarrow 1
	\quad \text{ if and only if } \quad (|\beta_{\text{long},j}|,\delta_j) \rightarrow (+\infty,1).
\end{align*}
\end{proposition}

\begin{proof}[Proof of Proposition \ref{prop:formal_asym_intuition}]
$(\Rightarrow)$. Suppose $R_{X \sim W_1, W_2,j}^2 \rightarrow 1$. By equation \eqref{eq:R2_2}, this is equivalent to $\varj(X^{\perp W_1,W_2})$ $\rightarrow 0$. Then
\begin{align*}
	X^{\perp W_1,W_2}
	&= X - \pi_{1,j}'W_1 - \pi_{2,j} W_2 \\
	&= X - \pi_1'W_1 - \pi_{2,j} W_2.
\end{align*}
The second equality follows from $\pi_1 = \var(W_1)^{-1}\cov(W_1,X)$ and \ref{assn:exog}. So
\begin{align*}
	\varj(X^{\perp W_1,W_2}) &= \var(X) - \var(\pi_1'W_1) - \varj(\pi_{2,j}W_2).
\end{align*}	
Hence $\varj(X^{\perp W_1,W_2}) \rightarrow 0$ implies $\varj(\pi_{2,j}W_2) \rightarrow \var(X) - \var(\pi_1'W_1) = \var(X^{\perp W_1})$, a fixed positive quantity by \ref{assn:posdef}. Next,
\begin{align}
	\var(Y)(R^2_\text{long} - R^2_\text{med})
	&= \varj(\gamma_{2,\text{long},j}W_2^{\perp X,W_1})\notag\\
	&= \varj(\gamma_{2,\text{long},j}W_2)(1 - R^2_{W_2 \sim X,W_1,j})\notag\\
	&= \varj(\gamma_{2,\text{long},j}W_2)\left(1 - R^2_{X \sim W_2 \sbullet W_1,j}\right)\notag\\
	&= \varj(\gamma_{2,\text{long},j}W_2)\left(1 - \frac{\varj(\pi_{2,j}W_2)}{\var(X^{\perp W_1})}\right).\label{eq:asy_proof_eq1}
\end{align}
The first line follows by \eqref{eq:R2_5}. The second by \eqref{eq:R2_2}. The third follows from
\begin{align*}
	R^2_{W_2 \sim X,W_1,j} 
	&= R^2_{W_2 \sim X \sbullet W_1,j}(1 - R^2_{W_2 \sim W_1,j}) + R^2_{W_2 \sim W_1,j} \\
	&= R^2_{W_2 \sim X \sbullet W_1,j} \\
	&= R^2_{X \sim W_2 \sbullet W_1,j}
\end{align*}
where we used \eqref{eq:R2_1} in the first equality and $\cov_{j}(W_1,W_2) = 0$ in the second. The third equality here uses $R_{A \sim B \sbullet C}^2 = R_{B \sim A \sbullet C}^2$ for scalar $A$ and $B$, which we can apply since $W_2$ is a scalar. The last equality in \eqref{eq:asy_proof_eq1} follows from 
\begin{align*}
	R^2_{X \sim W_2 \sbullet W_1,j} &= \frac{R^2_{X \sim W_1,W_2,j} - R^2_{X \sim W_1}}{1 - R^2_{X \sim W_1}} = \frac{R^2_{X \sim W_2,j} }{1 - R^2_{X \sim W_1}} = \frac{\varj(\pi_{2,j}W_2)}{\var(X^{\perp W_1})}.
\end{align*}
where we used \eqref{eq:R2_1}, \eqref{eq:R2_4} combined with \ref{assn:exog}, and then \eqref{eq:R2_3} combined with the definition of R-squared. 

Equation \eqref{eq:asy_proof_eq1} can then be used to show $\varj(\gamma_{2,\text{long},j}W_2) \rightarrow \infty$. This is the case because
\begin{align*}
	\varj(\gamma_{2,\text{long},j}W_2) &= \frac{\var(Y)(R^2_\text{long} - R^2_\text{med})}{1 - \varj(\pi_{2,j}W_2)/\var(X^{\perp W_1})}
\end{align*}
and $\var(Y)(R^2_\text{long} - R^2_\text{med}) > 0$ is fixed (with respect to $j$), while $\varj(\pi_{2,j}W_2) \rightarrow \var(X^{\perp W_1})$, implying that $1 - \varj(\pi_{2,j}W_2)/\var(X^{\perp W_1}) \rightarrow 0$.

By \eqref{eq:lemma1-2} we can write the squared omitted variable bias as
\[
	(\beta_{\text{long},j} - \beta_\text{med})^2
	= \frac{\varj(\pi_{2,j}W_2) \varj(\gamma_{2,\text{long},j}W_2)}{\var(X^{\perp W_1})^2}.
\]
Since $\varj(\pi_{2,j}W_2) \rightarrow \var(X^{\perp W_1}) > 0$ and $\varj(\gamma_{2,\text{long},j}W_2) \rightarrow \infty$, we have that $(\beta_{\text{long},j} - \beta_\text{med})^2 \rightarrow \infty$. Since $\beta_\text{med}$ is fixed, we must have $|\beta_{\text{long},j}| \rightarrow +\infty$.

We now show $\delta_j \rightarrow 1$. We can write
\begin{align*}
	\delta_j &= \frac{\covj(X,\gamma_{2,\text{long},j}W_2)}{\varj(\gamma_{2,\text{long},j}W_2)} \bigg/ \frac{\cov(X,\gamma_{1,\text{long},j}'W_1)}{\var(\gamma_{1,\text{long},j}'W_1)}\\
	&= \frac{\pi_{2,j}}{\gamma_{2,\text{long},j}} \frac{\gamma_{1,\text{long},j}}{\pi_1}\\
	&= \frac{\pi_{2,j}}{\gamma_{2,\text{long},j}} \frac{\gamma_{1,\text{med}} + (\beta_\text{med} - \beta_{\text{long},j})\pi_1} {\pi_1}\\
	&= \frac{\pi_{2,j}}{\gamma_{2,\text{long},j}} \frac{\gamma_{1,\text{med}}}{\pi_1} + \frac{\pi_{2,j}}{\gamma_{2,\text{long},j}}\frac{\pi_{2,j}\gamma_{2,\text{long},j}\var_j(W_2)}{\var(X^{\perp W_1})}\\
	&= \frac{\pi_{2,j}}{\gamma_{2,\text{long},j}} \frac{\gamma_{1,\text{med}}}{\pi_1} + \frac{\var_j(\pi_{2,j}W_2)}{\var(X^{\perp W_1})}.
\end{align*}
We used Assumption \ref{assn:exog} for the second equality, equation \eqref{eq:lemma1-1} for the third, equation \eqref{eq:lemma1-2} for the fourth, and gathered terms for the fifth. The third and fourth lines here are where we use the fact that $\gamma_{1,\text{long},j}$ and $\gamma_{2,\text{long},j}$ diverge at the same rate, as discussed in Section \ref{sec:deltaSignChangeBP}.

As derived earlier, $\varj(\pi_{2,j}W_2)/\var(X^{\perp W_1})\rightarrow 1$ as $j \rightarrow \infty$. Also, $\pi_{2,j}/\gamma_{2,\text{long},j} \rightarrow 0$ because
\begin{align*}
	\left|\frac{\pi_{2,j}}{\gamma_{2,\text{long},j}}\right| &= \left(\frac{\varj(\pi_{2,j}W_2)}{\varj(\gamma_{2,\text{long},j}W_2)}\right)^{1/2} \rightarrow \left(\frac{\var(X^{\perp W_1})}{\infty}\right)^{1/2} = 0
\end{align*}
as $j\rightarrow \infty$. We conclude that $\delta_j \rightarrow 1$. 

\bigskip

$(\Leftarrow)$ Suppose $|\beta_{\text{long},j}| \rightarrow +\infty$ and
\begin{align*}
	\delta_j &= \frac{\covj(X,\gamma_{2,\text{long},j}W_2)}{\varj(\gamma_{2,\text{long},j}W_2)} \bigg/ \frac{\cov(X,\gamma_{1,\text{long},j}'W_1)}{\var(\gamma_{1,\text{long},j}'W_1)} \rightarrow 1
\end{align*}
as $j\rightarrow\infty$. By equation \eqref{eq:lemma1-1}, we have that $\gamma_{1,\text{long},j} = \gamma_{1,\text{med}} + (\beta_{\text{long},j} - \beta_\text{med}) \pi_1$. By equation \eqref{eq:asy_proof_eq1} above and \eqref{eq:lemma1-2}, we have 
\begin{align*}
	\varj(\gamma_{2,\text{long},j}W_2) &= \var(Y)(R^2_\text{long} - R^2_\text{med}) + \frac{\varj(\pi_{2,j}W_2) \varj(\gamma_{2,\text{long},j}W_2)}{\var(X^{\perp W_1})^2} \var(X^{\perp W_1})\\
	&= \var(Y)(R^2_\text{long} - R^2_\text{med}) + (\beta_{\text{long},j} - \beta_\text{med})^2 \var(X^{\perp W_1}).
\end{align*}
Since $(\beta_{\text{long},j} - \beta_\text{med})^2 \rightarrow \infty$ and $\var(X^{\perp W_1}) > 0$, we have that $\varj(\gamma_{2,\text{long},j}W_2) \rightarrow \infty$. Again by equation \eqref{eq:asy_proof_eq1}, we have that
\begin{align*}
	1 - \frac{\varj(\pi_{2,j}W_2)}{\var(X^{\perp W_1})} &= \frac{\var(Y)(R^2_\text{long} - R^2_\text{med})}{\varj(\gamma_{2,\text{long},j}W_2)} \rightarrow 0
\end{align*}
and therefore $\varj(\pi_{2,j}W_2) \rightarrow \var(X^{\perp W_1})$. This implies that
\begin{align*}
	\varj(X^{\perp W_1, W_2}) &= \var(X^{\perp W_1}) - \varj(\pi_{2,j}W_2) \rightarrow 0,
\end{align*}
as desired.
\end{proof}

\subsection{OLS with nearly multicollinear covariates}\label{sec:OLSmultiCollinear}

In this section we analyze the behavior of OLS estimands as the distribution of the covariates approaches perfect multicollinearity. The main result is that any given coefficient can converge to any value in $\R \cup \{ -\infty, + \infty \}$.

Let $(Y,X_1,X_2)$ be scalar random variables with distribution indexed by $j \in \{1,2,\ldots\}$. Assume
\[
	\var(X_1,X_2)
	= \begin{pmatrix}
		1 & \rho_j\\
		\rho_j & 1
	\end{pmatrix}
\]
where $-1 \leq \rho_j \leq 1$ is the correlation. Let $\sigma_{1,j} \coloneqq \cov(Y,X_1)$ and $\sigma_{2,j} \coloneqq \cov(Y,X_2)$. We will consider sequences with $\rho_j \nearrow 1$ so that both of these covariances converge to a common value. Thus, in the limit, the univariate OLS estimands $\cov(Y,X_1) / \var(X_1)$ and $\cov(Y,X_2) / \var(X_2)$ are well defined and equal to each other.

When $|\rho_j |<1$, the projection coefficients in the linear projections of $Y$ on $(1,X_1,X_2)$ are well defined and equal
\begin{align*}
	\begin{pmatrix}
		\beta_{1,j}\\
		\beta_{2,j}
	\end{pmatrix}
	&\coloneqq \begin{pmatrix}
		1 & \rho_j\\
		\rho_j & 1
	\end{pmatrix}^{-1}\begin{pmatrix}
		\sigma_{1,j}\\
		\sigma_{2,j}
	\end{pmatrix}
	= \frac{1}{1-\rho_j^2}
	\begin{pmatrix}
		1 & -\rho_j\\
		-\rho_j & 1
	\end{pmatrix}\begin{pmatrix}
		\sigma_{1,j}\\
		\sigma_{2,j}
	\end{pmatrix}
	= \begin{pmatrix}
		\dfrac{\sigma_{1,j} - \rho_j \sigma_{2,j}}{1-\rho_j^2}\\[2em]
		\dfrac{\sigma_{2,j} - \rho_j \sigma_{1,j}}{1-\rho_j^2}
	\end{pmatrix}.
\end{align*}
For $\sigma \in \R$, let
\[
	\mathcal{B}(\sigma) \coloneqq \{(b_1,b_2) \in \R^2: b_1 + b_2 = \sigma\} \cup \{(-\infty,+\infty),(+\infty,-\infty)\}.
\]

\begin{proposition}\label{prop:multicollinearOLS}
For all $\sigma \in \R$, for each $(\beta_1,\beta_2) \in \mathcal{B}(\sigma)$ there exist sequences $\rho_j \nearrow 1$, $\sigma_{1,j} \to \sigma$, $\sigma_{2,j} \to \sigma$ as $j \rightarrow \infty$ such that $(\beta_{1,j},\beta_{2,j}) \rightarrow (\beta_1,\beta_2)$.
\end{proposition}

\subsection{Remark on the definition of a sign change}

The sign of zero is often not defined. Our definition of the sign change breakdown point $\delta^{\text{bp},\text{sign}}(R_\text{long}^2)$ in Section \ref{sec:deltaSignChangeBP} makes the convention that a change from a nonzero value to 0 is considered a sign change. We could therefore more precisely call $\delta^{\text{bp},\text{sign}}(R_\text{long}^2)$ the \emph{weak} sign change breakdown point. Under this convention, $\delta^{\text{bp},\text{sign}}(R_\text{long}^2)$ will always be weakly smaller than the explain away breakdown point $\delta^{\text{bp},\text{explain away}}(R_\text{long}^2)$.

We could alternatively define the \emph{strict} sign change breakdown point as requiring the coefficient to actually cross the axis, rather than just touch it. Except in knife-edge cases, the strict and weak sign change breakdown points for $\delta$ will be the same. This can be seen from the equation for the explain away breakdown point in Theorem \ref{thm:OsterProp3}, which is the ratio of two cubic polynomials in $\beta_\text{hypo}$. This ratio will generally be strictly monotonic in a neighborhood of $\beta_\text{hypo} = 0$, implying that the strict sign change breakdown point cannot be larger than the explain away breakdown point.

Finally, note that for other types of sensitivity parameters besides $\delta$, it is in principle possible for the strict and weak sign change breakdown points to be different. See Section 2 of our older draft of this paper (\citealt{MastenPoirier2022wp}) for a discussion of sign change breakdown points in general models.

\section{Some Common Pitfalls}\label{sec:commonPitfallsAppendix}

In this section we discuss a common mistake some researchers make when computing explain away breakdown points, along with a remark about the ``bounding set'' defined in \cite{Oster2019}.

\subsection{A Common Mistake when Computing Explain Away Breakdown Points}\label{sec:commonMistake}

{\sloppy
As pointed out in Footnote 1 in \cite{DMP2025ident}, several papers compute the explain away breakdown point $\delta^{\text{bp},\text{explain away}}(R_\text{long}^2)$ incorrectly (the correct value is given by Theorem \ref{thm:OsterProp3}). In this section we describe the source of this mistake. It can be avoided by using software based on the correct results, such as Oster's package \texttt{psacalc} or our package \texttt{regsensitivity} in Stata.

Proposition \ref{prop:OsterP1} shows how to point identify $\beta_\text{long}$ if $\delta = 1$, $R_\text{long}^2$ is known, and \ref{assump:OsterA2} is assumed. The identification result in Theorem \ref{thm:IDset} is strictly more general since it allows for any value of $\delta$ and does not use \ref{assump:OsterA2}. Because this general theorem holds for any $\delta$, it should be used for breakdown analysis, as we did earlier. However, on page 193 \cite{Oster2019} states the following equation:
\begin{equation}\label{eq:OsterP1approx}
	\beta_\text{long} \approx
	\beta_\text{med} + \delta (\beta_\text{med} - \beta_\text{short}) \frac{R_\text{long}^2 - R_\text{med}^2}{R_\text{med}^2 - R_\text{short}^2}
\end{equation}
where $\delta$ is no longer necessarily equal to one. We are not aware of any formal results that justify this equation---\cite{Oster2019} does not provide any and does not explain the nature of the approximation $\approx$ in this equation. To the contrary, as we explained earlier, Theorem \ref{thm:IDset} implies that $\beta_\text{long}$ can be arbitrarily far from its value at $\delta = 1$ even for $\delta$ values arbitrarily close to 1. Moreover, even at $\delta = 1$, $\beta_\text{long}$ is not necessarily point identified when \ref{assump:OsterA2} is not assumed. Consequently, equation \eqref{eq:OsterP1approx} should not be used to perform sensitivity analysis (except when $\delta$ is exactly equal to 1 and under the other conditions of Proposition \ref{prop:OsterP1}).

Several researchers, however, have apparently used equation \eqref{eq:OsterP1approx} to compute breakdown points. This includes \citet[\emph{JPE}]{SatyanathVoigtlanderVoth2017} and \citet[\emph{ECMA}]{BFG2020}. Specifically, they fix $R_\text{long}^2$ and solve
\[
	\beta_\text{med} + \delta (\beta_\text{med} - \beta_\text{short}) \frac{R_\text{long}^2 - R_\text{med}^2}{R_\text{med}^2 - R_\text{short}^2} = 0
\]
for $\delta$. This solution does not equal the explain away breakdown point $\delta^{\text{bp},\text{explain away}}(R_\text{long}^2)$. Note that this is true even if \ref{assump:OsterA2} holds---Proposition \ref{prop:OsterP1} is only about the case $\delta = 1$ and does not describe the bias when $\delta \neq 1$.

Equation \eqref{eq:OsterP1approx} also explains why some papers do not report the magnitude of their computed $\delta^{\text{bp},\text{explain away}}$ when it is negative. For example, \cite{SatyanathVoigtlanderVoth2017} report ``[$<$ 0]'' in their Table A.27 results rather than providing the exact magnitude of negative estimates. In the table footnote, they say
\begin{quote}
``The entry [$<$0] indicates that the respective Altonji et al. ratio is negative... suggesting a downward bias for our OLS estimates due to unobservables''. (online appendix page 44)
\end{quote}
If equation \eqref{eq:OsterP1approx} were correct for all $\delta$, then the sign of the omitted variable bias would be determined by the sign of $\delta$ \emph{and} $\beta_\text{med} - \beta_\text{short}$. 
So it is incorrect to conclude that the true bias $\beta_\text{med} - \beta_\text{long}$ is negative whenever this breakdown point is negative. Regardless, this line of reasoning is all based on equation \eqref{eq:OsterP1approx} which, as we explained above, generally does not hold except when $\delta$ is exactly equal to 1 and under the other conditions of Proposition \ref{prop:OsterP1}.

\subsection{Remark on Oster's ``Bounding Set''}\label{sec:ostersBoundingSet}

Based on the discussion on pages 194--196, \cite{Oster2019} defines
\[
	\beta^*(\delta, R_\text{long}^2) = \argmin_{b \in \mathcal{B}_I(\delta,R_\text{long}^2)} \; | \beta_\text{med} - b |.
\]
Using this, she defines the set $[\beta_\text{med}, \beta^*(\delta, R_\text{long}^2)]$ (or the reverse if $\beta_\text{med}$ is the larger value). She calls this a ``bounding set'' on page 196, but elsewhere calls it the ``identified set'' (for example, see her Table 3). The discussion on page 196 suggests that this interval should be interpreted as the identified set for $\beta_\text{long}$ under the assumption that $0 \leq \delta \leq 1$. As we have shown, this is not the correct interpretation. This follows because the identified set $\mathcal{B}_I(\delta,R_\text{long}^2)$ for a known $\delta$ is not monotonically increasing (in the set inclusion order) in $\delta$. Note that this is generally true even if we impose the magnitude constraint \ref{assn:magnRest}. The correct identified set under the assumption that $0 \leq \delta \leq 1$ is given by $\overline{\mathcal{B}}_I(\bar{\delta},R_\text{long}^2)$ from Corollary \ref{cor:cum_IDset} with $\bar{\delta} = 1$. The left plot in Figure \ref{fig:osterCumulative} shows an example of this set; the right plot shows its convex hull. We recommend that researchers present this set rather than the ``bounding set'' described above.

Finally, note that some authors (for example, \citealt{MianSufi2014}, online appendix page 2) define $\beta^*(\delta, R_\text{long}^2)$ as the value from equation \eqref{eq:OsterP1}, from Oster's Proposition 1. This is equivalent to the definition above when $\delta = 1$ and \ref{assump:OsterA2} holds, but otherwise they are different. Regardless, this approach also leads to an interval that does not have the correct interpretation as an identified set.

\section{Discussion of Estimation and Inference}\label{sec:estimation}

Our theoretical analysis in sections \ref{sec:OsterReview}--\ref{sec:delta_discussion} focuses on population level identification results. In practice, we only have finite sample data. The identified set \eqref{eq:ID_set_expression} only depends on the data through seven features of the data-generating process: $\beta_\text{med}$, $R_\text{med}^2$, $\beta_\text{short}$, $R_\text{short}^2$, $\var(Y)$, $\var(X)$, and $R_{X \sim W_1}^2$.
In our empirical analyses of Section \ref{sec:empirical} and Appendix \ref{sec:extraEmpirical_details}, as well as our meta-analysis in Section \ref{sec:metaAnalysis} and Appendix \ref{sec:extra_meta-analyses}, we estimate all breakdown points via plug in estimators by taking the sample analog of these seven parameters. 

A natural next question is to consider inference on these breakdown points (e.g., as in \citealt{MastenPoirier2019BF}). Our identification results, and the presence of the asymptote at $\delta = 1$ in particular, suggest that there may be substantial sampling uncertainty. For brevity we do not fully study this inferential question here, but keep in mind that all of our empirical analyses and our meta-analyses are based on point estimates. Given that the sign change breakdown points (including those that impose a magnitude restriction) are usually much closer to the $\delta = 1$ asymptote than the explain away breakdown points, we conjecture that adjusting for sampling uncertainty would further exacerbate the difference between the two types of breakdown points.

\section{Proofs}\label{sec:proofs}

Let $\Lin(A \mid B)$ denote the linear projection of the scalar random variable $A$ on the random column vector $B$, defined as $\Lin(A \mid B) = \Exp(A B') \Exp(BB')^{-1} B$.

\subsection*{Section \ref{sec:OsterReview}}

It is useful to also consider the linear projection of $X$ onto $(1,W_1,W_2)$:
\begin{equation}\label{eq:proj2}
	X = \pi_{1,\text{long}}'W_1 + \pi_{2,\text{long}} W_2 + X^{\perp W}. 
\end{equation}
Consider $\pi_1$, the coefficient on $W_1$ in the linear projection of $X$ on $(1,W_1)$. Then,
\begin{align*}
	\pi_1
	&= \var(W_1)^{-1}\cov(W_1,X) \\
	&= \var(W_1)^{-1}\cov(W_1,W_1'\pi_{1,\text{long}} + W_2 \pi_{2,\text{long}} + X^{\perp W})\\
	&= \var(W_1)^{-1}(\var(W_1)\pi_{1,\text{long}} + \cov(W_1,W_2)\pi_{2,\text{long}} + \cov(W_1,X^{\perp W}))\\
	&= \pi_{1,\text{long}}.
\end{align*}
The third line follows by the exogenous controls Assumption \ref{assn:exog}. Similarly, if we let $\pi_2$ denote the coefficient on $W_2$ from the linear projection of $X$ onto $(1,W_2)$ then $\pi_2 = \pi_{2,\text{long}}$. Given these results, we use $(\pi_1,\pi_2)$ to represent the common values of the regression coefficients throughout the proofs.

We denote the omitted variable bias by $\bias \coloneqq \beta_\text{med} - \beta_\text{long}$. We begin by showing the following identities that we will use throughout our proofs.

\begin{lemma}[Regression Algebra]\label{lemma:regfacts}
Suppose \ref{assn:posdef} and \ref{assn:exog} hold. Then the following hold:
\begin{align}
		\gamma_{1,\text{long}} &= \gamma_{1,\text{med}} + \bias \pi_1 \label{eq:lemma1-1}\\
		\pi_2 \gamma_{2,\text{long}}\var(W_2) &=\bias \var(X^{\perp W_1}) \label{eq:lemma1-2}\\
		\cov(X,\gamma_{1,\text{med}}'W_1) &= (\beta_\text{short} - \beta_\text{med})\var(X) \label{eq:lemma1-3}\\
		(R^2_\text{med} - R^2_\text{short})\var(Y) &= \var(\gamma_{1,\text{med}}'W_1) - \var(X)(\beta_\text{short} - \beta_\text{med})^2 \label{eq:lemma1-4}.
\end{align}
\end{lemma}

\begin{proof}[Proof of Lemma \ref{lemma:regfacts}]
\textbf{Proof of \eqref{eq:lemma1-1}}. We have that
\begin{align*}
	&\gamma_{1,\text{med}} + (\beta_\text{med} - \beta_\text{long}) \pi_1\\
	&= \var(W_1)^{-1}(\beta_\text{med} \cov(W_1,X) + \var(W_1) \gamma_{1,\text{med}} - \beta_\text{med} \cov(W_1,X)) + \bias \var(W_1)^{-1} \cov(W_1,X)\\
	&= \var(W_1)^{-1}(\cov(W_1,Y) - \beta_\text{med} \cov(W_1,X)) + \bias \var(W_1)^{-1} \cov(W_1,X)\\
	&= \var(W_1)^{-1}(\cov(W_1,Y) - \beta_\text{long}\cov(W_1,X))\\
	&= \var(W_1)^{-1}(\cov(W_1,\beta_\text{long} X + \gamma_{1,\text{long}}' W_1 + \gamma_{2,\text{long}} W_2 + Y^{\perp X,W}) - \beta_\text{long}\cov(W_1,X))\\
	&= \gamma_{1,\text{long}}.
\end{align*}
The second equality follows from 
\begin{align*}
	\cov(W_1,Y) &= \cov(W_1,\Lin(Y \mid 1,X,W_1) + Y^{\perp X,W_1})\\
	&= \cov(W_1,\beta_\text{med}X + \gamma_{1,\text{med}}'W_1)\\
	&= \beta_\text{med} \cov(W_1,X) + \var(W_1) \gamma_{1,\text{med}},
\end{align*}
while the final equality follows from \ref{assn:exog}.

\medskip

\textbf{Proof of \eqref{eq:lemma1-2}}. By the definition of $\beta_\text{med}$ and the FWL Theorem,
\begin{align*}
	\beta_\text{med}\var(X^{\perp W_1}) &= \cov(Y,X^{\perp W_1})\\
	&= \cov(\beta_\text{long} X + \gamma_{1,\text{long}}' W_1 + \gamma_{2,\text{long}} W_2 + Y^{\perp X,W},X^{\perp W_1})\\
	&= \beta_\text{long}\var(X^{\perp W_1}) + \gamma_{2,\text{long}} \cov(W_2,X^{\perp W_1}).
\end{align*}
To finish the proof, we note that $\cov(W_2,X^{\perp W_1}) = \cov(W_2,X - \pi_1'W_1) = \cov(W_2,\pi_2 W_2 + X^{\perp W}) = \pi_2\var(W_2)$ by \ref{assn:exog}. 

\medskip

\textbf{Proof of \eqref{eq:lemma1-3}}. By the definition of $\beta_\text{short}$,
\begin{align*}
	\beta_\text{short}\var(X) &= \cov(Y,X)\\
	&= \cov(\beta_\text{med} X + \gamma_{1,\text{med}}' W_1 + Y^{\perp X,W_1},X)\\
	&= \beta_\text{med} \var(X) + \cov(X,\gamma_{1,\text{med}}' W_1).
\end{align*}

\medskip

\textbf{Proof of \eqref{eq:lemma1-4}}. By the definitions of $R^2_\text{med}$ and $R^2_\text{short}$,
\begin{align*}
	&(R^2_\text{med} - R^2_\text{short})\var(Y) + \var(X)(\beta_\text{short} - \beta_\text{med})^2\\
	&= \var(X\beta_\text{med} + \gamma_{1,\text{med}}'W_1) - \var(X \beta_\text{short}) + \var(X)(\beta_\text{short} - \beta_\text{med})^2\\
	&= \var(X)(\beta_\text{med}^2 - \beta_\text{short}^2) + 2 \beta_\text{med} \cov(X,\gamma_{1,\text{med}}'W_1) + \var(\gamma_{1,\text{med}}'W_1) + \var(X)(\beta_\text{short} - \beta_\text{med})^2\\
	&= \var(X)(\beta_\text{med}^2 - \beta_\text{short}^2 + 2\beta_\text{med}(\beta_\text{short} - \beta_\text{med}) + (\beta_\text{short} - \beta_\text{med})^2) +  \var(\gamma_{1,\text{med}}'W_1) \\
	&= \var(\gamma_{1,\text{med}}'W_1).
\end{align*}
The third equality follows from equation \eqref{eq:lemma1-3}.
\end{proof}

We also use the following lemma to prove our main results.

\begin{lemma}[Oster's 3rd Equation]\label{lemma:osterEq3} 
Suppose assumptions \ref{assn:posdef}--\ref{assn:exog} hold. Then
\[
	(R_\text{long}^2 - R_\text{med}^2) \var(Y)\delta \cov(X, \gamma_{1,\text{long}}'W_1) = \bias \var(X^{\perp W_1}) \left( \var(\gamma_{1,\text{long}}'W_1)  - \bias \delta \cov(X, \gamma_{1,\text{long}}'W_1) \right).
\]
\end{lemma}

\begin{proof}[Proof of Lemma \ref{lemma:osterEq3}]

By definition,
\begin{align*}
	&R_\text{long}^2 \var(Y) \\
	&= \var(\beta_\text{long} X + \gamma_{1,\text{long}}' W_1 + \gamma_{2,\text{long}} W_2) \\
	&= \beta_\text{long}^2 \var(X) + \var(\gamma_{1,\text{long}}' W_1) + \var(\gamma_{2,\text{long}}W_2) \\
	&\qquad + 2 \cov(\beta_\text{long} X, \gamma_{1,\text{long}}'W_1) + 2 \cov(\beta_\text{long} X, \gamma_{2,\text{long}}W_2) + 2 \cov(\gamma_{1,\text{long}}'W_1, \gamma_{2,\text{long}}W_2) \\
	&= \beta_\text{long}^2 \var(X) + \var(\gamma_{1,\text{long}}' W_1) + \var(\gamma_{2,\text{long}}W_2) \\
	&\qquad + 2 \beta_\text{long} \cov(X, \gamma_{1,\text{long}}'W_1) + 2 \beta_\text{long} \cov(X,\gamma_{2,\text{long}}W_2) \\
	&= \beta_\text{long}^2 \var(X) + \var(\gamma_{1,\text{long}}'W_1) + \var(\gamma_{2,\text{long}}W_2) + 2 \beta_\text{long} \cov(X, \gamma_{1,\text{long}}'W_1) + 2 \beta_\text{long} \bias \var(X^{\perp W_1}).
\end{align*}
The third equality follows from Assumption \ref{assn:exog}. The fourth follows equation \eqref{eq:lemma1-2}. We can show that
\begin{align*}
	&R_\text{med}^2 \var(Y) \\
	&= \var(\beta_\text{med} X + \gamma_{1,\text{med}}' W_1) \\[0.5em]
	&= \beta_\text{med}^2 \var(X) + \var(\gamma_{1,\text{med}}' W_1) + 2 \beta_\text{med}  \cov(X,\gamma_{1,\text{med}}'W_1) \\[0.5em]
	&= (\beta_\text{long} + \bias)^2 \var(X) +  \var((\gamma_{1,\text{long}} - \bias \pi_1)'W_1) + 2 (\beta_\text{long} + \bias) \cov(X,(\gamma_{1,\text{long}} - \bias \pi_1)'W_1)\\[0.5em]
	&= (\beta_\text{long} + \bias)^2 \var(X) 
	+ \var(\gamma_{1,\text{long}}' W_1) + \bias^2 \big( \var(X) - \var(X^{\perp W_1}) \big)\\
	&\qquad - 2 \bias \cov(X, \gamma_{1,\text{long}}'W_1)  
	+ 2 (\beta_\text{long} + \bias) \Big[ \cov(X, \gamma_{1,\text{long}}'W_1) + \bias \big( \var(X^{\perp W_1}) - \var(X) \big) \Big] \\[0.5em]
	&= \beta_\text{long}^2 \var(X) + \bias^2 \var(X) + 2 \beta_\text{long} \bias \var(X)
	+ \var(\gamma_{1,\text{long}}' W_1) + \bias^2 \big( \var(X) - \var(X^{\perp W_1}) \big)\\[0.5em]
	&\qquad - 2 \bias \cov(X, \gamma_{1,\text{long}}'W_1) + 2 (\beta_\text{long} + \bias) \Big[ \cov(X, \gamma_{1,\text{long}}'W_1) + \bias \big( \var(X^{\perp W_1}) - \var(X) \big) \Big]\\[0.5em]
	&= \beta_\text{long}^2 \var(X) + \var(\gamma_{1,\text{long}}'W_1) + (2 \beta_\text{long} \bias + \bias^2) \var(X^{\perp W_1}) + 2 \beta_\text{long} \cov(X, \gamma_{1,\text{long}}'W_1).
\end{align*}
The third equality follows from equation \eqref{eq:lemma1-1}. Subtracting this from the R-squared long equation gives
\begin{align*}
	&(R_\text{long}^2 - R_\text{med}^2) \var(Y)\\[0.5em]
	&= \beta_\text{long}^2 \var(X) + \var(\gamma_{1,\text{long}}'W_1) + \var(\gamma_{2,\text{long}} W_2) + 2 \beta_\text{long} \cov(X, \gamma_{1,\text{long}}'W_1) + 2 \beta_\text{long} \bias \var(X^{\perp W_1}) \\[0.5em]
	&\qquad - \beta_\text{long}^2 \var(X) - \var(\gamma_{1,\text{long}}'W_1) - (2 \beta_\text{long} \bias + \bias^2) \var(X^{\perp W_1}) - 2 \beta_\text{long} \cov(X, \gamma_{1,\text{long}}'W_1) \\[0.5em]
	&= \var(\gamma_{2,\text{long}} W_2) - \bias^2 \var(X^{\perp W_1}) .
\end{align*}
Then, multiply both sides of this equation by $\delta \cov(X, \gamma_{1,\text{long}}'W_1)$:
\begin{align*}
	&(R_\text{long}^2 - R_\text{med}^2) \var(Y)\delta \cov(X, \gamma_{1,\text{long}}'W_1) \\
	&= \delta \cov(X, \gamma_{1,\text{long}}'W_1)\var(\gamma_{2,\text{long}} W_2) - \bias^2 \var(X^{\perp W_1})\delta \cov(X, \gamma_{1,\text{long}}'W_1)\\[0.5em]
	&=  \cov(X, \gamma_{2,\text{long}} W_2) \var(\gamma_{1,\text{long}}'W_1)  - \bias^2 \var(X^{\perp W_1})\delta \cov(X, \gamma_{1,\text{long}}'W_1)\\
	&=  \pi_2\gamma_{2,\text{long}}\var(W_2) \var(\gamma_{1,\text{long}}'W_1)  - \bias^2 \var(X^{\perp W_1})\delta \cov(X, \gamma_{1,\text{long}}'W_1)\\
	&=  \bias \var(X^{\perp W_1}) \var(\gamma_{1,\text{long}}'W_1)  - \bias^2 \var(X^{\perp W_1})\delta \cov(X, \gamma_{1,\text{long}}'W_1).
\end{align*}
The second equality follows from \ref{assn:delta}, the third from \ref{assn:exog}, and the fourth from equation \eqref{eq:lemma1-2}.
\end{proof}

\begin{proof}[Proof of Theorem \ref{thm:IDset}]
We split this proof into two parts.

\medskip

\emph{Part 1}: First we show that $\beta_\text{long} \in \mathcal{B}_I(\delta,R_\text{long}^2)$. To prove this, it suffices to show that $\bias = \beta_\text{med} - \beta_\text{long}$ satisfies $f(\bias,\delta,R_\text{long}^2) = 0$ and that $\beta_\text{long} \notin \tildeB$. By Lemma \ref{lemma:osterEq3} we can write
\begin{align*}
	(R_\text{long}^2 - R_\text{med}^2) \var(Y)\delta \cov(X, \gamma_{1,\text{long}}'W_1) &= \bias \var(X^{\perp W_1}) \left( \var(\gamma_{1,\text{long}}'W_1)  - \bias \delta \cov(X, \gamma_{1,\text{long}}'W_1) \right).
\end{align*}
Substituting $\gamma_{1,\text{long}} = \gamma_{1,\text{med}} + \bias \pi_1$, which was shown in Lemma \ref{lemma:regfacts}, we obtain the equation
\begin{align*}
	&\delta (R_\text{long}^2 - R_\text{med}^2)\var(Y)  \left(\cov(X,\gamma_{1,\text{med}}'W_1) + \bias\var(\pi_1'W_1)\right) \\
	&= \bias \var(X^{\perp W_1})\Big[\var(\gamma_{1,\text{med}}'W_1) + 2\bias \cov(X,\gamma_{1,\text{med}}'W_1) \\
	&\qquad + \bias^2 \var(\pi_1'W_1) - \bias \delta\left(\cov(X,\gamma_{1,\text{med}}'W_1) + \bias\var(\pi_1'W_1)\right)\Big]
\end{align*}
where we used that $\cov(X,\pi_1'W_1) = \var(\pi_1'W_1)$, which follows from equation \eqref{eq:proj2} and Assumption \ref{assn:exog}.

Collecting terms, this equation is equivalent to $f(\bias,\delta,R_\text{long}^2) = 0$  as defined above. Recall that $\gamma_{1,\text{long}} = \gamma_{1,\text{med}} + \bias \pi_1$. By \ref{assn:nonzero}, $\gamma_{1,\text{long}} \neq 0$ and thus $\beta_\text{long} = \beta_\text{med} - \bias \notin \tildeB$.
Therefore, $\beta_\text{long} \in \mathcal{B}_I(\delta,R_\text{long}^2)$.

\bigskip

\emph{Part 2}: Next we show that $\mathcal{B}_I(\delta,R_\text{long}^2)$ is sharp. Let $b \in \mathcal{B}_I(\delta,R_\text{long}^2)$ and $\bias = \beta_\text{med} - b$. This implies that $f(\bias,\delta,R_\text{long}^2) = 0$ and $b \notin \tildeB$. Fix the distribution of the observables $(Y,X,W_1)$ and assume that $\var(Y,X,W_1)$ is positive definite.

Let
\begin{align*}
	g_1 &\coloneqq \gamma_{1,\text{med}} + \bias \pi_1\\
	g_2 &\coloneqq \left((R_\text{long}^2 - R_\text{med}^2)\var(Y) + \bias^2 \var(X^{\perp W_1})\right)^{1/2}\\
	p_1 &\coloneqq \var(W_1)^{-1}\cov(W_1,X)\\
	p_2 &\coloneqq \frac{\bias \var(X^{\perp W_1})}{\left((R_\text{long}^2 - R_\text{med}^2)\var(Y) + \bias^2 \var(X^{\perp W_1})\right)^{1/2}}.
\end{align*}
Note that $g_2$ is well defined and strictly positive by $R_\text{long}^2 > R_\text{med}^2$ and $\var(Y) > 0$, which follow from \ref{assn:posdef}. $p_2$ is well defined by $g_2 > 0$. $g_1 \neq 0$ by $b \notin \tildeB$. Also note that $p_1 = \pi_1$ by our assumptions.

Under these choices, we have
\begin{align}
	b \var(X) + g_1'\var(W_1)p_1 + p_2g_2 &= (\beta_\text{med} - \bias)\var(X) +  (\gamma_{1,\text{med}} + \bias \pi_1)'\var(W_1)\pi_1 + \bias \var(X^{\perp W_1})\notag\\
	&= \beta_\text{med} \var(X) + \gamma_{1,\text{med}}'\var(W_1)\pi_1\notag\\
	&= \beta_\text{med} \var(X) + (\cov(Y,W_1) - \beta_\text{med} \cov(X,W_1))\var(W_1)^{-1}\var(W_1)\pi_1\notag\\
	&= \beta_\text{med} \var(X^{\perp W_1}) + \cov(Y,W_1)'\pi_1 \notag\\
	&= \cov(Y,X)\label{eq:covyx}
\end{align}
and
\begin{align}
	b p_1'\var(W_1) + g_1'\var(W_1) &= (\beta_\text{med} - \bias)\pi_1'\var(W_1) + (\gamma_{1,\text{med}} + \bias \pi_1)'\var(W_1)\notag\\
	&= \beta_\text{med} \pi_1'\var(W_1) + (\cov(Y,W_1) - \beta_\text{med} \cov(X,W_1))\var(W_1)^{-1}\var(W_1)\notag\\
	&= \cov(Y,W_1).\label{eq:covyw1}
\end{align}
Our goal is to construct a joint distribution for $(Y,X,W_1,W_2)$ that is consistent with the observed distribution of $(Y,X,W_1)$, that is such that $\var(Y,X,W_1,W_2)$ is finite, $\var(X,W_1,W_2)$ is positive definite, \ref{assn:delta}--\ref{assn:exog} hold, that is consistent with the value of $R_\text{long}^2$, and such that projection equations \eqref{eq:proj1} and \eqref{eq:proj2} hold.

To help define the joint distribution of $(Y,X,W_1,W_2)$, let
\begin{align*}
	\Sigma &\coloneqq \begin{pmatrix}
		\var(Y) & \cov(Y,X) & \cov(Y,W_1) & b p_2 + g_2\\
		\cov(Y,X) & \var(X) &\cov(X,W_1) & p_2\\
		\cov(W_1,Y) & \cov(W_1,X) & \var(W_1) & 0\\
		b p_2 + g_2 & p_2 & 0 & 1
	\end{pmatrix}.
\end{align*}
The top left $(d_1 + 2) \times (d_1 + 2)$ block is positive definite by assumption. By Sylvester's criterion, we can then verify that this matrix is positive semidefinite by showing its determinant is nonnegative. Before doing so, we first verify that $\var(X^{\perp W_1}) > p_2^2$:
\begin{align*}
	\var(X^{\perp W_1}) - p_2^2
	&= \var(X^{\perp W_1}) - \frac{\bias^2 \var(X^{\perp W_1})^2}{(R_\text{long}^2 - R_\text{med}^2)\var(Y) + \bias^2 \var(X^{\perp W_1})}\\
	&= \var(X^{\perp W_1})\frac{(R_\text{long}^2 - R_\text{med}^2)\var(Y)}{(R_\text{long}^2 - R_\text{med}^2)\var(Y) + \bias^2 \var(X^{\perp W_1})}\\
	& > 0.
\end{align*}
The inequality follows from the positive-definiteness of $\var(Y,X,W_1)$, $R_\text{long}^2 > R_\text{med}^2$, and by $\var(Y)>0$.

This inequality allows us to show that the matrix 
\begin{align*}
\Sigma' \coloneqq \begin{pmatrix}
		\var(X) & \cov(X,W_1) & p_2\\
		\cov(W_1,X) & \var(W_1) & 0\\
		p_2 & 0 & 1
	\end{pmatrix}
\end{align*}
is invertible via the partitioned inverse matrix formula:
\begin{align*}
& \begin{pmatrix}
		\var(X) & \cov(X,W_1) & p_2\\
		\cov(W_1,X) & \var(W_1) & 0\\
		p_2 & 0 & 1
	\end{pmatrix}^{-1}\\
	&= \left(\var(X) - \cov(X,W_1)\var(W_1)^{-1}\cov(W_1,X) - p_2^2\right)^{-1}\\
	&\cdot \begin{pmatrix}
		1  & - \begin{pmatrix}
		p_1' & p_2
		\end{pmatrix}\\
		\begin{pmatrix}
		p_1 \\ p_2
		\end{pmatrix} & \left(\var(X) - \cov(X,W_1)\var(W_1)^{-1}\cov(W_1,X) - p_2^2\right)\begin{pmatrix}
		\var(W_1)^{-1} & 0\\
		0 & 1
		\end{pmatrix} + \begin{pmatrix}
		p_1 \\ p_2
		\end{pmatrix}
	\begin{pmatrix}
		p_1' & p_2
		\end{pmatrix}\end{pmatrix}\\
		&= \left(\var(X^{\perp W_1}) - p_2^2\right)^{-1} \begin{pmatrix}
		 1 & -p_1' & -p_2 \\
		-p_1 & \var(W_1)^{-1}\left(\var(X^{\perp W_1}) - p_2^2\right) + p_1p_1' &  p_1p_2\\
		-p_2 & p_2 p_1' & \var(X^{\perp W_1})
		\end{pmatrix}.
\end{align*}
Since $\Sigma'$ is invertible, we can use the partitioned determinant formula to compute the determinant of $\Sigma$.
\begin{align*}
	&\det(\Sigma) \\
	&= \det\begin{pmatrix}
		\var(Y) & \cov(Y,X) & \cov(Y,W_1) & b p_2 + g_2\\
		\cov(Y,X) & \var(X) &\cov(X,W_1) & p_2\\
		\cov(W_1,Y) & \cov(W_1,X) & \var(W_1) & 0\\
		b p_2 + g_2 & p_2 & 0 & 1
	\end{pmatrix}\\
	&= 1 \cdot \det\left(\var(Y) - \begin{pmatrix}
	\cov(Y,X) & \cov(Y,W_1) & b p_2 + g_2
	\end{pmatrix} \begin{pmatrix}
		\var(X) & \cov(X,W_1) & p_2\\
		\cov(W_1,X) & \var(W_1) & 0\\
		p_2 & 0 & 1
	\end{pmatrix}^{-1}\begin{pmatrix}
	\cov(Y,X) \\ \cov(W_1,Y) \\ b p_2 + g_2
	\end{pmatrix}\right)\\
	&= \var(Y) - \begin{pmatrix}
	\cov(Y,X) & \cov(Y,W_1) & b p_2 + g_2
	\end{pmatrix}  \left(\var(X^{\perp W_1}) - p_2^2\right)^{-1} \\
	&\cdot \begin{pmatrix}
		 1 & -p_1' & -p_2 \\
		-p_1 & \var(W_1)^{-1}\left(\var(X^{\perp W_1}) - p_2^2\right) + p_1p_1' &  p_1p_2\\
		-p_2 & p_2 p_1' & \var(X^{\perp W_1})
		\end{pmatrix}\begin{pmatrix}
	\cov(Y,X) \\ \cov(W_1,Y) \\ b p_2 + g_2
	\end{pmatrix}.
\end{align*}
Using equations \eqref{eq:covyx}--\eqref{eq:covyw1}, we can write
\begin{align*}
	&\begin{pmatrix}
		 1 & -p_1' & -p_2 \\
		-p_1 & \var(W_1)^{-1}\left(\var(X^{\perp W_1}) - p_2^2\right) + p_1p_1' &  p_1p_2\\
		-p_2 & p_2 p_1' & \var(X^{\perp W_1})
		\end{pmatrix}
	\begin{pmatrix}
		\cov(Y,X) \\ \cov(W_1,Y) \\ b p_2 + g_2
	\end{pmatrix}\\[0.5em]
	&= \begin{pmatrix}
		 1 & -p_1' & -p_2 \\
		-p_1 & \var(W_1)^{-1}\left(\var(X^{\perp W_1}) - p_2^2\right) + p_1p_1' &  p_1p_2\\
		-p_2 & p_2 p_1' & \var(X^{\perp W_1})
		\end{pmatrix}
		\begin{pmatrix}
		b \var(X) + g_1'\var(W_1)p_1 + p_2g_2 \\ 
		b \var(W_1) p_1 + \var(W_1) g_1 \\ 
		b p_2 + g_2
		\end{pmatrix}\\[0.5em]
	&= \begin{pmatrix}
		b \var(X) + g_1'\var(W_1)p_1 + p_2 g_2 - b \var(p_1'W_1) - p_1'\var(W_1)g_1 - bp_2^2 - g_2p_2\\[1em]
		-p_1(\beta_\text{long} \var(X) + g_1'\var(W_1)p_1 + p_2g_2) + (\var(X^{\perp W_1}) - p_2^2)(b p_1 + g_1) \\
		+ bp_1 \var(p_1'W_1) + p_1 p_1'\var(W_1)g_1 + p_1p_2 (b p_2 + g_2)\\[1em]
		-p_2(b \var(X) + g_1'\var(W_1)p_1 + p_2g_2) + p_2 b \var(p_1'W_1) + p_2 p_1'\var(W_1) g_1 + \var(X^{\perp W_1})(b p_2 + g_2)
	\end{pmatrix}\\[0.5em]
	&= \begin{pmatrix}
		b (\var(X^{\perp W_1}) - p_2^2) \\
		 (\var(X^{\perp W_1}) - p_2^2) g_1 \\
		-p_2(p_2g_2)   + \var(X^{\perp W_1})(g_2)
	\end{pmatrix}\\[0.5em]
	&= (\var(X^{\perp W_1}) - p_2^2) \begin{pmatrix}
		b\\
		g_1\\
		g_2
	\end{pmatrix}.
\end{align*}
Therefore,
\begin{align*}
	&\var(Y) - \begin{pmatrix}
	\cov(Y,X) & \cov(Y,W_1) & b p_2 + g_2
	\end{pmatrix}  \left(\var(X^{\perp W_1}) - p_2^2\right)^{-1} \\
	&\cdot \begin{pmatrix}
		 1 & -p_1' & -p_2 \\
		-p_1 & \var(W_1)^{-1}\left(\var(X^{\perp W_1}) - p_2^2\right) + p_1p_1' &  p_1p_2\\
		-p_2 & p_2 p_1' & \var(X^{\perp W_1})
		\end{pmatrix}\begin{pmatrix}
	\cov(Y,X) \\ \cov(Y,W_1) \\ b p_2 + g_2
	\end{pmatrix}\\
	&\qquad = \var(Y) - \begin{pmatrix}
	b \var(X) + g_1'\var(W_1)p_1 + p_2g_2 & b \var(W_1)p_1 + \var(W_1)g_1 & b p_2 + g_2
	\end{pmatrix} \begin{pmatrix}
		b\\
		g_1\\
		g_2
	\end{pmatrix}\\
	&\qquad = \var(Y) - b^2 \var(X) - 2b g_1'\var(W_1)p_1 - 2bp_2g_2 - \var(g_1'W_1) - g_2^2.
\end{align*}
We now substitute the choices for $(b,g_1,p_1,p_2,g_2)$ and obtain
\begin{align*}
	&\det(\Sigma) \\[0.5em]
	&= \var(Y) - (\beta_\text{med} - \bias)^2 \var(X) - 2(\beta_\text{med} - \bias) (\gamma_{1,\text{med}} + \bias \pi_1)'\var(W_1)\pi_1 - 2(\beta_\text{med} - \bias)\bias \var(X^{\perp W_1})\\
	&\qquad - \var((\gamma_{1,\text{med}} + \bias \pi_1)'W_1) - (R_\text{long}^2 - R_\text{med}^2)\var(Y) - \bias^2 \var(X^{\perp W_1})\\[0.5em]
	&= \var(Y)(1 - R_\text{long}^2) - \beta_\text{med}^2 \var(X) + 2\beta_\text{med} \bias \var(X) - \bias^2 \var(X) -2\beta_\text{med} \gamma_{1,\text{med}}'\var(W_1)\pi_1 \\
	&\qquad -2\beta_\text{med} \bias \var(\pi_1'W_1) + 2\bias \gamma_{1,\text{med}}'\var(W_1)\pi_1 + 2\bias^2 \var(W_1'\pi_1) -2\beta_\text{med} \bias \var(X^{\perp W_1}) + 2\bias^2 \var(X^{\perp W_1}) \\
	&\qquad- \var(\gamma_{1,\text{med}}'W_1) - 2\bias \pi_1'\var(W_1)\gamma_{1,\text{med}} - \bias^2 \var(\pi_1'W_1) + R_\text{med}^2 \var(Y) - \bias^2 \var(X^{\perp W_1})\\[0.5em]
	&= \var(Y)(1 - R_\text{long}^2) - \Big( \beta_\text{med}^2 \var(X) + 2\beta_\text{med} \gamma_{1,\text{med}}'\var(W_1)\pi_1 + \var(\gamma_{1,\text{med}}'W_1)  \Big) + R_\text{med}^2 \var(Y) \\[0.5em]
	&= \var(Y)(1 - R_\text{long}^2) - R_\text{med}^2 \var(Y)  + R_\text{med}^2 \var(Y) \\[0.5em]
	&= \var(Y)(1-R_\text{long}^2)\\[0.5em]
	&\geq 0
\end{align*}
by $R_\text{long}^2 \leq 1$. From this, we conclude that $\Sigma$ is positive semidefinite. 

To construct $W_2$, let 
\begin{align*}
	v &\coloneqq 1 - \begin{pmatrix}
		bp_2 + g_2 & p_2 & 0
	\end{pmatrix}\var(Y,X,W_1)^{-1}\begin{pmatrix}
	bp_2 + g_2\\
	p_2\\
	0
	\end{pmatrix}.
\end{align*}
Since $\Sigma$ is positive semidefinite and from $\var(Y,X,W_1)$ being positive definite, we have that $v \geq 0$. Let $\widetilde{W}_2$ denote a unit-variance random variable that is uncorrelated with $(Y,X,W_1)$ and let
\begin{align*}
	W_2 &\coloneqq \begin{pmatrix}
		bp_2 + g_2 & p_2 & 0
	\end{pmatrix}\var(Y,X,W_1)^{-1}\begin{pmatrix}
	Y\\
	X\\
	W_1
	\end{pmatrix} + v^{1/2} \widetilde{W}_2.
\end{align*}

We can verify that $\var(Y,X,W_1,W_2) = \Sigma$. Note that we arbitrarily selected $1$ as the value of $\var(W_2)$. Thus we have constructed a joint distribution for $(Y,X,W_1,W_2)$ that is consistent with the observed distribution of $(Y,X,W_1)$ and has a finite variance matrix equal to $\Sigma$. $\var(X,W_1,W_2) = \Sigma'$ is positive definite since it is a positive semidefinite matrix satisfying $\det(\Sigma') > 0$. Therefore, Assumption \ref{assn:posdef} holds.

We now show this joint distribution and $(b,g_1,p_2,g_2)$ are consistent with Assumption \ref{assn:nonzero}. In what follows, let $(\beta,\gamma_{1,\text{long}},\gamma_{2,\text{long}},\pi_1,\pi_2) = (b,g_1,g_2,p_1,p_2)$. Then $\gamma_{2,\text{long}} \neq 0$ since $g_2$ was shown to be strictly positive. Similarly, $\gamma_{1,\text{long}} = \gamma_{1,\text{med}} + B\pi_1 \neq 0$ by $\beta_\text{med} - \bias \notin \tildeB$. Therefore, \ref{assn:nonzero} is satisfied.

Assumption \ref{assn:exog} holds since the value of $\cov(W_1,W_2)$ implied by matrix $\Sigma$ is $0$. 

We now show this joint distribution and $(\beta,\gamma_{1,\text{long}},\gamma_{2,\text{long}},\pi_1,\pi_2)$ are consistent with the definition of $R_\text{long}^2$ and Assumption \ref{assn:delta}.

For $R_\text{long}^2$, we have that
\begin{align*}
	&\beta^2 \var(X) + \var(\gamma_{1,\text{long}}' W_1) + \var(\gamma_{2,\text{long}}W_2) + 2\beta \cov( X, \gamma_{1,\text{long}}'W_1) \\
	&\qquad + 2 \beta\cov(X, \gamma_{2,\text{long}}W_2) + 2 \cov(\gamma_{1,\text{long}}'W_1, \gamma_{2,\text{long}}W_2)\\
	&= \beta^2 \var(X) + 2\beta \gamma_{1,\text{long}}'\var(W_1)\pi_1 + 2\beta \pi_2\gamma_{2,\text{long}} + \var(\gamma_{1,\text{long}}'W_1) + \gamma_{2,\text{long}}^2\\ 
	&= R_\text{long}^2 \var(Y)
\end{align*}
as shown above in the proof of positive semidefiniteness of $\Sigma$.

For $\delta$, using \ref{assn:nonzero} we can write
\begin{align*}
	&\var(\gamma_{1,\text{long}}'W_1)\var(\gamma_{2,\text{long}}W_2)
	\left(\frac{ \cov(X, \gamma_{2,\text{long}} W_2) }{ \var(\gamma_{2,\text{long}}W_2) } - \delta \cdot \frac{\cov(X,\gamma_{1,\text{long}}'W_1) }{\var(\gamma_{1,\text{long}}'W_1)}\right)\\
	&= \pi_2 \gamma_{2,\text{long}}\var(\gamma_{1,\text{long}}'W_1) - \delta \cov(X,\gamma_{1,\text{long}}'W_1) \gamma_{2,\text{long}}^2  \\
	&= \bias \var(X^{\perp W_1})\var((\gamma_{1,\text{med}} + \bias \pi_1)'W_1) - \delta \cov(X,(\gamma_{1,\text{med}} + \bias \pi_1)'W_1)( \gamma_{2,\text{long}}^2) \\
	&= \bias \var(X^{\perp W_1}) \left(\var(\gamma_{1,\text{med}}'W_1) + 2 \bias \cov(X,\gamma_{1,\text{med}}'W_1) + \bias^2 \var(\pi_1'W_1)\right) \\
	&\qquad - \delta ((R_\text{long}^2 - R_\text{med}^2)\var(Y) + \bias^2 \var(X^{\perp W_1}))(\cov(X,\gamma_{1,\text{med}}'W_1) + \bias \var(\pi_1'W_1))  \\
	&= -f_0(B) - \delta f_1(B,R^2_\text{long}).
\end{align*}
By $f(\bias,\delta,R_\text{long}^2) = f_0(\bias) + \delta f_1(\bias,R_\text{long}^2) = 0$ we have that 
\[
	\delta \cdot \frac{\cov(X,\gamma_{1,\text{long}}'W_1) }{\var(\gamma_{1,\text{long}}'W_1)} = \frac{ \cov(X, \gamma_{2,\text{long}} W_2) }{ \var(\gamma_{2,\text{long}}W_2) }
\]
and \ref{assn:delta} holds.

To finish the sharpness proof, we verify that this joint distribution is consistent with equations \eqref{eq:proj1} and \eqref{eq:proj2}.

To verify equation \eqref{eq:proj1}, we show that
\begin{align*}
	\cov\left(Y - b X - g_1'W_1 - g_2 W_2, \begin{pmatrix}
		X\\
		W_1\\
		W_2
	\end{pmatrix}\right) = 0.
\end{align*}
To see this, note that
\begin{align*}
	\cov(Y - b X - g_1' W_1 - g_2 W_2,X) &= \cov(Y,X) - b \var(X) - g_1'\var(W_1)p_1 - g_2 p_2\\
	&= 0\\
	\cov(Y - b X - g_1' W_1 - g_2 W_2,W_1) &= \cov(Y,W_1) - b\cov(X,W_1) - g_1' \var(W_1)\\
	&= 0
\end{align*}
from equations \eqref{eq:covyx} and \eqref{eq:covyw1} above. We can directly show that
\begin{align*}
	\cov(Y - b X - g_1' W_1 - g_2 W_2,W_2) &= b p_2 + g_2 - bp_2 - g_2 \var(W_2)\\
	&= 0
\end{align*}
since $\var(W_2) = 1$ by construction.

To verify equation \eqref{eq:proj2}, we see that 
\begin{align*}
	\cov\left(X - p_1'W_1 - p_2 W_2, \begin{pmatrix}
		W_1\\
		W_2
	\end{pmatrix}\right) &= \begin{pmatrix}
		\cov(X,W_1) - p_1'\var(W_1)\\
		\cov(X,W_2) - p_2 \var(W_2)
	\end{pmatrix}\\
	&= \begin{pmatrix}
		p_1'\var(W_1) - p_1'\var(W_1)\\
		p_2 - p_2
	\end{pmatrix}\\
	&= \begin{pmatrix}
	0\\
	0
	\end{pmatrix}.
\end{align*}
\end{proof}

\begin{proof}[Proof of Proposition \ref{prop:R2longEqualR2med}]
Let $A \coloneqq \Lin(Y \mid 1,X,W_1,W_2)$. Note that $R_\text{long}^2 \var(Y) = \var(A)$ and that $R^2_\text{med} \var(Y) = \var(\Lin(Y \mid 1,X,W_1))$. By iterated projections, $\Lin(A \mid 1,X,W_1) = \Lin(Y \mid 1,X,W_1)$. Therefore,
\begin{align*}
	&\var(A - \Lin(Y \mid 1,X,W_1)) \\
	&= \var(A) - 2 \cov(A,\Lin(Y \mid 1,X,W_1)) + \var(\Lin(Y \mid 1,X,W_1))\\
	&= \var(A) - 2 \cov(\Lin(A \mid 1,X,W_1) + A^{\perp X,W_1},\Lin(Y \mid 1,X,W_1)) + \var(\Lin(Y \mid 1,X,W_1))\\
	&= \var(A) - 2 \var(\Lin(Y \mid 1,X,W_1)) + \var(\Lin(Y \mid 1,X,W_1))\\
	&= \var(A) - \var(\Lin(Y \mid 1,X,W_1))\\
	&= (R_\text{long}^2 - R^2_\text{med})\var(Y).
\end{align*}
Therefore, $R^2_\text{med} = R_\text{long}^2$ implies that $A = \Lin(Y \mid 1,X,W_1)$ with probability 1. Rearranging this yields
\begin{align*}
	(\beta_\text{long} - \beta_\text{med}) X + (\gamma_{1,\text{long}} - \gamma_{1,\text{med}})' W_1 + \gamma_{2,\text{long}} W_2 &= 0
\end{align*}
with probability 1. By \ref{assn:posdef}, $(X,W_1,W_2)$ are not perfectly collinear and therefore $\beta_\text{long} = \beta_\text{med}$.
\end{proof}

\begin{proof}[Proof of Corollary \ref{cor:baselineCase}]
When $\delta = 0$,
\[
	f(B,\delta,R^2_\text{long}) = f_0(B) = - B\var(X^{\perp W_1}) \var((\gamma_{1,\text{med}} + B \pi_1)'W_1),
\]
a cubic polynomial. It has up to three roots because $\var(X^{\perp W_1}) > 0$ by \ref{assn:posdef}. $B = 0$ is a solution to the cubic equation $f_0(B) = 0$. Other solutions must satisfy $\var((\gamma_{1,\text{med}} + B \pi_1)'W_1) = 0$, or $\gamma_{1,\text{med}} + B \pi_1 = 0$ by $\var(W_1)$ being positive definite. Under assumptions \ref{assn:posdef}--\ref{assn:exog} and by Theorem \ref{thm:IDset}, $\beta_\text{long} \notin \tildeB$ so these other solutions do not correspond to elements of $\mathcal{B}_I(0,R^2_\text{long})$. The solution $B = 0$, corresponding to $\beta_\text{long} = \beta_\text{med}$, is part of the identified set provided that $\beta_\text{med} \notin  \tildeB$, or that $\gamma_{1,\text{med}} + (\beta_\text{med} - \beta_\text{med})\pi_1 \neq 0$. This is directly assumed in the statement of the corollary, so $\mathcal{B}_I(0,R^2_\text{long}) = \{\beta_\text{med}\}$.
\end{proof}

\begin{proof}[Proof of Theorem \ref{thm:OsterProp3}]
If $\beta_\text{long} = \beta_\text{hypo}$, then $\beta_\text{hypo} \in \mathcal{B}_I(\delta,R^2_\text{long})$. Then we must have
\[
	f_0(\beta_\text{med} - \beta_\text{hypo}) + \delta f_1(\beta_\text{med} - \beta_\text{hypo},R^2_\text{long}) = 0.
\]
Because of this, $f_1(\beta_\text{med} - \beta_\text{hypo},R^2_\text{long}) =  0$ only if $f_0(\beta_\text{med} - \beta_\text{hypo}) = 0$. As shown in the proof of Corollary \ref{cor:baselineCase}, this is the case only when $\beta_\text{hypo} = \beta_\text{med}$, in which case
\[
	f_1(0,R^2_\text{long}) = (R^2_\text{long} - R^2_\text{med})\var(Y) \cov(X,\gamma_{1,\text{med}}'W_1) \neq 0
\]
since $\beta_\text{short} \neq \beta_\text{med}$. Therefore, $f_1(\beta_\text{med} - \beta_\text{hypo},R^2_\text{long}) \neq 0$, and
\[
	\delta = \frac{-f_0(\beta_\text{med} - \beta_\text{hypo})}{f_1(\beta_\text{med} - \beta_\text{hypo}, R_\text{long}^2)}
\]
is identified.

Finally, by definition, $\{\delta \in \R: \beta_\text{hypo} \in \mathcal{B}_I(\delta, R_\text{long}^2)\}$ is a singleton equal to
\[
	\frac{-f_0(\beta_\text{med} - \beta_\text{hypo})}{f_1(\beta_\text{med} - \beta_\text{hypo}, R_\text{long}^2)}.
\]
Hence the exact value breakdown point equals the absolute value of this singleton.
\end{proof}

\subsection*{Section \ref{sec:OsterSignFlip}}

\begin{proof}[Proof of Theorem \ref{thm:osterExplainAwayBounded}]
We start this proof by noting that
\[
	\var(\pi_1'W_1) = \pi_1'\var(W_1)\pi_1 = \cov(X,W_1)\var(W_1)^{-1}\cov(W_1,X) > 0
\]because $\var(W_1)^{-1}$ is positive definite by \ref{assn:posdef}, and because $\cov(X,W_1) \neq 0$ by $\cov(X, \gamma_{1,\text{med}}'W_1) \neq 0$.

We also note that $\pi_1 \neq 0$ implies the set $\tildeB$ is a singleton or empty.

Let $B_n \uparrow +\infty$ and
\[
	\delta_n \coloneqq \frac{-f_0(B_n)}{f_1(B_n,R^2_\text{long})}.
\]
By construction, $\beta_\text{med} - B_n \in \mathcal{B}_I(\delta_n,R^2_\text{long}) \cup \tildeB$. Since $\tildeB$ contains at most one element, we let $\beta_\text{med} - B_n \in \mathcal{B}_I(\delta_n,R^2_\text{long})$ without loss of generality by removing at most one element from this sequence.

Then
\begin{align*}
	\delta_n 
	&= \Bigg( B_n^3 \var(\pi_1'W_1)\var(X^{\perp W_1}) + 2 B_n^2\var(X^{\perp W_1})\cov(X,\gamma_{1,\text{med}}'W_1) + B_n \var(X^{\perp W_1})\var(\gamma_{1,\text{med}}'W_1) \Bigg) \\
	&\qquad \Bigg/ \Bigg( B_n^3\var(X^{\perp W_1})\var(\pi_1'W_1) + B_n^2 \var(X^{\perp W_1})\cov(X,\gamma_{1,\text{med}}'W_1) \\
	&\qquad \qquad + B_n (R_\text{long}^2 - R^2_\text{med}) \var(Y) \var(\pi_1'W_1) + (R_\text{long}^2 - R^2_\text{med})\cov(X,W_1'\gamma_{1,\text{med}})\var(Y) \Bigg) \\	
	&\to 1
\end{align*}
as $n\to\infty$ since $\var(X^{\perp W_1})\var(\pi_1'W_1) > 0$ multiplies the cubic term in both the numerator and denominator. Note that $\var(X^{\perp W_1}) > 0$ by \ref{assn:posdef}. Similarly, we have that
\[
	\tilde\delta_n \coloneqq \frac{-f_0(-B_n)}{f_1(-B_n,R^2_\text{long})} \to 1
\]
as $n\to \infty$. Note that $\beta_\text{med} + B_n \in \mathcal{B}_I(\tilde\delta_n,R^2_\text{long})$ by construction and by the same argument used to show $\beta_\text{med} - B_n \notin \tildeB$.

Moreover, for $n$ sufficiently large we have that
\[
	\text{sign}(\delta_n - 1) = \text{sign}(\cov(X,\gamma_{1,\text{med}}'W_1)).
\]
We can see this from
\begin{align*}
	\delta_n - 1
	&= f_1(B_n,R^2_\text{long})^{-1}(-f_0(B_n) - f_1(B_n,R^2_\text{long}))\\
	&= f_1(B_n,R^2_\text{long})^{-1}\Big(B_n^3 \var(\pi_1'W_1)\var(X^{\perp W_1}) + 2 B_n^2\var(X^{\perp W_1})\cov(X,\gamma_{1,\text{med}}'W_1)  \\
	&\qquad\qquad\qquad\qquad + B_n \var(X^{\perp W_1})\var(\gamma_{1,\text{med}}'W_1) \Big)\\
	&= f_1(B_n,R^2_\text{long})^{-1} \Big(B_n^3\var(X^{\perp W_1})\var(\pi_1'W_1) + B_n^2 \var(X^{\perp W_1})\cov(X,\gamma_{1,\text{med}}'W_1) \\
	&\qquad\qquad\qquad\qquad + B_n (R_\text{long}^2 - R^2_\text{med}) \var(Y) \var(\pi_1'W_1) + (R_\text{long}^2 - R^2_\text{med})\cov(X,W_1'\gamma_{1,\text{med}})\var(Y) \Big)\\
	&= f_1(B_n,R^2_\text{long})^{-1}B_n^2 \left(\var(X^{\perp W_1})\cov(X,\gamma_{1,\text{med}}'W_1) + O(B_n^{-1})\right),
\end{align*}
and from $\var(X^{\perp W_1}) > 0$, $B_n^2 > 0$, $B_n^{-1} \to 0$, and $f_1(B_n,R^2_\text{long}) > 0$ as $n\to\infty$.

Likewise, for $n$ sufficiently large we have that
\[
	\text{sign}(\tilde\delta_n - 1) = -\text{sign}(\cov(X,\gamma_{1,\text{med}}'W_1))
\]
since
\[
	-(f_0(-B_n) + f_1(-B_n,R^2_\text{long})) = B_n^2 (\var(X^{\perp W_1})\cov(X,\gamma_{1,\text{med}}'W_1) + O(B_n^{-1})),
\]
and since $f_1(-B_n,R^2_\text{long}) < 0$ as $n\to\infty$.

Therefore, if $\cov(X,\gamma_{1,\text{med}}'W_1) > 0$, for sufficiently large $n$ the sequence $\delta_n$ is such that $\delta_n > 1$, $\delta_n \to 1$, and
\[
	\min \mathcal{B}_I(\delta_n,R^2_\text{long}) \leq \beta_\text{med} - B_n \to - \infty
\]
as $n\to\infty$. Also let $N$ be such that $\beta_\text{med} - B_n < 0$ and $\delta_n > 1$ for all $n \geq N$.

Likewise, for sufficiently large $n$ the sequence $\tilde\delta_n$ is such that $\tilde\delta_n < 1$, $\tilde\delta_n \to 1$, and
\[
	\max \mathcal{B}_I(\tilde\delta_n,R^2_\text{long}) \geq \beta_\text{med} + B_n \to +\infty
\]
as $n\to\infty$. Also let $\widetilde{N}$ be such that $\beta_\text{med} + B_n > 0$ and $\tilde\delta_n \leq 1$ for all $n \geq \widetilde{N}$.

If $\beta_\text{med} > 0$, then 
\begin{align*}
	\delta^{\text{bp,sign}}(R^2_\text{long}) &= \delta^{\text{bp},>}(R^2_\text{long})\\
	&= \inf \{ | \delta | : \delta \in \R, b \leq 0 \text{ for some } b \in \mathcal{B}_I(\delta, R_\text{long}^2) \}\\
	&\leq \inf\{|\delta_n|: n \geq N\}.
\end{align*}
The inequality follows from $\mathcal{B}_I(\delta_n,R^2_\text{long})$ containing a negative number for all $n \geq N$. Since $\delta_n \to 1$ and $\delta_n > 1$ for $n \geq N$, $\inf\{|\delta_n|: n \geq N\} = 1$ thus $\delta^{\text{bp,sign}}(R^2_\text{long}) \leq 1$.

If $\beta_\text{med} < 0$, then 
\begin{align*}
	\delta^{\text{bp,sign}}(R^2_\text{long}) &= \delta^{\text{bp},<}(R^2_\text{long})\\
	&= \inf \{ | \delta | : \delta \in \R, b \geq 0 \text{ for some } b \in \mathcal{B}_I(\delta, R_\text{long}^2) \}\\
	&\leq \inf\{|\tilde\delta_n|: n \geq \widetilde{N}\}\\
	& \leq 1.
\end{align*}
The first inequality follows from $\mathcal{B}(\delta_n,R^2_\text{long})$ containing a positive value for $n \geq \widetilde{N}$. The second inequality follows from $\tilde\delta_n \leq 1$ for all $n \geq \widetilde{N}$. An analogous argument shows that this conclusion holds if $\cov(X,\gamma_{1,\text{med}}'W_1) < 0$ as well.
\end{proof}

\begin{proof}[Proof of Corollary \ref{cor:magnitudeRestrictions}]
By Theorem \ref{thm:IDset}, $\beta_\text{long} \in \mathcal{B}_I(\delta,R^2_\text{long})$. By \ref{assn:magnRest}, $\beta_\text{long} \in [\beta_\text{med} - M, \beta_\text{med} + M]$. Combining these yields that $\beta_\text{long}$ lies in the intersection of these sets. Sharpness follows from the sharpness of $\mathcal{B}_I(\delta,R^2_\text{long})$, which was established in Theorem \ref{thm:IDset}, and from $\beta_\text{long} \in [\beta_\text{med} - M, \beta_\text{med} + M]$, which implies \ref{assn:magnRest}.
\end{proof}

\begin{proof}[Proof of Corollary \ref{cor:cum_IDset}]
	This follows directly from Theorem \ref{thm:IDset} and Corollary \ref{cor:magnitudeRestrictions}.
\end{proof}

\subsection*{Appendix \ref{sec:technical}}

\begin{proof}[Proof of Proposition \ref{prop:checkingOsterA2}]
\hfill
\begin{enumerate}
\item ($\Rightarrow$) As shown in Lemma \ref{lemma:regfacts}, $\gamma_{1,\text{long}} = \gamma_{1,\text{med}} + (\beta_\text{med} - \beta_\text{long})\pi_1$ under assumptions \ref{assn:posdef} and \ref{assn:exog}. Therefore, if \ref{assump:OsterA2} holds, then $\gamma_{1,\text{med}} = \gamma_{1,\text{long}} - (\beta_\text{med} - \beta_\text{long})\pi_1 = (C - (\beta_\text{med} - \beta_\text{long})) \pi_1 = C_\text{med} \cdot \pi_1$ with $C_\text{med} \coloneqq C - (\beta_\text{med} - \beta_\text{long})$. 

\item ($\Leftarrow$) If there is a constant $C_\text{med} \in \R$ such that $\gamma_{1,\text{med}} = C_\text{med} \cdot \pi_1$, then we immediately see that $\gamma_{1,\text{long}} = (C_\text{med} + (\beta_\text{med} - \beta_\text{long}))\pi_1$ so \ref{assump:OsterA2} holds with $C \coloneqq C_\text{med} + (\beta_\text{med} - \beta_\text{long})$.
\end{enumerate}
\end{proof}

\begin{proof}[Proof of Proposition \ref{prop:OsterP1}]
From the expression for $f(B,\delta,R^2_\text{long})$, we note that $f(B,1,R^2_\text{long})$ is quadratic in $B$ and is equal to
\begin{align*}
	&f(B,1,R^2_\text{long}) \\
	&= f_0(B) + f_1(B,R^2_\text{long})\\
	&= -B^2 \var(X^{\perp W_1})\cov(X,\gamma_{1,\text{med}}'W_1) + B \left((R^2_\text{long} - R^2_\text{med}) \var(Y) \var(\pi_1'W_1)- \var(X^{\perp W_1})\var(\gamma_{1,\text{med}}'W_1)\right)\\
	&\qquad + (R^2_\text{long} - R^2_\text{med})\cov(X,W_1'\gamma_{1,\text{med}})\var(Y).
\end{align*}
The discriminant of this quadratic polynomial is
\begin{align*}
	\Delta &\coloneqq \left((R^2_\text{long} - R^2_\text{med}) \var(Y) \var(\pi_1'W_1)- \var(X^{\perp W_1})\var(\gamma_{1,\text{med}}'W_1)\right)^2\\
	&\qquad+4 \var(X^{\perp W_1})\cov(X,\gamma_{1,\text{med}}'W_1)^2(R^2_\text{long} - R^2_\text{med})\var(Y)\\
	&= \left((R^2_\text{long} - R^2_\text{med}) \var(Y) \var(\pi_1'W_1)\right)^2 + \left(\var(X^{\perp W_1})\var(\gamma_{1,\text{med}}'W_1)\right)^2\\
	&\qquad - 2(R^2_\text{long} - R^2_\text{med}) \var(Y) \var(\pi_1'W_1)\var(X^{\perp W_1})\var(\gamma_{1,\text{med}}'W_1) \\
	&\qquad+ 4 \var(X^{\perp W_1})\cov(X,\gamma_{1,\text{med}}'W_1)^2(R^2_\text{long} - R^2_\text{med})\var(Y).
\end{align*}
To simplify this expression, we show that $\cov(X,\gamma_{1,\text{med}}'W_1)^2 = \var(\pi_1'W_1)\var(\gamma_{1,\text{med}}'W_1)$. To see this, note that \ref{assump:OsterA2} and Proposition \ref{prop:checkingOsterA2} implies that $\gamma_{1,\text{med}} = C_\text{med} \pi_1$. Therefore
\begin{align*}
	\cov(X,\gamma_{1,\text{med}}'W_1)^2 - \var(\pi_1'W_1)\var(\gamma_{1,\text{med}}'W_1)	&= C_\text{med}^2\cov(X,\pi_1'W_1)^2 - C_\text{med}^2\var(\pi_1'W_1)\var(\pi_1'W_1) \\
	&= C_\text{med}^2 \left(\cov(\pi_1'W_1 + \pi_2 W_2 + X^{\perp W},\pi_1'W_1)^2 - \var(\pi_1'W_1)^2\right)\\
	&= C_\text{med}^2 \left(\var(\pi_1'W_1)^2 - \var(\pi_1'W_1)^2\right)\\
	&= 0.
\end{align*}
The fourth equality follows from \ref{assn:exog}. Therefore, 
\begin{align}
	\cov(X,\gamma_{1,\text{med}}'W_1)^2 = \var(\pi_1'W_1)\var(\gamma_{1,\text{med}}'W_1) \label{eq:rel1}
\end{align}
which implies that the discriminant can be written as
\begin{align*}
	\Delta &= \left((R^2_\text{long} - R^2_\text{med}) \var(Y) \var(\pi_1'W_1)\right)^2 + \left(\var(X^{\perp W_1})\var(\gamma_{1,\text{med}}'W_1)\right)^2\\
	&\qquad + 2(R^2_\text{long} - R^2_\text{med}) \var(Y) \var(\pi_1'W_1)\var(X^{\perp W_1})\var(\gamma_{1,\text{med}}'W_1)\\
	&= \left((R^2_\text{long} - R^2_\text{med}) \var(Y) \var(\pi_1'W_1) + \var(X^{\perp W_1})\var(\gamma_{1,\text{med}}'W_1)\right)^2.
\end{align*}
Therefore, the quadratic equation $f(B,1,R^2_\text{long}) = 0$ has the following two solutions
\begin{align*}
	&\Bigg( (R^2_\text{long} - R^2_\text{med}) \var(Y) \var(\pi_1'W_1) -\var(X^{\perp W_1})\var(\gamma_{1,\text{med}}'W_1) \\
	&\qquad \pm \sqrt{\left((R^2_\text{long} - R^2_\text{med}) \var(Y) \var(\pi_1'W_1) + \var(X^{\perp W_1})\var(\gamma_{1,\text{med}}'W_1)\right)^2} \Bigg)
	\Bigg/ 2 \var(X^{\perp W_1})\cov(X,\gamma_{1,\text{med}}'W_1) \\[0.5em]
	&= \Bigg( (R^2_\text{long} - R^2_\text{med}) \var(Y) \var(\pi_1'W_1) -\var(X^{\perp W_1})\var(\gamma_{1,\text{med}}'W_1) \\
	&\qquad \pm \left((R^2_\text{long} - R^2_\text{med}) \var(Y) \var(\pi_1'W_1) + \var(X^{\perp W_1})\var(\gamma_{1,\text{med}}'W_1)\right) \Bigg)
	\Bigg/ 2 \var(X^{\perp W_1})\cov(X,\gamma_{1,\text{med}}'W_1) \\[0.5em]
	&= \{B_1,B_2\}
\end{align*}
where 
\[
	B_1 = \frac{(R^2_\text{long} - R^2_\text{med}) \var(Y) \var(\pi_1'W_1)}{\var(X^{\perp W_1})\cov(X,\gamma_{1,\text{med}}'W_1)}
	\qquad \text{and} \qquad
	B_2 = -\frac{\var(\gamma_{1,\text{med}}'W_1)}{\cov(X,\gamma_{1,\text{med}}'W_1)}.
\]
These are well defined by $\cov(X,\gamma_{1,\text{med}}'W_1) = \var(X)(\beta_\text{short} - \beta_\text{med}) \neq 0$. They are distinct by $\Delta > 0$. We now simplify the expressions for $B_1$ and $B_2$. Combining equations \eqref{eq:lemma1-3}, \eqref{eq:lemma1-4}, and \eqref{eq:rel1}, we can write
\begin{align*}
	\var(\gamma_{1,\text{med}}'W_1) &= (R^2_\text{med} - R^2_\text{short})\var(Y) + \var(X)(\beta_\text{short} - \beta_\text{med})^2\\
	&= (R^2_\text{med} - R^2_\text{short})\var(Y) + \frac{\cov(X,\gamma_{1,\text{med}}'W_1)^2}{\var(X)}\\
	&= (R^2_\text{med} - R^2_\text{short})\var(Y) +  \frac{\var(\pi_1'W_1)\var(\gamma_{1,\text{med}}'W_1)}{\var(X)}.
\end{align*}
Solving the previous equation for $\var(\gamma_{1,\text{med}}'W_1)$ yields
\begin{align*}
	\var(\gamma_{1,\text{med}}'W_1) &= (R^2_\text{med} - R^2_\text{short})\var(Y)\left(1 - \frac{\var(\pi_1'W_1)}{\var(X)}\right)^{-1}\\
	&= (R^2_\text{med} - R^2_\text{short})\var(Y)\frac{\var(X)}{\var(X) - \var(\pi_1'W_1)}\\
	&= (R^2_\text{med} - R^2_\text{short})\var(Y)\frac{\var(X)}{\var(X^{\perp W_1})}.
\end{align*}
This equality can be used to simplify $B_1$ as follows:
\begin{align*}
	B_1 &= \frac{(R^2_\text{long} - R^2_\text{med}) \var(Y) \var(\pi_1'W_1)}{\var(X^{\perp W_1})\cov(X,\gamma_{1,\text{med}}'W_1)}\\
	&= \frac{(R^2_\text{long} - R^2_\text{med}) \var(Y) \cov(X,\gamma_{1,\text{med}}'W_1)}{\var(X^{\perp W_1})\var(\gamma_{1,\text{med}}'W_1)}\\
	&= \frac{(R^2_\text{long} - R^2_\text{med}) \var(Y) (\beta_\text{short} - \beta_\text{med})\var(X)}{\var(X^{\perp W_1})(R^2_\text{med} - R^2_\text{short})\var(Y)\frac{\var(X)}{\var(X^{\perp W_1})}}\\
	&= (\beta_\text{short} - \beta_\text{med}) \frac{R^2_\text{long} - R^2_\text{med}}{R^2_\text{med} - R^2_\text{short}}.
\end{align*}
The second equality follows from equation \eqref{eq:rel1}, and the third follows from our expression for $\var(\gamma_{1,\text{med}}'W_1)$. 

Meanwhile, we can show that $\beta_\text{med} - B_2 \in \tildeB$ and thus $\beta_\text{med} - B_2 \notin \mathcal{B}_I(1,R^2_\text{long})$. To see this, note that
\begin{align*}
	&\var\left((\gamma_{1,\text{med}} + B_2 \pi_1)'W_1\right) \\
	&= \var\left(\gamma_{1,\text{med}}'W_1 -\frac{\var(\gamma_{1,\text{med}}'W_1)}{\cov(X,\gamma_{1,\text{med}}'W_1)} \pi_1'W_1\right)\\
	&= \var(\gamma_{1,\text{med}}'W_1) - 2\frac{\var(\gamma_{1,\text{med}}'W_1)}{\cov(X,\gamma_{1,\text{med}}'W_1)} \cov(\gamma_{1,\text{med}}'W_1,\pi_1'W_1) + \frac{\var(\gamma_{1,\text{med}}'W_1)^2}{\cov(X,\gamma_{1,\text{med}}'W_1)^2}\var(\pi_1'W_1)\\
	&= \var(\gamma_{1,\text{med}}'W_1) - 2\frac{\var(\gamma_{1,\text{med}}'W_1)}{\cov(X,\gamma_{1,\text{med}}'W_1)} \cov(\gamma_{1,\text{med}}'W_1,X) + \frac{\var(\gamma_{1,\text{med}}'W_1)^2}{\cov(X,\gamma_{1,\text{med}}'W_1)^2}\var(\pi_1'W_1)\\
	&= \var(\gamma_{1,\text{med}}'W_1) - 2 \var(\gamma_{1,\text{med}}'W_1) + \frac{\var(\gamma_{1,\text{med}}'W_1)^2\var(\pi_1'W_1)}{\var(\gamma_{1,\text{med}}'W_1)\var(\pi_1'W_1)}\\
	&= 0.
\end{align*}
The fourth equality follows from equation \eqref{eq:rel1} and from $\cov(X,\gamma_{1,\text{med}}'W_1) \neq 0$. Therefore, $\gamma_{1,\text{med}} + B_2 \pi_1 = 0$, which implies that $\beta_\text{med} - B_2 \notin \mathcal{B}_I(1,R^2_\text{long})$. Since $\pi_1 \neq 0$ and $\gamma_{1,\text{med}} \neq 0$, which follows from $\beta_\text{short} \neq \beta_\text{med}$, the set $\tildeB$ is a singleton. Since $B_1$ and $B_2$ are distinct, $\beta_\text{med} - B_1 \notin \tildeB = \{\beta_\text{med} - B_2\}$. By Theorem \ref{thm:IDset}, we conclude that $\beta_\text{med} - B_1$ is the unique element of the identified set.
\end{proof}

 \begin{proof}[Proof of Proposition \ref{prop:multicollinearOLS}]
 	First, consider the case where $|\beta_1| \neq \infty$ and $\beta_1 + \beta_2 = \sigma$. Let $c \coloneqq \sigma - 2\beta_1$, and define
 	\begin{align*}
 		(\sigma_{1,j},\sigma_{2,j},\rho_j) &= (\sigma,\sigma + c/j,1 - 1/j).
 	\end{align*}
 	Then,
 	\begin{align*}
 		\beta_{1,j} &= \frac{\sigma_{1,j} - \rho_j \sigma_{2,j}}{1-\rho_j^2}
 		= \frac{\sigma - (1 - 1/j) (\sigma + c/j)}{1-(1-1/j)^2}
 		= \frac{\sigma - c + c/j}{2 -1 /j}
 		\rightarrow \frac{\sigma - c}{2} = \beta_1
 	\end{align*}
 	and
 	\begin{align*}
 		\beta_{2,j}
		= \frac{\sigma_{2,j} - \rho_j \sigma_{1,j}}{1-\rho_j^2}
 		= \frac{\sigma + c/j - (1 - 1/j) \sigma}{1-(1-1/j)^2}
 		= \frac{\sigma + c}{2 -1/j}
 		\rightarrow \frac{\sigma + c}{2} = \sigma - \beta_1 = \beta_2.
 	\end{align*}
 	To show a sequence converging to $(-\infty,+\infty)$ let
 	\begin{align*}
 		(\sigma_{1,j},\sigma_{2,j},\rho_j) &= (\sigma,\sigma + 1/j,1 - 1/j^2). 
 	\end{align*}
 	Then,
 	\begin{align*}
 		\beta_{1,j}
		= \frac{\sigma - (1 - 1/j^2) (\sigma + 1/j)}{1-(1-1/j^2)^2}
 		= \frac{\sigma/j^2 - (1 - 1/j^2) 1/j}{2/j^2 - 1/j^4}
 		= \frac{\sigma}{2 + o(1)} - \frac{1-1/j^2}{2/j - 1/j^3}
 		\rightarrow -\infty
 	\end{align*}
 	and 
 	\begin{align*}
 		\beta_{2,j}
		= \frac{\sigma + 1/j - (1 - 1/j^2) \sigma}{1-(1-1/j^2)^2}
 		= \frac{1/j + \sigma/j^2}{2/j^2 - 1/j^4}
 		= \frac{j + \sigma}{2 - 1/j^2}
 		\rightarrow +\infty.
 	\end{align*}
 	Finally, for a sequence converging to $(+\infty,-\infty)$, let 
 	\begin{align*}
 		(\sigma_{1,j},\sigma_{2,j},\rho_j) &= (\sigma,\sigma - 1/j,1 - 1/j^2). 
 	\end{align*}
 \end{proof}

\nocitesurvey{Tabellini2020}
\nocitesurvey{DippelHeblich2021}
\nocitesurvey{Chen2021}
\nocitesurvey{GoldsmithSorkinSwitft2020}
\nocitesurvey{FinkelsteinEtAl2021}
\nocitesurvey{AbdulkadirougluEtAl2020}
\nocitesurvey{Squicciarini2020}
\nocitesurvey{VernerGrongyosi2020}
\nocitesurvey{Gregg2020}
\nocitesurvey{Couttenier2019}
\nocitesurvey{Arbatli2020}
\nocitesurvey{HoffmanTadelis2021}
\nocitesurvey{Enke2020}
\nocitesurvey{Eugster2019}
\nocitesurvey{Guriev2021}
\nocitesurvey{AbaluckEtAl2021}
\nocitesurvey{Weigel2020}
\nocitesurvey{AllcottEtAl2019}
\nocitesurvey{Valencia2019}
\nocitesurvey{AllcottEtAl2019sin}
\nocitesurvey{Enke2019}
\nocitesurvey{BlouinMacchiavello2019}
\nocitesurvey{Heldring2021}
\nocitesurvey{SchindlerWestcott2021}
\nocitesurvey{BladerEtAl2020}
\nocitesurvey{MurphyWeinhardt2020}
\nocitesurvey{BertrandEtAl2020}
\nocitesurvey{GrosfeldEtAl2020}
\nocitesurvey{GrosjeanKhattar2019}
\nocitesurvey{GavazzaEtAl2019}
\nocitesurvey{BFG2020}

\bibliographystylesurvey{econometrica}
\bibliographysurvey{OsterSignChange}

\bibliographystyle{econometrica}
\bibliography{OsterSignChange}

\makeatletter\@input{MP1aux.tex}\makeatother